\documentclass[10pt]{article}
\usepackage[margin=0.7in,footskip=0.25in]{geometry}
\date{}
\usepackage[shortlabels]{enumitem}

\usepackage{amsfonts,mathrsfs}
\usepackage{amsmath, amsthm, amssymb, dsfont,stmaryrd}
\usepackage{graphicx}
\usepackage{comment, cite}
\usepackage{psfrag}
\usepackage{bbm}

\usepackage[dvipsnames]{xcolor}
\usepackage{hyperref}
\hypersetup{
    colorlinks=true,
    linkcolor=blue,
    citecolor=blue,
    urlcolor=blue,
    allbordercolors={0 0 0},
    pdfborderstyle={/S/U/W 1}
}
\usepackage{bookmark}

\usepackage{tikz}

\usetikzlibrary{fit}					
\usetikzlibrary{backgrounds}	

\newtheorem{theorem}{Theorem}
\newtheorem{lemma}{Lemma}

\newtheorem{proposition}{Proposition}

\newtheorem{corollary}{Corollary}
\newtheorem{definition}{Definition}
\newtheorem{example}{Example}

\newtheorem{remark}{Remark}

\newtheorem{setting}{Setting}

\definecolor{darkblack}{rgb}{0, .07, .5}
\definecolor{darkred}{rgb}{0.5,0,0}
\definecolor{mahogany}{rgb}{0.65, 0., 0.5}

\newcommand{\Var}{\text{\rm{Var}}}

\newcommand{\E}{\mathbb{E}}

\newcommand{\ua}{\underline{a}}
\newcommand{\ub}{\underline{b}}
\newcommand{\mD}{\mathcal{D}}

\newcommand{\RN}[1]{%
	\textup{\uppercase\expandafter{\romannumeral#1}}%
}
\usepackage{authblk} 

\title{A Class of Subadditive Information Measures and their Applications}
\author[1]{Hamidreza Abin}
\author[1]{Mahdi Zinati}
\author[1]{Amin Gohari}
\author[2]{Mohammad Hossein Yassaee}
\author[2]{Mohammad Mahdi Mojahedian}

\affil[1]{Department of Information Engineering, The Chinese University of Hong Kong, Sha Tin, NT, Hong Kong \authorcr \texttt{\{hamabin, mahdiz, agohari\}@ie.cuhk.edu.hk}}
\affil[2]{Department of Electrical Engineering, Sharif University of Technology, Tehran, Iran
 \authorcr \texttt{\{yassaee, mojahedian\}@sharif.edu}}
\allowdisplaybreaks
\begin{document}
	\maketitle
  \begin{abstract}
We introduce a two-parameter family of discrepancy measures, termed \emph{$(G,f)$-divergences}, obtained by applying a non-decreasing function $G$ to an $f$-divergence $D_f$. Building on Csisz\'ar's formulation of mutual $f$-information, we define a corresponding $(G,f)$-information measure
$
I_{G,f}(X;Y)$. A central theme of the paper is subadditivity over product distributions and product channels. We develop reduction principles showing that, for broad classes of $G$, it suffices to verify divergence subadditivity on binary alphabets. Specializing  to the functions $G(x)\in\{x,\log(1+x),-\log(1-x)\}$, we derive tractable sufficient conditions on $f$ that guarantee subadditivity, covering many standard $f$-divergences. Finally, we present applications to finite-blocklength converses for channel coding, bounds in binary hypothesis testing, and an extension of the Shannon--Gallager--Berlekamp sphere-packing exponent framework to subadditive $(G,f)$-divergences.
\end{abstract}

\section{Introduction}
Divergence measures are a foundational tool in statistics, machine learning, and information theory for quantifying discrepancy between probability distributions. Among the most widely used families are \emph{$f$-divergences}, introduced by Ali--Silvey and Csisz\'ar~\cite{ali1966general,csiszar1967information,csiszar2011information}. Given distributions $p\ll q$ on a common alphabet, the $f$-divergence
\[
D_f(p\|q)=\sum_{x} q(x)\,f\!\left(\frac{p(x)}{q(x)}\right)
\]
is generated by a convex function $f$ with $f(1)=0$. This framework includes the Kullback--Leibler (KL) divergence, Hellinger-type divergences, and many others; it enjoys powerful structural properties such as convexity and the data processing inequality.

Many operationally meaningful measures, however, are not themselves $f$-divergences, but can be expressed as monotone transforms of an $f$-divergence. A canonical example is the relationship between R\'enyi divergence and Hellinger divergence: for $\alpha\in(0,1)\cup(1,\infty)$, the R\'enyi divergence can be written as $D_\alpha(P\|Q)=G(\mathscr{H}_\alpha(P\|Q))$ for a suitable logarithmic transform $G$, where $\mathscr{H}_\alpha$ is the Hellinger divergence of order $\alpha$. More broadly, logarithmic and power transforms of divergence-like quantities arise naturally in hypothesis testing, error exponents, and finite-blocklength analyses.

Motivated by these observations, we study divergences of the form
\[
\mathcal{D}_{G,f}(p\|q)\triangleq G\!\bigl(D_f(p\|q)\bigr),
\]
where $G$ is non-decreasing with $G(0)=0$ and $f$ is convex with $f(1)=0$. We refer to such $(G,f)$ as an \emph{admissible pair}. While taking a monotone function of an $f$-divergence has appeared in various contexts (e.g.,~\cite{polyanskiy2010arimoto}), a general framework connecting these transforms to information measures, subadditivity properties, and operational bounds remains incomplete. The present work aims to develop such a framework for finite alphabets.

\subsection*{$(G,f)$-information and structural properties}
Following Csisz\'ar's definition of mutual $f$-information~\cite{csiszar1972class}, we define the \emph{$(G,f)$-information measure}
\[
I_{G,f}(X;Y)\triangleq \min_{q_Y}\sum_{x} p_X(x)\,\mathcal{D}_{G,f}(p_{Y|X=x}\|q_Y)
= \min_{q_Y}\sum_{x} p_X(x)\,G\!\bigl(D_f(p_{Y|X=x}\|q_Y)\bigr).
\]
This reduces to Shannon mutual information under the classical choice $G(x)=x$ and $f(x)=x\log x$, but is generally not symmetric, i.e., $I_{G,f}(X;Y)\neq I_{G,f}(Y;X)$ in general. Nevertheless, $I_{G,f}$ retains several key ``information-like'' features: it is concave in the input distribution $p_X$ for a fixed channel $p_{Y|X}$ and satisfies data processing inequalities for Markov chains. We also establish an extension to the four-node Markov chain $A\to X\to Y\to B$ when $G$ is increasing and convex, which is convenient for coding-theoretic applications.

\subsection*{Subadditivity}
A central technical challenge in non-asymptotic and exponent analyses is to pass from single-letter quantities to blocklength-$n$ quantities. For KL divergence and Shannon mutual information, this is enabled by additivity over product measures/channels; for general divergences and generalized informations, exact additivity is typically unavailable, and \emph{subadditivity} becomes a key substitute. Indeed, in~\cite[Section~V]{polyanskiy2010arimoto}, the authors discuss using generalized divergence measures to obtain bounds on finite-blocklength channel coding, noting that (among the families they consider) only R\'enyi divergence is additive under product distributions and therefore leads directly to single-letter bounds. However, exact additivity is not necessary: subadditivity already suffices to obtain single-letter bounds, and it applies to a broader class of divergences beyond R\'enyi.

We study two notions of subadditivity. The first is \emph{divergence subadditivity}:
\[
\mathcal{D}_{G,f}(q_Yq_Z\|r_Yr_Z)\le \mathcal{D}_{G,f}(q_Y\|r_Y)+\mathcal{D}_{G,f}(q_Z\|r_Z),
\]
for arbitrary finite-alphabet distributions. The second is \emph{information subadditivity} over product channels:  
\[
I_{G,f}(X_1X_2;Y_1Y_2)\le I_{G,f}(X_1;Y_1)+I_{G,f}(X_2;Y_2)
\]
for arbitrary distributions of the form $p_{X_1,X_2}p_{Y_1|X_1}p_{Y_2|X_2}$.
Our main equivalence result shows that divergence subadditivity implies information subadditivity in full generality, and conversely, under mild differentiability and non-degeneracy assumptions. This equivalence is important because it allows one to verify a simpler inequality at the divergence level while obtaining subadditivity properties for the induced information measure.

\subsection*{Binary reductions and tractable sufficient conditions}
Even at the divergence level, verifying subadditivity for arbitrary finite alphabets can be challenging. We therefore develop reduction principles. For a broad class of increasing transforms $G$ satisfying a concavity condition on $x\mapsto G^{-1}(G(x)+G(c))$, we show that it suffices to verify divergence subadditivity on \emph{binary} alphabets, and we characterize this class through a simple condition on $1/G'(x)$. We also provide a complementary condition that can apply when the concavity reduction does not.

Specializing to the three transforms
\[
G(x)\in\Bigl\{x,\ \log(1+x),\ -\log(1-x)\Bigr\},
\]
we derive tractable sufficient conditions on $f$ that guarantee subadditivity. In particular, for $G(x)=x$ we give a simple second-derivative criterion---concavity of $x^2 f''(x)$---that covers many standard $f$-divergences.

\subsection*{Applications}
We illustrate how subadditive $(G,f)$-information and $(G,f)$-divergences lead to operational bounds. First, $I_{G,f}$ yields a finite-blocklength converse that lower-bounds the required blocklength $n$ in terms of $\max_{p_X} I_{G,f}(X;Y)$; for permutation-invariant channels the maximizer is the uniform input distribution. Second, in binary hypothesis testing, divergence subadditivity yields a single-letter bound in terms of $n\,\mathcal{D}_{G,f}(p\|q)$. Finally, we extend the Shannon--Gallager--Berlekamp sphere-packing methodology to a class of subadditive $(G,f)$-divergences, obtaining a generalized sphere-packing-style upper bound on the reliability function.

\subsection*{Related work}
The authors of \cite{cruz2025} define a notion of tensorization for an $f$-divergence as follows. They say that an $f$-divergence tensorizes if, for every $n \geq 1$, one can find a function $\tau_n$ such that
\begin{align}
    D_f\!\left(P_{X_1} P_{X_2} \cdots P_{X_n} \,\middle\|\, Q_{X_1} Q_{X_2} \cdots Q_{X_n}\right)
    = \tau_n\!\bigl(D_f(P_{X_1} \| Q_{X_1}), \ldots, D_f(P_{X_n} \| Q_{X_n})\bigr),
    \label{web:234}
\end{align}
for all distributions $P_{X_1}, P_{X_2}, \ldots, P_{X_n}$ and $Q_{X_1}, Q_{X_2}, \ldots, Q_{X_n}$ on arbitrary alphabets
$\mathcal{X}_1, \mathcal{X}_2, \ldots, \mathcal{X}_n$.
They proved that, under second-order polynomial conditions on $\tau_n$, only the KL divergence, cross-entropy, and Hellinger divergences admit tensorization.
In contrast to \eqref{web:234}, we do not require equality for
$D_f\!\left(P_{X_1} P_{X_2} \cdots P_{X_n} \,\middle\|\, Q_{X_1} Q_{X_2} \cdots Q_{X_n}\right)$,
but rather an upper bound in terms of the marginal divergences
$D_f(P_{X_i} \| Q_{X_i})$.

In this paper, we use Csiszar's notion of $f$-mutual information \cite{csiszar1972class} (when $G(x)=x$):
\[
I_f(X;Y) = \min_{q_Y}D_f\left(p_{XY} \| p_X q_Y\right),
\]
However, another common definition of $f$-information is
\begin{align}
\widetilde{I_f}(X;Y) = D_f\left(p_{XY} \| p_X p_Y\right).
\label{tIfXYeq}    
\end{align}
Characterizing the set of non-trivial functions \(f\) that satisfy the subadditivity property for $\widetilde{I_f}(X;Y)$ is mentioned as an open question in \cite[Section 7.8]{Polyanskiy25}, i.e., functions $f$ for which
\[
\widetilde{I_f}(X_1X_2;Y_1Y_2)\le \widetilde{I_f}(X_1;Y_1)+\widetilde{I_f}(X_2;Y_2)
\]
for arbitrary distributions of the form $p_{X_1,X_2}p_{Y_1|X_1}p_{Y_2|X_2}$.
It is stated that the above inequality holds for \(f(x)=x\log x\) and \(f(x)=(x-1)\log x\).

In \cite{masiha2023f}, the authors prove that if \(1/f''\) is concave, for any distribution of the form \(p_{ABC} = p_A p_B p_{C|AB}\), we have
\[
\widetilde{I_f}(C;A,B) \geq \widetilde{I_f}(C;A) + \widetilde{I_f}(C;B).
\]
They further note that for this class of functions, \(D_f\) is supermodular. Explicitly, for any joint distribution \(p_{X_1,X_2,X_3}\) and any product distribution \(q_{X_1}q_{X_2}q_{X_3}\) over arbitrary alphabets \(\mathcal{X}_1, \mathcal{X}_2, \mathcal{X}_3\), the following holds:
\[
D_f(p_{X_1X_2X_3} \| q_{X_1}q_{X_2}q_{X_3}) + D_f(p_{X_3} \| q_{X_3}) \geq D_f(p_{X_1X_3} \| q_{X_1}q_{X_3}) + D_f(p_{X_2X_3} \| q_{X_2}q_{X_3}).
\]
In \cite[Theorem 14]{beigi2018phi}, the set
\begin{align*}
    \bigg\{ (\lambda_1,\dots,\lambda_n) \in [0,1]^n:
    \sum_{i=1}^n \lambda_i \widetilde{I_f}(U;X_i)  
    \leq \widetilde{I_f}(U;X^n)\quad \forall p_{U|X^n}
    \bigg\}
\end{align*}
is called the $f$-hypercontractivity ribbon and is shown to have data processing and tensorization properties if $1/f''(x)$ is concave.

\subsection*{Organization}
The paper is organized as follows. We first define $(G,f)$-divergence and $(G,f)$-information and establish data processing and Fano-type inequalities. We then develop the relationship between divergence and information subadditivity, along with alphabet-reduction principles and characterizations of admissible transforms $G$. Next, we specialize to $G\in\{x,\log(1+x),-\log(1-x)\}$ and provide sufficient conditions on $f$ for subadditivity, with examples and tables. We conclude with applications to finite-blocklength coding converses, binary hypothesis testing, and error exponents.

Throughout, we restrict attention to finite alphabets for clarity; extensions beyond this setting are left for future work.


\subsection*{Notation}
Sets are denoted by calligraphic letters. 
We denote random variables by uppercase letters, such as \(X\) and \(Y\), and their realizations by the corresponding lowercase letters, such as \(x\) and \(y\). On the other hand, functions are denoted by lowercase letters (as in $f(x)$) or uppercase letters (as in $G(x)$). Throughout this paper, we assume all probability distributions are defined on finite alphabet sets. The distribution of a random variable \(X\) over the alphabet \(\mathcal{X}\) is denoted by a lowercase or uppercase letter  with the subscript \(X\), such as \(p_X\) or \(Q_X\). A joint distribution over \(\mathcal{X} \times \mathcal{Y}\) is denoted with the subscript \(XY\), as in \(p_{XY}\). The conditional distribution of \(Y\) given \(X\) is denoted with the subscript \(Y|X\), such as \(p_{Y|X}\) or \(q_{Y|X}\). The notation $X\rightarrow Y\rightarrow Z$ signifies that $X$ and $Z$ are conditionally independent given $Y$, which implies that $P_{X,Z|Y}=P_{X|Y}P_{Z|Y}$.  The sequences $(U_1, U_2, \ldots, U_{i})$  and $(U_i,U_{i+1},\cdots,U_n)$ are denoted as $U^i$ and $U_{i}^{n}$ respectively. 
Throughout the paper, we use the symbol $\log$ to denote the natural logarithm.
\section{$(G,f)$-divergence and information measure}
In this section, we define the $(G,f)$-divergence, which is a simple generalization of the concept of $f$-divergence. $f$-divergence was originally proposed by Ali and Silvey \cite{ali1966general} as well as Csiszár \cite{csiszar2011information,csiszar1967information}. The $f$-divergence metric quantifies the dissimilarity between two probability distributions on a common sample space. Since its inception, the framework of $f$-divergences has been utilized extensively across disciplines such as machine learning, statistics, and information theory. For a detailed examination of $f$-divergences, one may refer to works such as \cite{LieseV_book87,Vajda_1989,Vajda_2009,sason2016f,LieseV_IT2006}.

For simplicity of exposition, we restrict our attention in this work to distributions over finite alphabets. Extensions to continuous alphabets are possible.

\begin{definition}
Consider two probability mass functions \(p_X\) and \(q_X\) defined on a finite set \(\mathcal{X}\). We say that \(p_X\) is absolutely continuous with respect to \(q_X\), denoted \(p_X \ll q_X\), if for every \(x \in \mathcal{X}\),
\[
q_X(x) = 0 \ \implies\  p_X(x) = 0.
\]
\end{definition}

\begin{definition}\label{def:f-div-deff}
For a convex function $f:[0,\infty)\rightarrow\mathbb{R}$ satisfying $f(1)=0$, the $f$-divergence  between distributions $p_X$ and $q_X$ is defined as
\begin{align}
D_f(p_X\|q_X)
= \sum_{x\in\mathcal{X}} q_X(x)\, f\!\left(\frac{p_X(x)}{q_X(x)}\right)
\end{align}
with the convention that when both $q_X(x)=p_X(x)=0$, we define $0f\!\left(\frac{0}{0}\right)=0$.
\end{definition}
\begin{lemma}\label{domain:lem}
    For any distributions $p_X$ and $q_X$ on any set $\mathcal{X}$, we have
    \begin{align}
        D_f(p_X\|q_X) \leq \mathsf{D_m}(f) \triangleq f(0) + \lim_{t\rightarrow\infty}\frac{f(t)}{t}. \label{defDmf}
    \end{align}
\end{lemma}
\begin{proof}
  The proof can be found in Appendix~ \ref{domain:lem:proof}.  
\end{proof}
\begin{definition}\label{def:admissible}
    A pair of functions $(G(x),f(x))$ is called admissible if $f:[0,\infty)\rightarrow\mathbb{R}$ is a convex function satisfying $f(1)=0$, and $G:[0,\mathsf{D_m}(f))\rightarrow [0,\infty)$ is a non-decreasing function satisfying $G(0)=0$. Here, $\mathsf{D_m}(f)$ is defined as in \eqref{defDmf}.
\end{definition}

\begin{definition}\label{def:div:mutual}
Take an admissible pair $(G,f)$. Given 
 $p_X\ll q_X$ on a finite alphabet $\mathcal{X}$, 
we define the $(G,f)$-divergence $\mathcal{D}_{G,f}$ by
\begin{align}
\mathcal{D}_{G,f}(p_X\|q_X)
&\triangleq G\!\big(D_f(p_X\|q_X)\big).
\end{align}
\end{definition}

The idea of taking a non-decreasing function of an $f$-divergence is not new \cite{polyanskiy2010arimoto}. Powers of certain $f$-divergences define a distance measure on the set of probability distributions; see, e.g., \cite{Vajda_2009,osterreicher1996class}. Notably, the R\'enyi divergence can be written as a function of the Hellinger divergence, which is itself an $f$-divergence. 
For $\alpha \in (0,1) \cup (1,\infty)$, this relationship is given by:
\begin{align}\label{renyimeetshellinger}
D_\alpha(P \| Q ) = G\bigl(\mathscr{H}_\alpha(P \| Q)\bigr),
\end{align}
where 
\begin{equation*}
G(x)=\frac{1}{\alpha -1} \log \left( 1 + (\alpha - 1) \, x \right),
\end{equation*}
and the Hellinger divergence of order $\alpha \in (0,1) \cup (1, \infty)$ is 
defined as \cite{jeffreys1946invariant, LieseV_book87}:
\begin{align} \label{eq: Hel-divergence}
\mathscr{H}_{\alpha}(P \| Q) &= D_{f_\alpha}(P \| Q),
\end{align}
generated by the function
\begin{align} \label{eq: H as fD}
f_\alpha(x) &= \frac{x^\alpha-1}{\alpha-1}.
\end{align}

\subsection{$(G,f)$-information measure}
Following Csiszar's definition of mutual $f$-information in
\cite{csiszar1972class}, we define the $(G,f)$-information measure as follows:

\begin{definition}

For any two random variables $X$ and $Y$, we define the associated mutual information $I_{G,f}(X;Y)$ as
\begin{align*}
I_{G,f}(X;Y)
\triangleq \min_{q_Y} \sum_x p_X(x)\,\mathcal{D}_{G,f}(p_{Y|X=x}\|q_Y)
= \min_{q_Y} \sum_x p_X(x)\, G\!\big(D_f(p_{Y|X=x}\|q_Y)\big).
\end{align*}
\end{definition}
It follows directly that \(I_{G,f}(X;Y)\) reduces to Shannon's mutual information when we choose \(G(x)=x\) and \(f(x)=x\log x\). However, in general, unlike Shannon’s mutual information, $I_{G,f}(X;Y)$ is not symmetric, i.e., $I_{G,f}(X;Y) \neq I_{G,f}(Y;X)$. 

\begin{remark}
It is also possible to use Sibson's approach to define $(G,f)$-information:
\[
\min_{q_Y}\,\mD_{G,f}(p_{XY}\|p_Xq_Y)=
G\left(
\min_{q_Y}\,D_{f}(p_{XY}\|p_Xq_Y)
\right)
\]
where we used the fact that $G$ is non-decreasing. However, we do not adopt Sibson's approach in this paper.
\end{remark}

\begin{example} Take some $\alpha > 1$ and let $G(x) = x^{\frac{1}{\alpha}}$ and $f(x) = |x-1|^\alpha$. This leads to
\begin{align}
  I_{G,f}(X;Y)=\min_{q_Y} \sum_{x} p_X(x)\left( \sum_{y} q_Y(y)\,\left|\frac{p_{Y|X}(y|x) - q_Y(y)}{q_Y(y)}\right|^\alpha \right)^{\frac{1}{\alpha}}.
\end{align}
One can compare this with the correlation measure defined in \cite{mojahedian2019correlation} as
\begin{align}
  V_\alpha(Y;X)
  &= \sum_{x} \left( \sum_{y} p_Y(y)\,\bigl|p_{X|Y}(x|y) -
  p_X(x)\bigr|^\alpha \right)^{\frac{1}{\alpha}} 
  \\&= \sum_{x} p_X(x)\left( \sum_{y} p_Y(y)\,\left|\frac{p_{Y|X}(y|x) - p_{Y}(y)}{p_{Y}(y)}\right|^\alpha \right)^{\frac{1}{\alpha}} 
  \end{align}
  
\end{example}

 \subsection{Data processing and Fano's inequality}
We expect a measure of information to satisfy the data processing inequality. 
\begin{lemma}\label{IGF:concave} Take an admissible pair $(G,f)$. Then:
    \begin{enumerate}
        \item $I_{G,f}(X;Y)$ is a concave function of the input distribution $p_X$ for a fixed channel $p_{Y|X}$. 
        \item If $X \rightarrow Y \rightarrow Z$ forms a Markov chain, then $I_{G,f}$ satisfies the data processing inequality
        \begin{align}
           I_{G,f}(X;Y) \geq I_{G,f}(X;Z). \label{DPI:Igf}
        \end{align} 
        \item If $X \rightarrow Y \rightarrow Z$ forms a Markov chain, then
           \begin{align}
           I_{G,f}(X;YZ) &= I_{G,f}(X;Y), \label{DPI:Igfn1}\\
           I_{G,f}(XY;Z) &= I_{G,f}(Y;Z).\label{DPI:Igfn2}
        \end{align} 
    \end{enumerate}
\end{lemma}
\begin{proof}
The proof is given in Appendix \ref{lemma:1:new}.
\end{proof}
In the following lemma, the data processing inequality is extended to the Markov chain $A \to X \to Y \to B$ for increasing \underline{convex} functions $G$ (for example, $G(x)=x$ or $G(x)=-\log(1-x)$).
\begin{lemma}\label{4DPI:lem}
Assume that $A \to X \to Y \to B$ forms a Markov chain, $f$ is a convex function with $f(1)=0$, and $G$ is an increasing convex function on $[0,\infty)$. Then the following data processing inequality holds:
\begin{align}
    I_{G,f}(A;B) \leq I_{G,f}(X;Y).
\end{align}
\end{lemma}

\begin{proof}
The proof can be found in Appendix \ref{lemma:2:new}.
\end{proof}

\begin{lemma}[Fano-type inequality for $I_{G,f}$]\label{lemmaFano}
Let $(G,f)$ be an admissible pair where $G$ is an increasing convex function. Let $p_M$ be the uniform distribution on a finite set $\mathcal{M}$, and let $p_{M\hat{M}}$ be a joint distribution on $\mathcal{M} \times \mathcal{M}$ such that $p_{M\hat{M}}(M \neq \hat{M}) = \epsilon$. Then, under $p_{M\hat{M}}$,
\begin{align}
    I_{G,f}(M;\hat{M})
    \;\ge\;
    G\!\left(
        \frac{1}{|\mathcal{M}|}\,
        f\bigl(|\mathcal{M}|(1-\epsilon)\bigr)
        + \frac{|\mathcal{M}|-1}{|\mathcal{M}|}\,
          f\!\left(\frac{|\mathcal{M}|\epsilon}{|\mathcal{M}|-1}\right)
    \right),\label{lower:on:IMM}
\end{align}
where $|\mathcal{M}|$ denotes the cardinality of $\mathcal{M}$.
\end{lemma}

\begin{proof}
See Appendix \ref{lemma:3:new} for the proof.
\end{proof}

\section{Subadditivity of $D_{G,f}$ and $I_{G,f}$}
\begin{definition}
We say that an admissible  pair $(G,f)$ is  divergence-subadditive, if for any arbitrary distributions $q_Y\ll r_Y$ on any finite set $\mathcal{Y}$ and any $q_Z\ll r_Z$ on any finite set $\mathcal{Z}$, the following holds:
\begin{align}
  \mathcal{D}_{G,f}(q_Y q_Z \,\|\, r_Y r_Z)
  \;\leq\;
  \mathcal{D}_{G,f}(q_Y \,\|\, r_Y)
  +
  \mathcal{D}_{G,f}(q_Z \,\|\, r_Z). \label{gen-def-1}
\end{align}
Next, we say 
that a pair $(G,f)$ is  \emph{information-subadditive} if for any random variables $X,Y,Z$ defined on finite sets we have
\begin{align}
  I_{G,f}(X;YZ)&\leq I_{G,f}(X;Y) + I_{G,f}(X;Z)\label{sub_I_def}
\end{align}
whenever $Y \rightarrow X \rightarrow Z$ forms a Markov chain, i.e., $p_{XYZ} = p_X p_{Y|X} p_{Z|X}$.
\end{definition}
\begin{remark}\label{rmksubaddinf} Using part 3 of Lemma \ref{IGF:concave} and substituting $X=(X_1,X_2), Y=Y_1, Z=Y_2$, we get that the \emph{information-subadditive} property is equivalent with subadditivity of $(G,f)$-information over product channels:
    $$I_{G,f}(X_1X_2;Y_1Y_2)\leq I_{G,f}(X_1;Y_1)+I_{G,f}(X_2;Y_2)$$
    for any $$p(x_1,x_2)p(y_1|x_1)p(y_2|x_2).$$
\end{remark}

In the following theorem, we show that information-subadditivity and divergence-subadditivity are equivalent under some mild assumptions.
\begin{theorem} \label{gen-thm1}
    Let $(G,f)$ be an admissible pair (see Definition \ref{def:admissible}) for a continuous function $G$ where $G(x)>0$ for all $x>0$, and a function $f(x)$ that is strictly convex at $x=1$. Then, $(G,f)$ is \emph{divergence-subadditive} if and only if it is \emph{information-subadditive}.
\end{theorem}
\begin{proof} 
We first show that divergence-subadditivity implies information-subadditivity.
Take some arbitrary $p_{XYZ}=p_Xp_{Y|X}p_{Z|X}$. Let $q_Y^\star$ be the minimizer of $I_{G,f}(X;Y)$ and $q_Z^\star$ be the minimizer of $I_{G,f}(X;Z)$. We have
\begin{align}
I_{G,f}(X;Y)+I_{G,f}(X;Z) & =  \mathbb{E}_{p_X}\big[\mD_{G,f}(p_{Y|X}\|q_Y^\star)\big]+\mathbb{E}_{p_X}\big[\mD_{G,f}(p_{Z|X}\|q_Z^\star)\big]\nonumber\\
 & =  \mathbb{E}_{p_X}\big[\mD_{G,f}(p_{Y|X}\|q_Y^\star)+\mD_{G,f}(p_{Z|X}\|q_Z^\star)\big]\nonumber \\
 & \geq \mathbb{E}_{p_X}\big[\mD_{G,f}(p_{Y|X}p_{Z|X}\|q_Y^\star q_Z^\star)] \label{eq2-gen}\\
&=  \mathbb{E}_{p_X}\big[\mD_{G,f}(p_{YZ|X}\|q_Y^\star q_Z^\star)]\nonumber\\
& \geq  \min_{q_{YZ}} \mathbb{E}_{p_X}\big[\mD_{G,f}(p_{YZ|X}\|q_{YZ} )] \nonumber\\
&= I_{G,f}(X;YZ)\nonumber
\end{align}
where \eqref{eq2-gen} follows from the assumption that inequality  \eqref{gen-def-1} holds.

Next, we show that information-subadditivity implies divergence-subadditivity. Take arbitrary $r_Y,r_Z,q_Y,q_Z$ and define a binary input channel $p_{Y|X}$
\begin{align}
    p(Y=y|X=0)&=q_Y(y)\\
    p(Y=y|X=1)&=r_Y(y)\\
    p(Z=z|X=0)&=q_Z(z)\\
    p(Z=z|X=1)&=r_Z(z).
\end{align}
Assume that
\begin{align*}
    p(X=0)&=\epsilon\\
p(X=1)&=1-\epsilon.
\end{align*}
Also, assume that $Y\rightarrow X\rightarrow Z$ forms a Markov chain. Then, set 
\begin{align*}
\Upsilon(\epsilon)&=I_{G,f}(X;Y)+I_{G,f}(X;Z)-I_{G,f}(X;YZ)
\\&=
    \min_{s_Y}\sum_x p(x)\mD_{G,f}(p_{Y|x}\|s_Y) +\min_{s_Z}\sum_x p(x)\mD_{G,f}(p_{Z|x}\|s_Z) -\min_{s_{YZ}}\sum_x p(x)\mD_{G,f}(p_{YZ|x}\|s_{YZ}) 
    \\&=
    \min_{s_Y}\left(\epsilon \mD_{G,f}(p_{Y|X=0}\|s_Y) +(1-\epsilon)\mD_{G,f}(p_{Y|X=1}\|s_Y)\right)
    \\&\quad +\min_{s_Z}\left(\epsilon \mD_{G,f}(p_{Z|X=0}\|s_Z) +(1-\epsilon)\mD_{G,f}(p_{Z|X=1}\|s_Z)\right) 
    \\&\quad-\min_{s_{YZ}}\left(\epsilon \mD_{G,f}(p_{YZ|X=0}\|s_{YZ}) +(1-\epsilon)\mD_{G,f}(p_{YZ|X=1}\|s_{YZ})\right).
\end{align*}
Let 
\begin{align}
    \phi_{XY}(\epsilon,s_Y)
    &\triangleq
    \epsilon \mD_{G,f}(p_{Y|X=0}\|s_Y)
    +(1-\epsilon)\mD_{G,f}(p_{Y|X=1}\|s_Y),
    \\
    \phi_{XZ}(\epsilon,s_Z)
    &\triangleq
    \epsilon \mD_{G,f}(p_{Z|X=0}\|s_Z)
    +(1-\epsilon)\mD_{G,f}(p_{Z|X=1}\|s_Z), 
    \\
    \phi_{XYZ}(\epsilon,s_{YZ})
    &\triangleq
    \epsilon \mD_{G,f}(p_{YZ|X=0}\|s_{YZ})
    +(1-\epsilon)\mD_{G,f}(p_{YZ|X=1}\|s_{YZ}).
\end{align}
The functions $\phi_{XY}(\epsilon,\cdot)$, $\phi_{XZ}(\epsilon,\cdot)$, and $\phi_{XYZ}(\epsilon,\cdot)$ are linear in $\epsilon$ and continuous in their distribution arguments. Moreover, when we minimize over the (compact) spaces of distributions, the minimizers at $\epsilon=0$ are unique and given by $s^*_Y=p_{Y|X=1}$, $s^*_Z=p_{Z|X=1}$, and $s^*_{YZ}=p_{YZ|X=1}$, respectively. Uniqueness follows from the strict convexity of $f(x)$ at $x=1$ \cite[Proposition 1]{sason2016f} and $G(x)>0$ for all $x>0$. Hence, by Danskin’s theorem {\cite[p.~805]{Bertsekas}}, we may differentiate $\Upsilon(\epsilon)$ through the minima at $\epsilon=0$.
Since 
\[
\Upsilon(\epsilon)
=
\min_{s_Y}\phi_{XY}(\epsilon,s_Y)
+ \min_{s_Z}\phi_{XZ}(\epsilon,s_Z)
- \min_{s_{YZ}}\phi_{XYZ}(\epsilon,s_{YZ}).
\]
Danskin’s theorem implies that $\Upsilon(\epsilon)$ is differentiable at $\epsilon=0$ and that
\begin{align}
    \frac{d\Upsilon(\epsilon)}{d\epsilon}\bigg|_{\epsilon=0}
    &= \frac{\partial \phi_{XY}(\epsilon,p_{Y|X=1})}{\partial \epsilon}
     + \frac{\partial \phi_{XZ}(\epsilon,p_{Z|X=1})}{\partial \epsilon}
     - \frac{\partial \phi_{XYZ}(\epsilon,p_{YZ|X=1})}{\partial \epsilon}\bigg|_{\epsilon=0} \\
    &= \mD_{G,f}(q_{Y}\|r_{Y})
       + \mD_{G,f}(q_{Z}\|r_{Z})
       - \mD_{G,f}(q_{Z}q_{Y}\|r_{Z}r_{Y}),
    \label{der:dons:1}
\end{align}
Furthermore, $\Upsilon(0)=0$, because at $\epsilon=0$ the variable $X$ is deterministic and the equality in \eqref{sub_I_def} holds trivially. By assumption, inequality \eqref{sub_I_def} holds for all $\epsilon\in[0,1]$, and therefore $\Upsilon(\epsilon)\ge 0$ on $[0,1]$. In particular, the right derivative at $0$ must satisfy
\[
\Upsilon'(0)\ge 0,
\]
which, together with \eqref{der:dons:1}, yields the desired inequality. This completes the proof.
\end{proof}
Although the inequalities \eqref{gen-def-1} and \eqref{sub_I_def} are equivalent, it is more convenient to work with \eqref{gen-def-1}, since the definition of $I_{G,f}$ involves a minimization problem. Nevertheless, \eqref{gen-def-1} still requires one to verify the inequality for arbitrary finite–alphabet random variables $q_Y,r_Y,q_Z,r_Z$. In the following lemma, we show that for a certain class of functions $G$, the subadditivity of $\mathcal{D}_{G,f}$ over arbitrary finite alphabets already follows from its subadditivity for binary random variables. In other words, within this class, establishing subadditivity for general random variables with arbitrary alphabet size is equivalent to establishing it for binary random variables.
\begin{lemma}\label{gen-lemma1}
Let $G:[0,\nu)\rightarrow[0,+\infty)$ be an increasing function with $G(0)=0$ such that the range of $G$ is $[0, \infty)$. Moreover, assume that for every $c>0$, the function
\begin{align}
x \mapsto G^{-1}\big(G(x) + G(c)\big)\label{cond:G}
\end{align}
is concave on $[0,\nu)$.  
Let $f$ be a convex function on $[0,\infty)$ with $f(1)=0$ such that $\mathsf{D_m}(f)\leq \nu$ where $\mathsf{D_m}(f)$ was defined in \eqref{defDmf}.  Assume that \eqref{gen-def-1} holds for any distributions $q_Y,r_Y$ on $\mathcal{Y}$ and $q_Z,r_Z$ on $\mathcal{Z}$ when both alphabets $\mathcal{Y}$ and $\mathcal{Z}$ are binary (i.e., of size two).  
Then \eqref{gen-def-1} also holds for this pair and for  any distributions $q_Y,r_Y$ and $q_Z,r_Z$ defined on finite alphabets $\mathcal{Y}$ and $\mathcal{Z}$ of arbitrary size.
\end{lemma}
\begin{proof}
The proof can be found in Appendix~\ref{lem:1:pro}.
\end{proof}
The following lemma whose proof can be found in Appendix~\ref{lem:1:pro}, characterizes the class of increasing functions $G(x)$ that satisfy condition~\eqref{cond:G}.
\begin{lemma}\label{func}
Let $ G : [0, \nu) \to [0, \infty) $ be a strictly increasing function with $ G(0) = 0$, and suppose $ G $ is twice differentiable with $ G'(x) > 0 $ for all $ x \in  (0,\nu) $. Moreover, assume that the range of $G$ is $[0, \infty)$. For any $ c > 0 $,  the function  
\[x\mapsto
 G^{-1}(G(x) + G(c)).
\]
 is concave for all $ c > 0 $ if and only if the function  
\[
u(x) = \frac{1}{G'(x)}
\]  
is concave.
\end{lemma}

\begin{corollary}\label{corol:g}
The class of functions satisfying the condition in Lemma \eqref{func} is characterized by  
\[
G(x) = \int_0^x \frac{1}{u(t)}  dt,
\]  
where $ u(t) $ is any positive concave function on $[0,\nu)$. Note that the domain of $G$ is also contained in $[0,\nu)$.
  Notable examples include:  

\begin{itemize}
    \item  For $u(x)=1$ on $[0,+\infty)$, we obtain $ G(x) = x $.
    \item For $u(x)=\frac{x^{1-p}}{p}$ on $[0,+\infty)$, we obtain $ G(x) = x^p $ for $ 0 < p \leq 1 $. 
    \item For $u(x)=\frac{1}{1+x}$ on $[0,+\infty)$, we obtain $ G(x) = \log(1 +  x) $.
    \item For $u(x)=\tanh(x+\sinh^{-1}(1))$ on $[0,+\infty)$, we obtain $G(x)=\log\left(\sinh \left(x+\sinh^{-1}(1)\right)\right)$. 
     \item For $u(x)=\frac{1}{1-x}$ on $[0,1)$, we obtain $G(x)=-\log(1-x)$.
    \end{itemize}
\end{corollary}
\begin{remark}
    We can also define functions like $G(x)=\frac{1}{a} \log(1 + a x) $ for $ a < 0 $. However, the domain of this function is not the entire positive numbers and the range of $D_f$ needs to be inside the domain of $G$. 
\end{remark}

The condition in Lemma \ref{func} only imposes a constraint on $G$, but no constraint on $f$. In contrast, the following lemma imposes a constraint on $f$ but no constraint on $G$. 
\begin{lemma}[Sufficiency of Small Supports via Stationarity] \label{lem:root-counting}
Let $(G,f)$ be an admissible pair where $G$ is strictly increasing. Define the function
\begin{align}
    H(t; \lambda) \triangleq \lambda f(t) - \sum_{z \in \mathcal{Z}} r_Z^*(z) f\left( t \frac{q_Z^*(z)}{r_Z^*(z)} \right),
\end{align}
where $\lambda \in \mathbb{R}_{+}$ is a positive scalar and $(q_Z^*, r_Z^*)$ are arbitrary probability distributions on a finite set $\mathcal{Z}$.

If for every $\lambda > 0$ and every pair of distributions $(q_Z^*, r_Z^*)$, the equation
\begin{align} \label{eq:root-counting}
    H(t; \lambda) = a + b t
\end{align}
has at most $k$ distinct solutions for $t \in [0, \infty)$ (where $a, b$ are arbitrary real constants), then it suffices to verify the divergence subadditivity inequality \eqref{gen-def-1} only for distributions $q_Y, r_Y$ supported on alphabets of size at most $k$.

In particular, if $H(t; \lambda)$ is strictly convex or strictly concave for all $\lambda$ and all $(q_Z^*, r_Z^*)$, then \eqref{eq:root-counting} has at most $k=2$ roots, and checking binary alphabets is sufficient.
\end{lemma}

\begin{proof}
    The proof can be found in Appendix~\ref{lem:3:proof}.
\end{proof}

\begin{example}[Comparison of Sufficiency Conditions]
Consider the admissible pair given by $G(x) = e^x - 1$ and $f(t) = t^2 - 1$. We demonstrate that this pair fails the condition of Lemma \ref{func} yet satisfies the sufficiency condition of Lemma \ref{lem:root-counting}. To check Lemma \ref{func}, we examine the function $u(x) = 1/G'(x)$. Since $G'(x) = e^x$, we have $u(x) = e^{-x}$. The second derivative is $u''(x) = e^{-x}$, which is strictly positive. Consequently, $u(x)$ is strictly convex, meaning the concavity argument cannot be used to reduce the alphabet size.

In contrast, we apply Lemma \ref{lem:root-counting} by constructing the stationarity function $H(t; \lambda)$. Substituting the definitions yields
\begin{align}
    H(t; \lambda) = \lambda(t^2 - 1) - \mathbb{E}_Z \left[ \left(t \frac{q_Z^*}{r_Z^*}\right)^2 - 1 \right] = \left(\lambda - \mathbb{E}_Z \left[\left(\frac{q_Z^*}{r_Z^*}\right)^2\right]\right) t^2 + (1-\lambda).
\end{align}
This is a quadratic polynomial in $t$. The resulting stationarity equation $H(t; \lambda) = a + bt$ is a quadratic equation, which admits at most $k=2$ real roots. Therefore, despite violating the concavity condition, the root-counting argument guarantees that checking binary distributions is sufficient to verify subadditivity for this divergence measure.
\end{example}
From now on, we restrict ourselves to the functions $G(x) = x$, $G(x) = \log(1+x)$, and $G(x) = -\log(1-x)$. 
\subsection{Subadditive functions}
In this subsection, we aim to identify convex functions $f(x)$ on the domain $[0,\infty)$ with $f(1)=0$ such that $\mathcal{D}_{G,f}$ is subadditive for $G(x)\in\{x,\log(1+x),-\log(1-x)\}$. As shown in Corollary~\ref{corol:g}, for these choices of $G$ and any convex function $f(x)$, it suffices to verify inequality~\eqref{gen-def-1} only for binary random variables. We define the following set:
\begin{align}
\mathcal{F}_{\{G\}}=\left\{ f:
        (G,f) \text{ is admissible and divergence-subadditive } \right\}.
\end{align}
Characterizing the sets $\mathcal{F}_{\{G\}}$ for $G(x)\in\{x,\log(1+x),-\log(1-x)\}$ are, in general,   difficult problems. In the following, we focus on identifying tractable subsets of these classes.
\subsubsection{A subset of $\mathcal{F}_{\{x\}}$}
\emph{Note:} In this section we assume $G(x)=x$. We utilize a function $g(x)$ which should not be confused with $G(x)$.

We have the following theorem, whose proof can be found in Appendix~\ref{app:main:G:x}. 
\begin{theorem}\label{Th:main:G:x}
Let $f(x)$ be a twice continuously differentiable convex function on $[0,\infty)$ with $f(1)=0$, and define
\begin{align}\label{defT}
    \mathcal{T}
    = \Big\{ g:[0,\infty)\to[0,\infty) \,\Big|\,
        g(x)\ \text{is concave} \Big\}.
\end{align}
If $x^2 f''(x) \in \mathcal{T}$, then $f(x) \in \mathcal{F}_{\{x\}}$.
\end{theorem}
The above theorem shows that many known $f$-divergences are subadditive with $G(x)=x$. See Table \ref{tab:divergences}.  
\begin{corollary}
Let $g(x)$ be a non-negative concave function on $[0,\infty)$ and define
\begin{align*}
  f(x) &= \widetilde{a}x+\widetilde{b} +\int_1^x (x-y)\frac{g(y)}{y^2}~dy,
  \end{align*}
where  constants $\widetilde{a}$ and $\widetilde{b}$ are chosen such that $f(1)=0$. Then the function $f(x)$ belongs to the class $\mathcal{F}_{\{x\}}$.

\end{corollary}
\begin{example}
\leavevmode
\begin{itemize}
    
    \item For function $g(x)=\log(1+x) \in \mathcal{T}$ we get 
    \[
    f(x)=-\int_{1}^{x}\frac{\log(1+t)}{t}dt+x\log(x)-(1+x)\log(1+x)+\widetilde{a}(x-1)+2\log(2).
    \]
    \item Let $\mu$ be a finite support, non-negative measure on $[0,\infty)$. A function of the form
\[
  g(x)
  = \alpha + \beta x + \int_{0}^{\infty} \frac{x\lambda}{\lambda + x}\,d\mu(\lambda),
  \qquad \alpha,\beta \ge 0,
\]
is called an operator concave function and satisfies $g(x)\in\mathcal{T}$. The corresponding function $f$ is
\begin{align*}
  f(x)
  = -\alpha \log x
    + \beta x \log x
    + \int_{0}^{\infty} \left[x\log x - (x+\lambda)\log\left(\frac{x+\lambda}{\lambda+1}\right)\right]\,d\mu(\lambda)
    + \widetilde{a} x + \widetilde{b},
\end{align*}
where $\widetilde{a}$ and $\widetilde{b}$ are chosen such that $f(1)=0$.
\end{itemize}
\end{example}

\begin{remark}
    In a concurrent work \cite{Mike}, the third author and his student show that for every $f$ where $x^2 f''(x) \in \mathcal{T}$, the following inequality holds:
    \begin{align*}
    \widetilde{I_f}(U;X)+\widetilde{I_f}(U;Y)-\widetilde{I_f}(U;XY)\leq \widetilde{I_f}(X;X)+\widetilde{I_f}(Y;Y)-\widetilde{I_f}(XY;XY), \qquad\forall p_{UXY},
    \end{align*}
    where $\widetilde{I_f}(X;Y)$ was defined in \eqref{tIfXYeq}.
\end{remark}

\begin{table}[h!]
\centering
\caption{Case of $G(x)=x$. Verification of the subadditivity condition.}
\label{tab:divergences}
\begin{tabular}{|l|c|c|c|c|}
\hline
\textbf{Divergence Name} & \textbf{Generator} $f(x)$ & $g(x) = x^2 f''(x)$ & \textbf{Concave?} & \textbf{Covered by Thm \ref{Th:main:G:x}?} \\
\hline
KL Divergence & $x \log x$ & $x$ & Yes & Yes \\
\hline
Reverse KL & $-\log x$ & $1$ & Yes & Yes \\
\hline
Squared Hellinger & $(\sqrt{x}-1)^2$ & $\frac{1}{2}\sqrt{x}$ & Yes & Yes \\
\hline
Jensen-Shannon & $x \log x - (1+x)\log(\frac{1+x}{2})$ & $\frac{x}{1+x}$ & Yes & Yes \\
\hline
$\alpha$-Divergence ($0<\alpha<1$) & $\frac{1}{\alpha(\alpha-1)}(x^\alpha - 1)$ & $\frac{1}{\alpha} x^\alpha$ & Yes & Yes \\
\hline
Pearson $\chi^2$ & $(x-1)^2$ & $2x^2$ & No (Convex) & No \\
\hline
Triangular Discrimination & $\frac{(x-1)^2}{x+1}$ & $\frac{8x^2}{(x+1)^3}$ & No (S-shape) & No \\
\hline
Le Cam & $\frac{1-x}{2x+2}$ & $\frac{2x^2}{(x+1)^3}$ & No  & No \\
\hline
\end{tabular}
\end{table}

\subsubsection{A subset of $\mathcal{F}_{\{\log(1+x)\}}$}

\begin{table}[h!]
\centering
\caption{{Case of $G(x)=\log(1+x)$. Verification of the subadditivity condition. }}
\label{tab:divergences2}
\begin{tabular}{|l|c|c|c|}
\hline
\textbf{Divergence Name} & \textbf{Generator} $\hat{f}(x)$ & $g(x) = x^2 \hat{f}''(x)$ & $g(x)\in \mathcal{T}^{+}$ \\
\hline
Pearson $\chi^2$ & $(x-1)^2$ & $2x^2$ & Yes  \\
\hline
Triangular Discrimination & $\frac{(x-1)^2}{x+1}$ & $\frac{8x^2}{(x+1)^3}$ & No \\
\hline
Le Cam & $\frac{1-x}{2x+2}$ & $\frac{2x^2}{(x+1)^3}$ & No  \\
\hline
\end{tabular}
\end{table}
Assume that $\hat{f}(x)$ belongs to $\mathcal{F}_{\{\log(1+x)\}}$. Then inequality~\eqref{gen-def-1} for this function takes the form
\begin{align*}
 \log\bigl(1+ D_{\hat{f}}(q_Y q_Z \,\|\, r_Y r_Z)\bigr)
  \;\leq\;
  \log\bigl(1+D_{\hat{f}}(q_Y \,\|\, r_Y)\bigr)
  +
  \log\bigl(1+D_{\hat{f}}(q_Z \,\|\, r_Z)\bigr).
\end{align*}
Let $f(x) = \hat{f}(x) + 1$. Then $f(x)$ is also a convex function on $[0,\infty)$ with $f(1)=1$. In terms of $f(x)$, the above inequality simplifies to
\begin{align}
 D_{f}(q_Y q_Z \,\|\, r_Y r_Z)
  \;\leq\;
  D_{f}(q_Y \,\|\, r_Y)\,
  D_{f}(q_Z \,\|\, r_Z).
  \label{log+-new-form}
\end{align}
Thus, we are interested in convex functions on $[0,\infty)$ with $f(1)=1$ such that inequality~\eqref{log+-new-form} holds. 
In other words, by defining
\begin{align}
\mathcal{F}_{\{\log(1+x)\}}^{+}
  = \left\{ f:[0,\infty)\to\mathbb{R} \,\Big|\,
        f \text{ is convex},\ f(1)=1,\ \text{and inequality \eqref{log+-new-form} holds} \right\},
\end{align}
there exists a bijection between the sets $\mathcal{F}_{\{\log(1+x)\}}^{+}$ and $ \mathcal{F}_{\{\log(1+x)\}}$. Specifically, $f\in \mathcal{F}_{\{\log(1+x)\}}^{+}$ if and only if $f-1\in \mathcal{F}_{\{\log(1+x)\}}$. 
In the following theorem, we identify a tractable subset of $\mathcal{F}_{\{\log(1+x)\}}^{+}$.

\begin{theorem}\label{Th:main:G:log+}
Let $f(x)$ be a four times continuously differentiable convex function on $[0,\infty)$ with $f(1)=1$, and define
\begin{align}
    \mathcal{T}^{+}
    = \Big\{ g:[0,\infty)\to[0,\infty) \,\Big|\,
        x^2 g''(x) - g(\alpha)\,g\!\left(\frac{x}{\alpha}\right) \leq 0,\ \forall x\ge 0,\ \forall \alpha>0 \Big\}.
\end{align}
If $x^2 f''(x) \in \mathcal{T}^{+}$, then $f(x)\in \mathcal{F}_{\{\log(1+x)\}}^{+}$.
\end{theorem}

\begin{proof}
The proof can be found in Appendix~\ref{app:main:G:log+:p1}.
\end{proof}
Note that  $\mathcal{T}\subset\mathcal{T}^+$ where $\mathcal{T}$ was defined in \eqref{defT}. 

\begin{corollary}
Let $g(x)\in \mathcal{T}^{+}$ and define
\begin{align*}
  f(x) &= \int_1^x (x-y)\frac{g(y)}{y^2}~dy + \widetilde{a} x + \widetilde{b},
\end{align*}
where the constants $\widetilde{a}$ and $\widetilde{b}$ are chosen such that $f(1)=1$. Then the function $f(x)$ belongs to the class $\mathcal{F}_{\{\log(1+x)\}}^{+}$.
\end{corollary}
Note that characterizing the set $\mathcal{T}^{+}$ requires solving a functional second-order differential inequality for non-negative functions, due to the presence of the term $g\!\left(\frac{x}{\alpha}\right)$. This problem is, in general, quite difficult. In Appendix~\ref{app:main:G:log+:p2}, we identify  explicit subsets of $\mathcal{T}^{+}$.
\begin{example}
\leavevmode
\begin{itemize}
  \item For $g(x) = p(p-1)x^p \in \mathcal{T}^{+}$ with $p \in \mathbb{R}\setminus[0,1]$, we obtain $f(x) = x^p$.
 
  \item For $g(x) = \dfrac{x+2}{x+1} \in \mathcal{T}^{+}$, we obtain
  \[
    f(x)
    = -2\log x - x\log x + (1+x)\log(1+x)
      + \widetilde{a}(x-1) + 1 - 2\log 2.
  \]
  \item For $g(x) = p(p-1)x^{p} + p(1-\theta)x^{\theta} \in \mathcal{T}^{+}$, where $p>1$ and $\theta\in(0,1)$, the corresponding function is
  \[
    f(x) = x^{p} - \frac{p}{\theta}x^{\theta} + \frac{p}{\theta}.
  \]
\end{itemize}
\end{example}
For more examples and further details, see Appendix~\ref{app:main:G:log+:p2}.

\subsubsection{A subset of $\mathcal{F}_{\{-\log(1-x)\}}$}

Suppose that $\hat{f}(x)\in \mathcal{F}_{\{-\log(1-x)\}}$. Since the domain of $G(x) = -\log(1-x)$ is $(-\infty,1)$, we must have $\mathcal{D}_{\hat{f}} < 1$ for all distributions, and thus we assume that $\hat{f}(x)\le 1$. In this case, inequality~\eqref{gen-def-1} takes the form
\begin{align*}
  -\log\bigl(1- D_{\hat{f}}(q_Y q_Z \,\|\, r_Y r_Z)\bigr)
  \;\leq\;
  -\log\bigl(1-D_{\hat{f}}(q_Y \,\|\, r_Y)\bigr)
  -
  \log\bigl(1-D_{\hat{f}}(q_Z \,\|\, r_Z)\bigr).
\end{align*}
Let $f(x) = 1-\hat{f}(x)$. Then $f(x)$ is a non-negative concave function on $[0,\infty)$ with $f(1)=1$. In terms of $f(x)$, the above inequality simplifies to
\begin{align}
  D_{f}(q_Y q_Z \,\|\, r_Y r_Z)
  \;\geq\;
  D_{f}(q_Y \,\|\, r_Y)\,
  D_{f}(q_Z \,\|\, r_Z).
  \label{logm-new-form}
\end{align}
Thus, we are interested in non-negative concave functions on $[0,\infty)$ with $f(1)=1$ such that inequality~\eqref{logm-new-form} holds.
In other words, by defining
\begin{align}
\mathcal{F}_{\{-\log(1-x)\}}^{-}
\triangleq
\Big\{
  f:[0,\infty)\to[0,\infty)
  \,\Big|\,
  f \text{ is concave},\ f(1)=1,\ \text{and inequality \eqref{logm-new-form} holds}
\Big\},
\end{align}
we observe that if $f\in \mathcal{F}_{\{-\log(1-x)\}}^{-}$, then $1-f\in \mathcal{F}_{\{-\log(1-x)\}}$. Consequently, for simplicity we focus on the set $\mathcal{F}_{\{-\log(1-x)\}}^{-}$.
 In the following theorem, we identify a tractable subset of $\mathcal{F}_{\{-\log(1-x)\}}^{-}$.
\begin{theorem}\label{Th:main:G:log-}
Let $f(x)$ be a four times continuously differentiable non-negative concave function on $[0,\infty)$ with $f(1)=1$, and define
\begin{align}
    \mathcal{T}^{-}
    = \Big\{ g:[0,\infty)\to[0,\infty) \,\Big|\,
        x^2 g''(x) + g(\alpha)\,g\!\left(\frac{x}{\alpha}\right) \le 0,\ \forall x\ge 0,\ \forall \alpha>0 \Big\}.
\end{align}
If $-x^2 f''(x) \in \mathcal{T}^{-}$, then $f(x)\in \mathcal{F}_{\{-\log(1-x)\}}^{-}$.
\end{theorem}
\begin{proof}
Please see Appendix \ref{app:main:G:log-:p1} for the  proof.
\end{proof}
Note that  $\mathcal{T}^-\subset\mathcal{T}$ where $\mathcal{T}$ was defined in \eqref{defT}. 
\begin{corollary}
Let $g(x)\in \mathcal{T}^{-}$ and define 
\begin{align*}
f(x) &= -\int_1^x (x-y)\frac{g(y)}{y^2}~dy + \widetilde{a} x + \widetilde{b},
\end{align*}
where the constants $\widetilde{a}$ and $\widetilde{b}$ are chosen such that $f(1)=1$ and $f(x)\ge 0$ for all $x\ge 0$. Then the function $f(x)$ belongs to the class $\mathcal{F}_{\{-\log(1-x)\}}^{-}$.
\end{corollary}
We believe that characterizing the set $\mathcal{T}^{-}$ is more difficult than characterizing $\mathcal{T}^{+}$. In Appendix~\ref{app:main:G:log-:p2}, we present several useful properties and illustrative examples for the set $\mathcal{T}^{-}$.
\begin{example}
  \leavevmode
\begin{itemize}
  \item For $g(x) = p(1-p)x^p \in \mathcal{T}^{-}$ with $p\in(0,1)$, we obtain
$f(x) = x^p$.
 \item For $g(x) = \dfrac{\sqrt{x}}{128}\bigl(9+\sin(\log x)\bigr)\in \mathcal{T}^{-}$, we obtain
  \[
    f(x) = \sqrt{x}\left(\frac{45 + \sin(\log x)}{160}\right) + \frac{115}{160}x.
  \]
\end{itemize}
\end{example}

\begin{remark}
    The case $p=1/2$ in the example above yields $f(x)=\sqrt{x}$. Combined with $G(x)=-\log(1-x)$, this recovers the Bhattacharyya distance, defined by $-\log \sum \sqrt{p(y)q(y)}$. Since $f(x)=\sqrt{x}$ belongs to $\mathcal{F}_{\{-\log(1-x)\}}^{-}$, this confirms the subadditivity of the Bhattacharyya distance as a special case of our framework.
\end{remark}
For more examples and further details, see Appendix~\ref{app:main:G:log-:p2}.

\section{Applications}
\subsection{Error probability in the finite blocklength regime}

Consider communication over a point-to-point memoryless channel in the finite blocklength regime. 

\begin{setting}\label{setting:1}
    Let $\mathcal{M} = \{1,2,\ldots,|\mathcal{M}|\}$ be the message set. The transmitter selects each message independently and uniformly over $\mathcal{M}$, and the associated message random variable is denoted by $M$.

For every $m \in \mathcal{M}$, the encoder $\mathtt{E}$ maps $m$ to a codeword $X^{n}(m) \in \mathcal{X}^{n}$, which is sent through a memoryless channel $W_{Y|X}$, producing the output sequence $Y^{n}$. The conditional distribution of $Y^{n}$ given $X^{n}=x^{n}$ is
\[
    p_{Y^{n}|X^{n}}(y^{n}| x^{n})
    = \prod_{i=1}^{n} W_{Y|X}(y_i| x_i).
\]
The decoder $\mathtt{D}$ observes $Y^{n}$ and attempts to recover the transmitted message, yielding an estimate $\hat{m}$. The corresponding random variable is denoted by $\hat{M}$, with realizations $\hat{m} \in \mathcal{M}$. The average probability of decoding error is assumed to be $\epsilon$, i.e.,
\[
    p_{M\hat{M}}(M \neq \hat{M}) = \epsilon.
\]
\end{setting}
In the next theorem, we derive a converse that links the parameters $n$, $\epsilon$, and $R$ to $\max_{p_X} I_{G,f}(X;Y)$ for the underlying memoryless channel.
\begin{theorem}
    Let $(G,f)$ be an admissible pair such that $G$ is an increasing convex function on $[0,\infty)$. If  $I_{G,f}$ is subadditive, under Setting~\ref{setting:1}, we have
    \begin{align}
        n \;\ge\;
        \frac{
            G\!\left(
                \frac{1}{|\mathcal{M}|} f\bigl(|\mathcal{M}|(1-\epsilon)\bigr)
                + \frac{|\mathcal{M}|-1}{|\mathcal{M}|}
                  f\!\left(\frac{|\mathcal{M}|\epsilon}{|\mathcal{M}|-1}\right)
            \right)
        }{
            \max_{p_X} I_{G,f}(X;Y)
        }.\label{lower:fano:bound}
    \end{align}
\end{theorem}

\begin{proof}
First, observe that the subadditivity of $(G,f)$-information (see Remark \ref{rmksubaddinf}) implies that
\begin{align}
    I_{G,f}(X^n;Y^n)
    &\le \sum_{i=1}^{n} I_{G,f}(X_i;Y_i)
      \le n \max_{p_X} I_{G,f}(X;Y). \label{subadditive:n}
\end{align}
Next, note that $M \to X^n \to Y^n \to \hat{M}$ forms a Markov chain, so by Lemma~\ref{4DPI:lem},
\begin{align*}
    I_{G,f}(M;\hat{M})
    &\le I_{G,f}(X^n;Y^n).
\end{align*}
This equation, along with  \eqref{subadditive:n}, and the Fano's inequality for $(G,f)$-information (Lemma \ref{lemmaFano}) yields the desired result. 
\end{proof}
In the following, we aim to identify input distributions that maximize $I_{G,f}(X;Y)$ for certain symmetric channels. We begin with a definition.

\begin{definition}
We call the channel $W_{Y|X}$ \emph{permutation invariant} if for every permutation $\pi$ on the input alphabet $\mathcal{X}$, there is a permutation $\pi'$ on $\mathcal{Y}$ such that $$W_{Y|X}(\pi'(y)|\pi(x))=W_{Y|X}(y|x), \qquad \forall x,y.$$ 
\end{definition}
\begin{example}
Any channel whose transition matrix has the form (where rows represent the input symbols and columns represent the output symbols)
\[
    \begin{bmatrix}
        a & b & b & \cdots & b & c & \cdots & c\\
        b & a & b & \cdots & b & c & \cdots & c\\
        \vdots & \vdots & \vdots & \ddots & \vdots & \vdots & \ddots & \vdots\\
        b & b & b & \cdots & a & c & \cdots & c
    \end{bmatrix}
\]
is permutation invariant. In particular, binary symmetric channels (BSC) and binary erasure channels (BEC) are permutation invariant.
\end{example}

\begin{proposition}
If the channel $W_{Y|X}$ is permutation invariant, the uniform distribution on $X$ maximizes $I_{G,f}(X;Y)$.     
\end{proposition}

\begin{proof}
For the given channel $W_{Y|X}$, define
\[
    \mathcal{U}(p_X,q_Y)
    = \sum_{x} p(x)\,
      G\!\left(D_f\bigl(W_{Y|X=\!x}\,\|\,q_Y\bigr)\right).
\]
Let
\[
    \varphi(p) \triangleq \min_{q} \mathcal{U}(p,q) = I_{G,f}(X;Y).
\]
Then we claim that for any permutation $\pi$ of the input alphabet,
    \[
        \varphi(\pi p) = \varphi(p).
    \]
    For a given $p$, let $q_p^*$ be a minimizer of $\mathcal{U}(p,q)$. Take some permutation $\pi$ on $\mathcal{X}$ and let $\pi'$ be the corresponding ``invariant" permutation on $\mathcal{Y}$.
    \begin{align}
        \varphi(\pi p)
        &= \min_q \mathcal{U}(\pi p,q)
         \le \mathcal{U}(\pi p,\pi' q_p^*) \\
        &= \mathcal{U}(p,q_p^*)
         = \varphi(p). \label{permu}
    \end{align}
    Thus $\varphi(\pi p) \le \varphi(p)$ for any permutation $\pi$. Hence, for any natural number $n\geq 1$ and using induction, we have $\varphi(\pi^n p) \le \varphi(\pi p)$. If we choose $n$ such that $\pi^n$ is the identity permutation, we get $\varphi(p) \le \varphi(\pi p)$. Thus, $\varphi(\pi p) = \varphi(p)$. 
    
   We claim that the uniform distribution maximizes $\varphi(p)$.   
   From Lemma~\ref{IGF:concave}, we know that $\varphi(p) = I_{G,f}(X;Y)$ is concave in $p$ for fixed $W_{Y|X}$. Let $p^{*}$ be a maximizer of $\varphi$. Then, by concavity and permutation invariance,
    \begin{align*}
        \varphi(p^{*})
        &\ge \varphi\!\left(\frac{1}{|\mathcal{X}|!}\sum_{\pi} \pi p^{*}\right)
         \ge \frac{1}{|\mathcal{X}|!}\sum_{\pi} \varphi(\pi p^{*})
         = \varphi(p^{*}).
    \end{align*}
    Hence equality must hold throughout. The distribution
    \(\frac{1}{|\mathcal{X}|!}\sum_{\pi} \pi p^{*}\) is uniform on the input alphabet, so the uniform distribution also maximizes $\varphi(p)$.
\end{proof}
\subsection{Binary Hypothesis Testing}

\begin{setting}\label{settingHT}
Consider the binary hypothesis testing problem. Under the null hypothesis $H_0$, the observations satisfy
\[
X^n \sim \prod_{i=1}^n p_X(x_i),
\]
while under the alternative hypothesis $H_1$, they satisfy
\[
X^n \sim \prod_{i=1}^n q_X(x_i).
\]
Let $U(X^n)\in\{0,1\}$ denote a binary decision rule. Define
\[
\alpha \triangleq \Pr[U=1| H_0],\qquad
\beta \triangleq \Pr[U=1| H_1].
\]
\end{setting}
\begin{theorem}
    Let $(G,f)$ be an admissible pair. If  $\mD_{G,f}$ is subadditive, under Setting~\ref{settingHT}, we have
    \begin{align}
    \mD_{G,f}\bigl(\mathrm{Bern}(\alpha)\|\mathrm{Bern}(\beta)\bigr)
    \;\leq\;
    n\,\mD_{G,f}(p_X\|q_X).
    \label{hypo:main}
\end{align}
\end{theorem}
\begin{proof}
    By the subadditivity property, we have
\[
\mD_{G,f}\!\Bigl(\prod_{i=1}^n p_X(x_i)\Big\| \prod_{i=1}^n q_X(x_i)\Bigr)
\;\leq\;
n\,\mD_{G,f}(p_X\|q_X).
\]

Considering the binary decision rule $U(X^n)\in\{0,1\}$,  the data processing inequality yields
\[
\mD_{G,f}\!\Bigl(\prod_{i=1}^n p_X(x_i)\Big\| \prod_{i=1}^n q_X(x_i)\Bigr)
\;\geq\;
\mD_{G,f}(p_U\|q_U).
\]
Since $U$ is binary, $(p_U,q_U)$ are Bernoulli distributions, $p_U = \mathrm{Bern}(\alpha)$ and $q_U = \mathrm{Bern}(\beta)$. Thus,
\[
\mD_{G,f}(p_U\|q_U)
=
\mD_{G,f}\bigl(\mathrm{Bern}(\alpha)\|\mathrm{Bern}(\beta)\bigr).
\]

Combining the two bounds yields \eqref{hypo:main}.
\end{proof}

\begin{definition}
  The $\epsilon$-optimal exponent in Stein’s regime is
$$V_\epsilon \stackrel{\triangle}{=}\sup\{E:\exists n_0,\forall n\geq n_0,\exists~ p_{U|X^n} s.t.~ \alpha =1-\epsilon ,\beta =\exp(-nE)\}.$$
and Stein’s exponent is defined as $V \stackrel{\triangle}{=} \lim_{\epsilon\rightarrow 0} V_{\epsilon}$.  
\end{definition}
\begin{theorem}\cite[p.~286]{Polyanskiy25}\label{poly:th}
  For any $\epsilon>0$  we have $V_\epsilon=D_{\text{KL}}(p_X||q_X)$. Thus, $V=D_{\text{KL}}(p_X||q_X)$.
\end{theorem}
We have the following theorem:

\begin{theorem}\label{comparison:log+}
    Let $G(x) = \log(1+x)$ and let $f$ be a convex function on $[0,\infty)$ with $f(1)=1$ such that $x^2 f''(x) \in \mathcal{T}^{+}$. Assume furthermore that $f(x)\geq c x^{s}$ for some $c>0$ and some $s>1$. Let $\hat{f}(x)=f(x)-1$. Then
    \begin{align}
        D_{\text{KL}}(p_X\|q_X)
        \;\leq\;
        \frac{\mD_{G,\hat f}(p_X\|q_X)}{s-1}.
    \end{align}
    Equivalently,
     \begin{align}
\sum_{x}q_X(x)f\left(\frac{p_X(x)}{q_X(x)}\right)\geq e^{(s-1)D_{\text{KL}}(p_X\|q_X)}.
    \end{align}
\end{theorem}

\begin{proof}
Fix $\epsilon>0$. For any exponent $E_\epsilon$ achievable in binary hypothesis testing between $p_X$ and $q_X$, there exist an integer $n_\epsilon$ and a decision rule $p_{U_\epsilon|X^n}$ such that for all $n \ge n_\epsilon$, under this rule the type-I and type-II error probabilities satisfy
\[
    \alpha = \Pr[U_\epsilon=1| H_0] = 1-\epsilon,
    \qquad
    \beta = \Pr[U_\epsilon=1| H_1] = e^{-nE_\epsilon}.
\]
For such $n$ and this decision rule, applying \eqref{hypo:main} gives
\begin{align}
    \mD_{G,\hat f}\bigl(\mathrm{Bern}(\alpha)\|\mathrm{Bern}(\beta)\bigr)
    \;\leq\;
    n\,\mD_{G,\hat f}(p_X\|q_X).\label{eq:22}
\end{align}
The equivalent form of \eqref{eq:22} is 
\begin{align}
     D_f(\alpha\|\beta)
    \;\leq\;
      (D_f(p_X\|q_X))^n.
\end{align}
On the other hand, for Bernoulli arguments we have
\[
    D_f(\alpha\|\beta)
    = \beta f\!\left(\frac{\alpha}{\beta}\right)
      + (1-\beta) f\!\left(\frac{1-\alpha}{1-\beta}\right).
\]
Using the lower bound $f(x)\ge c x^{s}$ and substituting $\alpha=1-\epsilon$, $\beta=e^{-nE_\epsilon}$ yields
\begin{align}
    D_f(\alpha\|\beta)
    &\ge
    c\,\beta \left(\frac{\alpha}{\beta}\right)^{s}
    = c\,e^{-n(1-s)E_\epsilon}(1-\epsilon)^{s}.
    \label{com:2}
\end{align}
Combining \eqref{eq:22} and \eqref{com:2} gives
\begin{align}
    c\,e^{-n(1-s)E_\epsilon}(1-\epsilon)^{s}
    \;\leq\;
    \bigl(D_f(p_X\|q_X)\bigr)^{n},
\end{align}
or, equivalently,
\begin{align}
    E_\epsilon
    &\le
    \frac{\log D_f(p_X\|q_X)}{(1-\epsilon)^{s}(s-1)}
    - \frac{\log c}{n(1-\epsilon)^{s}(s-1)}.
    \label{combine:1:hypo}
\end{align}
Since \eqref{combine:1:hypo} holds for every achievable exponent $E_\epsilon$, we obtain
\begin{align}
    V_\epsilon
    \;\le\;
    \frac{\log D_f(p_X\|q_X)}{(1-\epsilon)^{s}(s-1)}
    - \frac{\log c}{n(1-\epsilon)^{s}(s-1)},
    \label{combine:12:hypo}
\end{align}
Letting first $n\to\infty$ and then $\epsilon\to 0$ in \eqref{combine:12:hypo}, we obtain
\begin{align}
    V
    \;\le\;
    \frac{\log D_f(p_X\|q_X)}{s-1},
    \label{combine:11:hypo}
\end{align}
 Therefore, using Theorem~\ref{poly:th} 
\begin{align}
    V=D_{\text{KL}}(p_X\|q_X)
    \;\; 
    \;\le\;
    \frac{\log D_f(p_X\|q_X)}{s-1}=\frac{ \mD_{G,\hat f}(p_X\|q_X)}{s-1}.
\end{align}
This completes the proof.
\end{proof}
\begin{example}
The following functions satisfy the assumptions of Theorem~\ref{comparison:log+}:
\begin{enumerate}
    \item $f(x) = x^{s}$ for $s>1$, with $c = 1$.
    \item $f(x) = x^{s} - \dfrac{s}{\theta} x^{\theta} + \dfrac{s}{\theta}$ for $0<\theta<1$ and $1 < s \le 2\theta$, with $c = \tfrac{1}{2}$.
\end{enumerate}
\end{example}
\begin{lemma}\label{comparison:log-}
    Let $G(x) = -\log(1-x)$ and let $f$ be non-negative concave function on $[0,\infty)$ with $f(1)=1$ such that $-x^2 f''(x) \in \mathcal{T}^{-}$. Assume furthermore that $f(x)\leq c x^{s}$ for some $c>0$ and some $0<s<1$. Then
    \begin{align}
        D_{\text{KL}}(p_X\|q_X)
        \;\leq\;
        \frac{\mD_{G,f}(p_X\|q_X)}{1-s}.
    \end{align}
\end{lemma}
\begin{proof}
The proof is completely analogous to that of Lemma~\ref{comparison:log+}, and is therefore omitted.
\end{proof}
\begin{example}
The function $f(x)=x^s$ for $0<s<1$ with $c=1$ satisfy the assumptions of Lemma~\ref{comparison:log-}.
\end{example}
\subsection{Error exponent for the lower bound on the error probability of memoryless channels}\label{subsec:expo}

In this section, we extend the classical Shannon–Gallager–Berlekamp  error exponent \cite{Shannon1967} to the class of function pairs $(G,f)$ for which the generalized divergence $\mD_{G,f}$ is subadditive. As we will see, subadditivity plays a key role by enabling a tensorization of the single–letter exponent expression to blocklength $n$. 
\begin{setting}\label{set3} Consider Setting \ref{setting:1}. Let $R=(1/n)\log|\mathcal{M}|$. 
    The maximal and average error probabilities are defined as
    \begin{align}
        P_{e}^{\max} &= \max_{m \in \mathcal{M}} \Pr(M \neq \hat{M} | M = m),\\
        P_{e} &= \frac{1}{|\mathcal{M}|} \sum_{m \in \mathcal{M}} \Pr(M \neq \hat{M} | M = m).
    \end{align}
    We assume that the transition matrix $W_{Y|X}$ is strictly positive and define
    \[
        W_{\min} \triangleq \min_{x \in \mathcal{X},\,y \in \mathcal{Y}} W_{Y|X}(y|x),
    \]
    so that $W_{\min} > 0$.

Define $\mathscr{E}(R)$ be the (maximal) error exponent by
\begin{align}
    \mathscr{E}(R)
    \triangleq \limsup_{\text{codes}, n\to\infty, } \frac{-\log P_{e}^{\max}}{n}.
\end{align}

    \end{setting}
    
    \begin{remark}
It is known (see, e.g., \cite[p.~32]{Shannon1967}) that the error exponent functions corresponding to maximal and average error probabilities coincide.
\end{remark}

    \begin{theorem}\label{main:313:theoremn}
    Let $b>a>0$ be arbitrary positive numbers. 
Let $f_s(x) = x^{s}\psi_{s}(\log x)$ with $s\in(0,1)$, where $\psi_{s}$ is function satisfying 
$\psi_{s}(x)\in[a,b]$,  
$\psi_{s}(0)=1$, and  $-x^{2}f_s''(x)\in\mathcal{T}^{-}$  for all $s\in(0,1)$. Under Setting~\ref{set3}, for any positive rate $R$, the maximal error probability satisfies the lower bound
\begin{align}
   P_{e}^{\max}
   >
   \exp\Bigl\{
      -n\Bigl(
         \mathtt{E}_{f-\mathrm{sp}}\bigl(R-\mathcal{O}(\tfrac{1}{n})\bigr)
         + \mathcal{O}(\tfrac{1}{\sqrt{n}})
      \Bigr)
   \Bigr\},
\end{align}
where
\begin{align}
  \mathtt{E}_{f-\mathrm{sp}}(R)
  \triangleq
 \sup_{p_X}\sup_{s\in(0,1)} 
  \Biggl[
    \frac{I_{G,\widetilde f_s}(X;Y)}{1-s}
    - \frac{s}{1-s}\,R
  \Biggr].
  \label{expo:13:1n}
\end{align}
where $G(x) = -\log(1-x)$ and $\widetilde f_s(x)=x(1-f_s(\frac1x))$.
Consequently,
\begin{align}
    \mathscr{E}(R) \le \mathtt{E}_{f-\mathrm{sp}}(R).
    \label{ineq:fi:1}
\end{align}

\end{theorem}
Concretely speaking, the above theorem can be utilized, if we can find a function family $\psi_{s} : \mathbb{R} \to [a,b]$ satisfying 
\begin{enumerate}
    \item $\psi_s(0)=1$,
    \item $\phi_{s}(x)\ge 0$ for all $x\in\mathbb{R}$, where
    \[
       \phi_{s}(x)\triangleq
       \psi_{s}(x)
       + \frac{1-2s}{s(1-s)}\,\psi_{s}'(x)
       - \frac{1}{s(1-s)}\,\psi_{s}''(x).
    \]
    \item For all $x,y\in\mathbb{R}$,
    \[
      \phi_{s}(x)
      + \frac{1-2s}{s(1-s)}\,\phi_{s}'(x)
      - \frac{1}{s(1-s)}\,\phi_{s}''(x)
      \;\ge\;
      \phi_{s}(y)\,\phi_{s}(x-y).
    \]
\end{enumerate}

\begin{remark}
Although we have theoretically extended the Shannon–Gallager–Berlekamp sphere–packing bound to the generalized divergence in Theorem~\ref{main:313:theoremn}, concrete examples beyond the classical choice $f_s(x)=x^{s}$ are not yet known. In particular, identifying functions $f_s(x)\neq x^{s}$ that satisfy all the assumptions of Theorem~\ref{main:313:theoremn} remains an open problem. 
\end{remark}

Proof of Theorem~\ref{main:313:theoremn} is given in Appendix~\ref{er:expo:proof}.

\begin{remark}
Note that, for non negative functions of the form
\[
f_s(x) = \widetilde{a}x+\widetilde{b}+ x^{s}\psi_{s}(\log x),
\]
with $\widetilde{a}\cdot\widetilde{b}=0$ and satisfying $-x^2 f_{s}''(x)\in \mathcal{T}^{-}$, one can derive alternative expressions (by a suitable adaptation of the proof of Theorem~\ref{main:313:theoremn}) that differ from \eqref{expo:13:1n} over certain ranges of rates. However, our calculations yielded exponents that were strictly weaker than the classical sphere-packing bound \cite{Shannon1967}.

\end{remark}

\section*{Acknowledgment}
The authors would like to thank Mr.~Thana Somsirivattana and Man Hon (Harry) Wong. 
In the early stages of this work, they obtained preliminary results related to the 
subadditivity of $\tilde{I}_f$ as defined in \eqref{tIfXYeq}, addressing the 
subadditivity problem posed in \cite[Section~7.8]{Polyanskiy25}. 
Although their results concerned a different notion of information measure, their 
insights were helpful to our thinking throughout this project.

\bibliographystyle{IEEEtran}
\bibliography{biblioarx}
\appendix

\section{Proof of Lemmas \ref{domain:lem}--\ref{lemmaFano}}\label{lemma:1--3}
This Appendix contains the proofs for Lemma \ref{domain:lem} through Lemma \ref{lemmaFano}.
\subsection{Proof of Lemma~\ref{domain:lem}}\label{domain:lem:proof}
\begin{proof}

    We have
    \begin{align*}
        D_f(p_X\|q_X) = \mathbb{E}_{X\sim q_X}[f(Z)],
    \end{align*} 
    where $Z = p_X(X)/q_X(X)$. Observe that $\mathbb{E}_{X\sim q_X}[Z] = 1$. Thus,
    \begin{align}
        D_f(p_X\|q_X) \leq \sup_{Z\geq 0:\,\mathbb{E}[Z]=1}\mathbb{E}[f(Z)]. \label{eq:main}
    \end{align} 
    The supremum in \eqref{eq:main} is attained (or approached) by a binary random variable $Z$. This follows from Carathéodory's theorem and standard cardinality bound reduction techniques. 
    
    Assume that $Z$ takes values $z_0 \in[0,1]$ and $z_1 \geq 1$ with probabilities $p$ and $1-p$, respectively, where $pz_0 + (1-p)z_1 = 1$. Fix $p$, and view $pf(z_0) + (1-p)f(z_1)$ in terms of $z_0$ where $z_1 = (1-pz_0)/(1-p)$. From the convexity of the objective $pf(z_0) + (1-p)f(z_1)$ in $z_0$, the maximum occurs at the boundary, i.e., when $z_0 = 0$ or $z_0=1$. The case $z_0=1$ yields value zero. So, the maximizer is $z_0=0$.
    
    To sum this up, the supremum in \eqref{eq:main} is attained (or approached) by two-point random variables of the form:
    \begin{equation*}
       Z_t := \begin{cases}
              0 & \text{with probability } 1-t, \\[4pt]
              \frac{1}{t} & \text{with probability } t,
            \end{cases}
            \qquad t \in (0,1].
    \end{equation*}
    Substituting this back into the expectation:
    \begin{align}
        D_f(p_X\|q_X) \leq \sup_{t\in(0,1]} \left[ (1-t)f(0) + t f\left(\tfrac{1}{t}\right) \right].
    \end{align} 
    Since $f$ is convex, the perspective transform $g(t) = t f(1/t)$ is convex in $t$. The term $(1-t)f(0)$ is linear and thus also convex. Therefore, the supremum of the convex sum occurs at the boundary points of the interval $[0,1]$.
    
    At $t=1$, the expression is $1 \cdot f(1) = 0$ (since $f(1)=0$ for standard $f$-divergences). Therefore, the supremum must occur as $t$ converges to $0$:
    \begin{align*}
        \lim_{t \to 0} \left[ (1-t)f(0) + t f\left(\tfrac{1}{t}\right) \right] = f(0) + \lim_{u \to \infty} \frac{f(u)}{u}.
    \end{align*}
\end{proof}
\subsection{Proof of Lemma \ref{IGF:concave}}\label{lemma:1:new}
\begin{proof}
    \begin{enumerate} 
        \item By definition,
        \[
            I_{G,f}(X;Y)
            = \min_{q_Y} \sum_{x\in \mathcal{X}} p_X(x)\,
                G\!\left(D_f\bigl(p_{Y|X=x}\,\|\,q_Y\bigr)\right).
        \]
        For a fixed channel $p_{Y|X}$ and fixed $q_Y$, the expression inside the minimum is linear in $p_X$. Hence $I_{G,f}(X;Y)$ is the pointwise minimum of a family of linear functions of $p_X$, which is a concave functional. 
        \item Let $q_Y^{*}$ be a minimizer in the definition of $I_{G,f}(X;Y)$. Then, for all $x$, the data processing property of $f$-divergence \cite[p.~120]{Polyanskiy25} implies that $$D_f\bigl(p_{Y|X=\!x}\,\|\,q_Y^{*}\bigr)\geq D_f\bigl(p_{Z|X=\!x}\,\|\,q_Z^{*}\bigr)$$
        where $q_Z^{*}$ is the image of $q_Y^{*}$ under the channel $p_{Z|Y}$, i.e.,
        \[
            q_Z^{*}(z) = \sum_{y} q_Y^{*}(y)\,p_{Z|Y}(z|y).
        \]
        Then,
        \begin{align}
            I_{G,f}(X;Y)
            &= \mathbb{E}_{p_X}\!\left[
                   G\!\left(D_f\bigl(p_{Y|X=\!x}\,\|\,q_Y^{*}\bigr)\right)
               \right] \\
            &\ge \mathbb{E}_{p_X}\!\left[
                   G\!\left(D_f\bigl(p_{Z|X=\!x}\,\|\,q_Z^{*}\bigr)\right)
               \right] \\
            &\ge \min_{q_Z}
               \mathbb{E}_{p_X}\!\left[
                   G\!\left(D_f\bigl(p_{Z|X=\!x}\,\|\,q_Z\bigr)\right)
               \right] \\
            &= I_{G,f}(X;Z),
        \end{align}
        where the first inequality uses monotonicity of $G$ and of the expectation operator, and the second inequality follows by optimizing over $q_Z$. 
        \item Equation \eqref{DPI:Igfn1} follows from part 2 because $X\rightarrow (Y,Z)\rightarrow Y$ and $X\rightarrow Y\rightarrow (Y,Z)$ are Markov chains.
To prove \eqref{DPI:Igfn2}, since $X \rightarrow Y \rightarrow Z$, we observe $p_{Z|XY}=p_{Z|Y}$, thus for any  $q_Z$ we have
\begin{align*}
    \mathbb{E}_{p_{XY}}\!\left[
                  \mD_{G,f}\bigl(p_{Z|XY}\,\|\,q_Z\bigr)
               \right]=
               \mathbb{E}_{p_Y}\!\left[
                   \mD_{G,f}\bigl(p_{Z|Y}\,\|\,q_Z\bigr)
               \right]
\end{align*}
Taking minimum over $q_Z$ proves \eqref{DPI:Igfn2}.
    \end{enumerate}
\end{proof}
\subsection{Proof of Lemma \ref{4DPI:lem}}\label{lemma:2:new}
\begin{proof}
    Lemma \ref{IGF:concave} implies that
\begin{align}
    I_{G,f}(A;B) \leq I_{G,f}(A;Y).
\end{align}
Therefore, it suffices to show that
\begin{align}
    I_{G,f}(A;Y) \leq I_{G,f}(X;Y).
\end{align}

Let $q_Y^*$ be a minimizer in the definition of $I_{G,f}(X;Y)$.
Define
\[
    V(a) \triangleq D_f\bigl(p_{Y|A=a} \,\|\, q_Y^*\bigr).
\]
Then
\begin{align}
    V(a)
    &= D_f\bigl(p_{Y|A=a} \,\|\, q_Y^*\bigr)
      = D_f\Bigl(\mathbb{E}_{p_{X|A=a}}\bigl[p_{Y|X}\bigr] \,\Big\|\, q_Y^*\Bigr)\\
    &\overset{(a)}{\leq} \mathbb{E}_{p_{X|A=a}}\Bigl[
           D_f\bigl(p_{Y|X} \,\|\, q_Y^*\bigr)
        \Bigr]
\end{align}
where $(a)$ uses the joint convexity of $D_f$ in its arguments\cite[p.~120]{Polyanskiy25}.

Applying $G$ to both sides and using its monotonicity and convexity yields
\begin{align}
    G\bigl(V(a)\bigr)
    &\overset{(a)}{\leq}
       G\Bigl(\mathbb{E}_{p_{X|A=a}}\bigl[D_f\bigl(p_{Y|X} \,\|\, q_Y^*\bigr)\bigr]\Bigr) \\
    &\overset{(b)}{\leq}
       \mathbb{E}_{p_{X|A=a}}\Bigl[
           G\bigl(D_f\bigl(p_{Y|X} \,\|\, q_Y^*\bigr)\bigr)
       \Bigr], \label{4DPI:b:last}
\end{align}
where $(a)$ uses that $G$ is increasing, and $(b)$ uses Jensen's inequality based on the convexity of $G$.

Taking expectation with respect to $p_A$ on both sides of \eqref{4DPI:b:last}, we obtain
\begin{align}
   \mathbb{E}_{p_A}\Bigl[G\bigl(D_f(p_{Y|A}\,\|\,q_Y^*)\bigr)\Bigr]
   &= \mathbb{E}_{p_A}\bigl[G(V(A))\bigr]\\
   &\leq \mathbb{E}_{p_A}\mathbb{E}_{p_{X|A}}\Bigl[
          G\bigl(D_f(p_{Y|X}\,\|\,q_Y^*)\bigr)
       \Bigr]\\
   &= \mathbb{E}_{p_X}\Bigl[
          G\bigl(D_f(p_{Y|X}\,\|\,q_Y^*)\bigr)
       \Bigr]
    = I_{G,f}(X;Y).
\end{align}
Therefore,
\begin{align}
    I_{G,f}(A;Y)
    &= \min_{q_Y}
       \mathbb{E}_{p_A}\Bigl[
          G\bigl(D_f(p_{Y|A}\,\|\,q_Y)\bigr)
       \Bigr]\\
    &\leq \mathbb{E}_{p_A}\Bigl[
            G\bigl(D_f(p_{Y|A}\,\|\,q_Y^*)\bigr)
         \Bigr]
     = I_{G,f}(X;Y),
\end{align}
which proves the desired data processing inequality.
\end{proof}
\subsection{Proof of Lemma \ref{lemmaFano}}\label{lemma:3:new}
\begin{proof}
    Let $q_{\hat{M}}^{*}$ be a minimizer in the definition of $I_{G,f}(M;\hat{M})$, and define the error indicator $E = \mathbf{1}_{\{M = \hat{M}\}}$. Under $p_{M\hat{M}}$ and $p_M q_{\hat{M}}^{*}$, the distributions of $E$ are $\mathrm{Ber}(1-\epsilon)$ and $\mathrm{Ber}\!\bigl(\tfrac{1}{|\mathcal{M}|}\bigr)$, respectively. By definition,
\begin{align*}
    I_{G,f}(M;\hat{M})
    &= \mathbb{E}_{p_M}\!\left[
           G\!\left(D_f\bigl(p_{\hat{M}|M}\,\|\,q_{\hat{M}}^{*}\bigr)\right)
       \right] \\
    &\overset{(a)}{\ge}
       G\!\left(D_f\bigl(p_{M\hat{M}}\,\|\,p_M q_{\hat{M}}^{*}\bigr)\right) \\
    &\overset{(b)}{\ge}
       G\!\left(
           D_f\bigl(\mathrm{Ber}(1-\epsilon)\,\big\|\,\mathrm{Ber}(\tfrac{1}{|\mathcal{M}|})\bigr)
       \right),
\end{align*}
where $(a)$ follows from the convexity of $G$, and $(b)$ uses the data processing inequality for $f$-divergences applied to the mapping $(M,\hat{M}) \mapsto E$. A direct calculation of the $f$-divergence between these Bernoulli distributions gives
\begin{align*}
    D_f\bigl(\mathrm{Ber}(1-\epsilon)\,\big\|\,\mathrm{Ber}(\tfrac{1}{|\mathcal{M}|})\bigr)
    &= \frac{1}{|\mathcal{M}|}\,
       f\bigl(|\mathcal{M}|(1-\epsilon)\bigr)
       + \frac{|\mathcal{M}|-1}{|\mathcal{M}|}\,
         f\!\left(\frac{|\mathcal{M}|\epsilon}{|\mathcal{M}|-1}\right),
\end{align*}
which yields the desired bound.
\end{proof}

\section{Proof of Lemma  \ref{gen-lemma1} and Lemma  \ref{func} }
\subsection{Proof of Lemma  \ref{gen-lemma1}}\label{lem:1:pro}
\begin{proof}
    We show that we can restrict the alphabet of $\mathcal{Y}$ to size two; the proof for $\mathcal{Z}$ is similar. 
    We can write subadditivity as
    \begin{align*}
G\left(\sum_{y}r_Y(y)f\left( \frac{q_Y(y)}{r_Y(y)}\right)\right)+G\left(\sum_{z}r_Z(z)f\left( \frac{q_Z(z)}{r_Z(z)}\right)\right)\geq  G\left(\sum_{y,z}r_Y(y)r_Z(z)f\left( \frac{q_Y(y)q_Z(z)}{r_Y(y)r_Z(z)}\right)\right) .
    \end{align*}
Since $G$ is increasing, taking $G^{-1}$ from both sides of the inequality obtains
\begin{align}\label{inv}
    G^{-1}\left(G\left(\sum_{y}r_Y(y)f\left( \frac{q_Y(y)}{r_Y(y)}\right)\right)+G\left(\sum_{z}r_Z(z)f\left( \frac{q_Z(z)}{r_Z(z)}\right)\right)\right)\geq  \sum_{y,z}r_Y(y)r_Z(z)f\left( \frac{q_Y(y)q_Z(z)}{r_Y(y)r_Z(z)}\right)
\end{align}
    Fix some arbitrary $r^*_Y(y), q^*_Y(y), r^*_Z(z), q^*_Z(z)$ such that $r^*_Y(y),r^*_Z(z)>0$ for all $y,z$. Take some arbitrary distribution $r_Y(y)$ such that 
    \begin{align}
       \sum_y \frac{q^*_Y(y)}{r^*_Y(y)}r_Y(y)=1. \label{gen-eqn2}
    \end{align}
    Note that \eqref{gen-eqn2} is indeed a distribution. Also, let 
    $$q_Y(y)=\frac{q^*_Y(y)}{r^*_Y(y)}r_Y(y).$$
    Note that when $r_Y(y)=r_Y^*(y)$, we get $q_Y(y)=q_Y^*(y)$. 
    Thus, \eqref{inv} implies 
    \begin{align}\label{inv2}G^{-1}\left(G\left(\sum_{y}r_Y(y)f\left( \frac{q^*_Y(y)}{r^*_Y(y)}\right)\right)+G\left(\sum_{z}r^*_Z(z)f\left( \frac{q^*_Z(z)}{r^*_Z(z)}\right)\right)\right)&\ge\sum_{y,z}r_Y(y)r^*_Z(z)f\left( \frac{q_Y(y)q^*_Z(z)}{r_Y(y)r^*_Z(z)}\right) 
    \end{align}
Since in the right-hand-side of \eqref{inv2}, the first term is linear in $r_Y$ and the second term is constant, we can rewrite \eqref{inv2} as 
\begin{align}\label{inv3}
G^{-1}\left(G(x)+G(c)\right)- \sum_{y,z}r_Y(y)r^*_Z(z)f\left( \frac{q^*_Y(y)q^*_Z(z)}{r^*_Y(y)r^*_Z(z)}\right) \ge 0 
\end{align}
where
$$x=\sum_{y}r_Y(y)f\left( \frac{q^*_Y(y)}{r^*_Y(y)}\right)$$
$$c=\sum_{z}r^*_Z(z)f\left( \frac{q^*_Z(z)}{r^*_Z(z)}\right)$$
Consider minimizing  \eqref{inv3} over $r_Y(y)$.
Note that $x$ (as defined above) and $\sum_{y,z}r_Y(y)r^*_Z(z)f\left( \frac{q_Y(y)q^*_Z(z)}{r_Y(y)r^*_Z(z)}\right)$ are linear in $r_Y$. Thus, the concavity of $G^{-1}\left(G(x)+G(c)\right) $ implies that the minimum value of \eqref{inv3} occurs when $r_Y$ is a vertex of its domain. The constraints on $r_Y(y)$ are $ r_Y(y)\geq 0$ and $\sum_{y}r_Y(y)=1$ as well as \eqref{gen-eqn2}. Any vertex of the domain must satisfy  $|\mathcal{Y}|$ equations with equality, i.e., it must satisfy $ r_Y(y)=0$ for $|\mathcal{Y}|-2$ elements $y$. In other words, the support set of vertices of the domain is binary. This shows that establishing \eqref{inv3} for binary $Y$ suffices to establish it for non-binary $Y$ as well.
\end{proof}
\subsection{Proof of Lemma  \ref{func}}\label{lem:2:pro}
\begin{proof}
 Let $\varphi(x)=G^{-1}(G(x)+G(c))$. Since $ G $ is strictly increasing and differentiable, $ G^{-1} $ exists and is differentiable. Differentiating both sides of  
\[
G(\varphi(x)) = G(x) + G(c)
\]  
with respect to $ x $ yields  
\[
G'(\varphi(x)) \cdot \varphi'(x) = G'(x),
\]  
so that  
\[
\varphi'(x) = \frac{G'(x)}{G'(\varphi(x))}.
\]  
Differentiating again, we obtain  
  
\[
\varphi''(x) = \frac{G''(x) G'(\varphi(x)) - G'(x) \cdot G''(\varphi(x)) \cdot \varphi'(x)}{[G'(\varphi(x))]^2}.
\]  
Substituting $ \varphi'(x) = \frac{G'(x)}{G'(\varphi(x))} $, we have  
\[
\varphi''(x) = \frac{G''(x) [G'(\varphi(x))]^2 - [G'(x)]^2 G''(\varphi(x))}{[G'(\varphi(x))]^3}.
\]  
Since $ G'(\varphi(x)) > 0 $, the sign of $ \varphi''(x) $ is determined by the numerator. Thus, $ \varphi(x) $ is concave if and only if  
\[
G''(x) [G'(\varphi(x))]^2 \leq [G'(x)]^2 G''(\varphi(x)).
\]  
Let $ y = \varphi(x) $. 
Since $G(c) > 0$ and $G^{-1}(x)$ is increasing on $[0,\infty)$, we have
\[
G^{-1}\bigl(G(x)+G(c)\bigr) \geq G^{-1}\bigl(G(x)\bigr) = x.
\]
Thus, $y \geq x$. Dividing both sides by $ [G'(x)]^2 [G'(y)]^2 $, we obtain  
\[
\frac{G''(x)}{[G'(x)]^2} \leq \frac{G''(y)}{[G'(y)]^2} \quad \text{for all } x \leq y.
\]  
Define $ u(x) = \frac{1}{G'(x)} $. Then $ u'(x) = -\frac{G''(x)}{[G'(x)]^2} $, and the inequality becomes  
\[
-u'(x) \leq -u'(y) \quad \Longleftrightarrow \quad u'(x) \geq u'(y) \quad \text{for all } x \leq  y.
\]  
This implies that $ u'(x) $ is decreasing, so $ u(x) $ is concave. Conversely, if $ u(x) $ is concave, then $ u'(x) $ is decreasing, and the above inequalities hold, implying $ \varphi''(x) \leq 0 $. Hence, $ \varphi(x) $ is concave for all $ c > 0 $.
\end{proof}
\subsection{Proof of Lemma  \ref{lem:root-counting}}\label{lem:3:proof}
\begin{proof}
    Let $(q_Z^*, r_Z^*)$ be arbitrary distributions on $\mathcal{Z}$. We aim to minimize the subadditivity gap function over all pairs $(q_Y, r_Y)$ on any finite alphabet $\mathcal{Y}$. Any distribution pair $(q_Y, r_Y)$ can be fully characterized by the probability mass function $r_Y$ and the likelihood ratio function $t: \mathcal{Y} \to [0, \infty)$, defined by $t_y = q_Y(y)/r_Y(y)$.
    
    The subadditivity gap function depends on the pair $(q_Y, r_Y)$ only through expectations of the form $\sum_y r_Y(y) f(t_y)$. Specifically, we write the gap as:
    \begin{align}
        \mathcal{J}(r_Y, t_y) \triangleq G\left(\sum_{y} r_Y(y) f(t_y)\right) + G(C_Z) - G\left(\sum_{y} r_Y(y) \beta(t_y) \right),\label{eqgapm}
    \end{align}
    where $C_Z = \sum_z r_Z^*(z) f(q_Z^*(z)/r_Z^*(z))$ is a constant and $\beta(\tau) = \sum_{z} r_Z^*(z) f\big(\tau \frac{q_Z^*(z)}{r_Z^*(z)}\big)$.
    
    Crucially, observe that if multiple symbols $y_1, y_2$ share the same likelihood ratio $t_{y_1} = t_{y_2} = \tau$, they can be merged into a single symbol $y'$ with mass $r_{y'} = r_{y_1} + r_{y_2}$ and ratio $\tau$. This operation preserves the values of the constraints ($\sum r_y = 1, \sum r_y t_y = 1$) and the objective terms ($\sum r_y f(t_y), \sum r_y \beta(t_y)$). Consequently, the value of the gap $\mathcal{J}$ remains unchanged. Thus, without loss of generality, we can restrict our search to distributions where every active symbol corresponds to a distinct likelihood ratio.

Fix $t_y,q^*_Z,r^*_Z$ and view the gap \eqref{eqgapm} in terms of $r_Y$. 
    Let $\hat{r}_Y$ be an optimal distribution that minimizes the gap. We define the Lagrangian 
    \[
    \mathcal{L} = \mathcal{J} - \mu_1 \left(\sum_y r_Y(y) - 1\right) - \mu_2 \left(\sum_y r_Y(y) t_y - 1\right).
    \] 
 The constraints on $r_Y$ are linear equalities and non-negativity inequalities ($r_Y(y) \ge 0$). Since $\hat{r}_Y$ is a valid distribution satisfying these constraints, the feasible set is non-empty. Therefore, the Linearity Constraint Qualification (LCQ) slackness condition is satisfied, which guarantees that the Lagrange multipliers $\mu_1, \mu_2$ exist.
    For any symbol $y$ in the support of $\hat{r}_Y$ (i.e., where $\hat{r}_Y(y) > 0$), the partial derivative of the Lagrangian with respect to $r_Y(y)$ must vanish (by complementary slackness). Computing this derivative yields
    \begin{align}
        G'(\sigma_1) f(t_y) - G'(\sigma_2) \beta(t_y) - \mu_1 - \mu_2 t_y = 0,
    \end{align}
    where $\sigma_1$ and $\sigma_2$ represent the arguments of the first and third $G$ terms in $\mathcal{J}$, evaluated at the optimum. Since $G$ is strictly increasing, $G'(\sigma_2) > 0$. We can therefore divide the equation by $G'(\sigma_2)$ and define the constants $\lambda = G'(\sigma_1)/G'(\sigma_2)$, $A = \mu_1/G'(\sigma_2)$, and $B = \mu_2/G'(\sigma_2)$. Rearranging the terms, we find that the likelihood ratio $t_y$ of any active symbol must satisfy the condition
    \begin{align}\label{eq:direct_kkt}
        \lambda f(t_y) - \beta(t_y) = a + b t_y.
    \end{align}
    Recognizing that $\beta(t_y)$ is the expectation term $\mathbb{E}_Z[f(t_y q_Z^*/r_Z^*)]$, the left-hand side is precisely the function $H(t_y; \lambda)$ defined in the statement of the Lemma. Consequently, the necessary condition for optimality is $H(t_y; \Lambda) = a + b t_y$.

    By the hypothesis of the Lemma, the equation $H(t; \lambda) = a + b t$ has at most $k$ distinct solutions for $t$. Let the set of these roots be $\mathcal{T}_{roots} = \{\tau_1, \dots, \tau_m\}$ where $m \le k$. Equation \eqref{eq:direct_kkt} implies that for any active symbol $y$, the associated likelihood ratio $t_y$ must belong to $\mathcal{T}_{roots}$.
    
    Since every active symbol must map to one of these $m$ roots, and symbols sharing the same root can be merged into a single representative, the optimal gap value is achievable by a distribution with support size $m \le k$. Therefore, if the subadditivity inequality holds for all distributions of support size $k$, it necessarily holds for distributions of any arbitrary finite support size.
\end{proof}

\section{Proof of Theorem \ref{Th:main:G:x}}
\label{app:main:G:x}
\begin{proof}
 From Lemma~\ref{gen-lemma1}–\ref{func} and Corollary~\ref{corol:g}, it suffices to prove inequality~\eqref{gen-def-1} for binary distributions $r_Y,q_Y,r_Z,q_Z$. In other words, we may assume $\mathcal{Y}=\mathcal{Z}=\{0,1\}$. Let
\begin{align}
  &r_Y(Y=0)=x,\quad r_Y(Y=1)=1-x,\label{r_Y}\\
  &r_Z(Z=0)=r,\quad r_Z(Z=1)=1-r,\label{r_Z}\\
  &q_Y(Y=0)=y,\quad q_Y(Y=1)=1-y,\label{q_Y}\\
  &q_Z(Z=0)=s,\quad q_Z(Z=1)=1-s\label{q_Z},
\end{align}
where $x,y,r,s\in[0,1]$. Define the function $M_{x,r,s}(y)$ by
\begin{align*}
  M_{x,r,s}(y)
  \triangleq\;&
    x r\,f\!\left(\frac{y s}{x r}\right)
  + x(1-r)\,f\!\left(\frac{y(1-s)}{x(1-r)}\right)
  + (1-x)r\,f\!\left(\frac{(1-y)s}{(1-x)r}\right)
  + (1-x)(1-r)\,f\!\left(\frac{(1-y)(1-s)}{(1-x)(1-r)}\right)\\
  &{}- x\,f\!\left(\frac{y}{x}\right)
  - (1-x)\,f\!\left(\frac{1-y}{1-x}\right)
  - r\,f\!\left(\frac{s}{r}\right)
  - (1-r)\,f\!\left(\frac{1-s}{1-r}\right).
\end{align*}
Then inequality~\eqref{gen-def-1} is equivalent to $M_{x,r,s}(y)\le 0$ for all $x,y,r,s\in[0,1]$. Fix $x,r,s\in[0,1]$ and regard $M_{x,r,s}(y)$ as a function of $y$. We will show that
\[
  M_{x,r,s}(y)\big|_{y=x} = 0,\qquad
  M_{x,r,s}'(y)\big|_{y=x} = 0,
\]
and that the second derivative satisfies $M_{x,r,s}''(y)\le 0$ for all $y\in[0,1]$. Since $x,r,s$ are arbitrary in $[0,1]$ and $M_{x,r,s}$ is concave in $y$ with a critical point at $y=x$, this implies that $y=x$ is a global maximizer and hence $M_{x,r,s}(y)\le 0$ for all $x,y,r,s\in[0,1]$.
 For $M_{x,r,s}(y)\big|_{y=x}$ we have
 \begin{align*}
  M_{x,r,s}(y)\big|_{y=x}&=
    x r\,f\!\left(\frac{ s}{ r}\right)
  + x(1-r)\,f\!\left(\frac{(1-s)}{(1-r)}\right)
  + (1-x)r\,f\!\left(\frac{s}{r}\right)
  + (1-x)(1-r)\,f\!\left(\frac{(1-s)}{(1-r)}\right)\\
  &{}- x\,f\!\left(1\right)
  - (1-x)\,f\!\left(1\right)
  - r\,f\!\left(\frac{s}{r}\right)
  - (1-r)\,f\!\left(\frac{1-s}{1-r}\right)\\
  &=r(x+1-x)f\left(\frac{s}{r}\right)+(1-r)(x+1-x)f\left(\frac{1-s}{1-r}\right)-rf\left(\frac{s}{r}\right)-(1-r)f\left(\frac{1-s}{1-r}\right)\\
  &rf\left(\frac{s}{r}\right)+(1-r)f\left(\frac{1-s}{1-r}\right)-rf\left(\frac{s}{r}\right)-(1-r)f\left(\frac{1-s}{1-r}\right)=0.
 \end{align*}
 For $M'_{x,r,s}(y)\big|_{y=x}$ we get
 \begin{align*}
  M'_{x,r,s}(y)\big|_{y=x}&=
    s\,f'\!\left(\frac{ s}{ r}\right)
  + (1-s)\,f'\!\left(\frac{(1-s)}{(1-r)}\right)
  -s\,f'\!\left(\frac{s}{r}\right)
  -(1-s)\,f'\!\left(\frac{(1-s)}{(1-r)}\right)-\,f'\!\left(1\right)
  +\,f'\!\left(1\right)=0.
 \end{align*}
  For $M''_{x,r,s}(y)$ we get
 \begin{align*}
     &M''_{x,r,s}(y)=\frac{s^2}{xr}f''\left(\frac{ys}{xr}\right)+\frac{(1-s)^2}{x(1-r)}f''\left(\frac{y(1-s)}{x(1-r)}\right)-\frac{1}{x}f''\left(\frac{y}{x}\right)
     \\&+
     \frac{s^2}{(1-x)r}f''\left(\frac{(1-y)s}{(1-x)r}\right)+\frac{(1-s)^2}{(1-x)(1-r)}f''\left(\frac{(1-y)(1-s)}{(1-x)(1-r)}\right)-\frac{1}{1-x}f''\left(\frac{1-y}{1-x}\right)
 \end{align*}
  Letting $g(x)=x^2f''(x)$, we can write the above as
   \begin{align*}
     &M''_{x,r,s}(y)=\frac{x}{y^2}\left(rg\left(\frac{ys}{xr}\right)+(1-r)g\left(\frac{y(1-s)}{x(1-r)}\right)-g\left(\frac{y}{x}\right)\right)
     \\&+
     \frac{1-x}{(1-y)^2}\left(rg\left(\frac{(1-y)s}{(1-x)r}\right)+(1-r)g\left(\frac{(1-y)(1-s)}{(1-x)(1-r)}\right)-g\left(\frac{1-y}{1-x}\right)\right)
 \end{align*}
  Jensen's inequality and concavity of $g$ imply that $$rg\left(\frac{ys}{xr}\right)+(1-r)g\left(\frac{y(1-s)}{x(1-r)}\right)-g\left(\frac{y}{x}\right)\leq 0$$
  and
  $$rg\left(\frac{(1-y)s}{(1-x)r}\right)+(1-r)g\left(\frac{(1-y)(1-s)}{(1-x)(1-r)}\right)-g\left(\frac{1-y}{1-x}\right)\leq 0.$$
 Therefore, $M_{x,r,s}''(y)\leq 0$. Consequently, if $g(x) = x^2 f''(x)$ is concave, inequality~\eqref{gen-def-1} holds when $G(x)=x$. This completes the proof.
\end{proof}

\section{Proof of Theorem \ref{Th:main:G:log+} and Properties of $\mathcal{T}^{+}$}
\label{app:main:G:log+}
In the first part of this appendix, we prove Theorem~\ref{Th:main:G:log+}. In the remainder of this section, we investigate the properties of $\mathcal{T}^{+}$ and identify several of its important subsets.
\subsection{Proof of Theorem~\ref{Th:main:G:log+}}\label{app:main:G:log+:p1}
   \begin{proof}
    Since binary distributions are adequate to find sufficient conditions for \eqref{log+-new-form}, we can rewrite \eqref{log+-new-form} as
    \begin{align}
    &M_{x,r,s}(y)=xrf\left(\frac{ys}{xr}\right)+x(1-r)f\left(\frac{y(1-s)}{x(1-r)}\right)+(1-x)rf\left(\frac{(1-y)s}{(1-x)r}\right)+(1-x)(1-r)f\left(\frac{(1-y)(1-s)}{(1-x)(1-r)}\right)
    \nonumber\\&-\left(xf\left(\frac{y}{x}\right)+(1-x)f\left(\frac{1-y}{1-x}\right)\right)\times
   \left(rf\left(\frac{s}{r}\right)+(1-r)f\left(\frac{1-s}{1-r}\right)\right).\label{gen:2}
\end{align}    
where $x,y,r,s\in[0,1]$ are defined as in \eqref{r_Y}–\eqref{q_Z}. \\
We want to prove that $M_{x,r,s}(y)\le 0$ for all $x,y,r,s\in[0,1]$. Note that
\[
  M_{x,r,s}(y)\big|_{y=x} = 0
  \quad\text{and}\quad
  \frac{\partial M_{x,r,s}(y)}{\partial y}\bigg|_{y=x} = 0.
\]
Therefore, if
\[
  \frac{\partial^2 M_{x,r,s}(y)}{\partial y^2} \le 0
  \quad\text{for all }x,y,r,s\in[0,1],
\]
then $M_{x,r,s}(y)$ is concave in $y$ and attains its maximum at $y=x$, which implies $M_{x,r,s}(y)\le 0$ and hence inequality~\eqref{gen:2} holds.
 
The $\frac{\partial^2 M_{x,r,s}(y)}{\partial y^2}$ can be expanded as follows:

    \begin{align*}
     &\frac{s^2}{xr}f''\left(\frac{ys}{xr}\right)+\frac{(1-s)^2}{x(1-r)}f''\left(\frac{y(1-s)}{x(1-r)}\right)-\frac{1}{x}f''\left(\frac{y}{x}\right)\times A_{r,s}
     \nonumber \\&+
     \frac{s^2}{(1-x)r}f''\left(\frac{(1-y)s}{(1-x)r}\right)+\frac{(1-s)^2}{(1-x)(1-r)}f''\left(\frac{(1-y)(1-s)}{(1-x)(1-r)}\right)-\frac{1}{1-x}f''\left(\frac{1-y}{1-x}\right)\times A_{r,s},
\end{align*}
where $A_{r,s}=rf\left(\frac{s}{r}\right)+(1-r)f\left(\frac{1-s}{1-r}\right)$.\\
Let $g=x^2f''(x)$. We obtain:
\begin{align}
    \frac{\partial^2 M_{x,r,s}(y)}{\partial y^2}&=\nonumber\\
    &\frac{x}{y^2}\left(rg\left(\frac{ys}{xr}\right)+(1-r)g\left(\frac{y(1-s)}{x(1-r)}\right)-g\left(\frac{y}{x}\right)A_{r,s}\right)\nonumber \\
    &+\frac{1-x}{(1-y)^2}\left(rg\left(\frac{(1-y)s}{(1-x)r}\right)+(1-r)g\left(\frac{(1-y)(1-s)}{(1-x)(1-r)}\right)-g\left(\frac{1-y}{1-x}\right)A_{r,s}\right).
\end{align}
Define $Q_{r,s}(v)=rg\left(\frac{vs}{r}\right)+(1-r)g\left(\frac{v(1-s)}{(1-r)}\right)-g(v)A_{r,s}$. Thus,
\begin{align}
    \frac{\partial^2 M_{x,r,s}(y)}{\partial y^2}=
    \frac{x}{y^2}\left(Q_{r,s}\left(\frac{y}{x}\right)\right)
    +\frac{1-x}{(1-y)^2}\left(Q_{r,s}\left(\frac{1-y}{1-x}\right)\right).
\end{align}
If $Q_{r,s}(v)\le 0$ for all $v\ge 0$ and $r,s\in[0,1]$, then inequality $M_{x,r,s}(y)\leq 0$ follows. The function $Q_{r,s}(v)$ can be rewritten as
\begin{align}
  Q_{r,s}(v)
  &= r\,g\!\left(\frac{v s}{r}\right)
   + (1-r)\,g\!\left(\frac{v(1-s)}{1-r}\right)
   - g(v)\,A_{r,s} \nonumber\\
  &= r\Bigl(g\!\left(\tfrac{v s}{r}\right)-g(v)\,f\!\left(\tfrac{s}{r}\right)\Bigr)
   + (1-r)\Bigl(g\!\left(\tfrac{v(1-s)}{1-r}\right)-g(v)\,f\!\left(\tfrac{1-s}{1-r}\right)\Bigr) \nonumber\\
  &= r\,H\!\left(v,\tfrac{s}{r}\right)
   + (1-r)\,H\!\left(v,\tfrac{1-s}{1-r}\right),
  \label{gen:3}
\end{align}
where $H(\beta,\alpha)=g(\alpha\beta)-g(\beta)f(\alpha)$.

If $H(\beta,\alpha)$ is concave in $\alpha$ for all $\alpha,\beta\ge 0$, then by Jensen’s inequality,
\begin{align}
  Q_{r,s}(v)
  &= r\,H\!\left(v,\tfrac{s}{r}\right)
   + (1-r)\,H\!\left(v,\tfrac{1-s}{1-r}\right)\\
  &\le H\!\left(v,\;r\cdot\tfrac{s}{r}+(1-r)\cdot\tfrac{1-s}{1-r}\right)
   = H(v,1)\\
  &= g(v)-g(v)f(1)=0,
\end{align}
so $Q_{r,s}(v)\le 0$ for all $v\ge 0$ and $r,s\in[0,1]$, using $f(1)=1$.

The concavity of $H(\beta,\alpha)$ implies
\begin{align}
  \frac{\partial^2 H}{\partial\alpha^2}
  = \beta^2 g''(\alpha\beta) - g(\beta) f''(\alpha) \le 0.
\end{align}
Multiplying by $\alpha^2$ and using $g(\alpha)=\alpha^2 f''(\alpha)$, we obtain
\begin{align}
  \alpha^2 \beta^2 g''(\alpha\beta) - g(\beta) g(\alpha) \le 0,
\end{align}
which is equivalent to
\begin{align}
  x^2 g''(x) - g(\alpha)\,g\!\left(\tfrac{x}{\alpha}\right)\le 0,
  \qquad \forall x,\alpha\ge 0.\label{final:log+:proof}
\end{align}
Thus, if $g(x)=x^2 f''(x)$ satisfies inequality~\eqref{final:log+:proof}, then inequality~\eqref{log+-new-form} holds. Note that the convexity of $f(x)$ implies that $g(x)\ge 0$. This completes the proof.

\end{proof}


\subsection{Properties of $\mathcal{T}^{+}$}\label{app:main:G:log+:p2}
We begin this subsection by characterizing those functions in $\mathcal{T}^{+}$ that satisfy the defining inequality with equality. In other words, we seek all non-negative functions $g$ on $[0,\infty)$ such that
\[
  x^{2} g''(x) - g(\alpha)\,g\!\left(\frac{x}{\alpha}\right) = 0,
  \qquad \forall x\ge 0,\ \forall \alpha>0.
\]
We have the following lemma.
\begin{lemma}\label{eq:def:log+}
The only non-negative solution of the equation $x^2g''(x)-g(\alpha)g\left(\frac{x}{\alpha}\right)=0, \forall x,\alpha\geq 0,$ has the following form $$g(x)=\gamma(\gamma-1)x^{\gamma},~\gamma\geq1 ~~\text{or}~~ \gamma\leq 0.$$
\end{lemma}
\begin{proof}
    Since $x^2g''(x)-g(\alpha)g\left(\frac{x}{\alpha}\right)=0$ must hold for every non-negative value of $x$ and $\alpha$, we write this equation for $x,\alpha$ and $\beta\geq0$ as follows:
    \begin{align*}
       x^2g''(x)&=g(\alpha)g\left(\frac{x}{\alpha}\right)\\
       x^2g''(x)&=g(\beta)g\left(\frac{x}{\beta}\right).
    \end{align*}
    Thus we must have
      \begin{align*}
       g(\beta)g\left(\frac{x}{\beta}\right)=g(\alpha)g\left(\frac{x}{\alpha}\right).
    \end{align*}
    Let $x=\alpha$ in the above, we obtain:
    \begin{align*}
        g(\alpha)g(1)=g(\beta)g\left(\frac{\alpha}{\beta}\right).
    \end{align*}
    Let $g(x)=\exp(E(\log(x)))$. We have:
     \begin{align*}
       E(\log(\alpha))+E(0)=E(\log(\beta))+E(\log(\alpha)-\log(\beta)).
    \end{align*}
    Let $a=\log(\alpha)$ and $b=\log(\beta)$, we have:
\begin{align}
    E(a)+E(0)=E(b)+E(a-b).\label{eq:lin}
\end{align}
Observe that $a,b\in \mathbb{R}$. \eqref{eq:lin} is a very basic functional equation and it is well known that $E(x)=\gamma x +d$.
Thus $g(x)=\exp(E(\log(x)))=\widetilde{a}x^{\gamma}$ where $\widetilde{a}=e^{d}$. By substituting this form in the original equation, we have:
\begin{align}
  x^2g''(x)-g(\alpha)g\left(\frac{x}{\alpha}\right)&=\widetilde{a}(\gamma(\gamma-1)-\widetilde{a})x^{\gamma}=0,  
\end{align}
Thus $\widetilde{a}=\gamma(\gamma-1)$.
\end{proof}

In the following lemma, we establish some basic properties of the set $\mathcal{T}^{+}$.
\begin{lemma}\label{T:+:obser}
    We have the following observations:
    \begin{enumerate}
    \item Let $g$ be a non-negative concave function on $[0,\infty)$. Then $g(x)\in \mathcal{T}^{+}$. 
        \item Let $g_1,g_2\in \mathcal{T}^{+}$. If $a,b\geq 1$, then 
        $ag_1+bg_2\in \mathcal{T}^{+}$.
        \item Let $g_1=\gamma(\gamma-1)x^{\gamma}$ and 
        $g_2=\theta(\theta-1)x^{\theta}$ be two members of $\mathcal{T}^{+}$. Then 
        $ag_1+bg_2\in \mathcal{T}^{+}$ if and only if $a,b\geq 1$.\\
        \item Every function $g\in \mathcal{T}^{+}$ does not have any positive zero (as a result, functions such as $\sin^2(x),(x-1)^2,\cdots \notin \mathcal{T}^{+}$). This shows that the only zeros of $g\in \mathcal{T}^{+}$ (if they exist) are at $x=0$.\\
        \item Suppose that $h(x)$ is any non-negative function with a bounded second derivative by $M$ (i.e., $|h''(x)|\leq M$).
        Then the function $h+cx^2\in \mathcal{T}^{+}$ if $c\geq 1+\sqrt{M+1}$. Note that $h$ is not necessarily a member of $\mathcal{T}^{+}$.
\end{enumerate}
\end{lemma}
\begin{proof}
\begin{enumerate}
  \item Since $g$ is concave, $g''(x)\le 0$. Non-negativity of $g$ implies $-g(\alpha)g\!\left(\frac{x}{\alpha}\right)\le 0$. Adding these gives
  \[
    x^2 g''(x) - g(\alpha)\,g\!\left(\frac{x}{\alpha}\right) \le 0.
  \]

  \item Let $g = a g_1 + b g_2$ with $g_1,g_2\in\mathcal{T}^{+}$ and $a,b\ge 0$. Using $a^2\ge a$ and $b^2\ge b$, we have
  \begin{align*}
    &x^2\bigl(a g_1''(x)+b g_2''(x)\bigr)
    -a^2 g_1(\alpha)g_1\!\left(\frac{x}{\alpha}\right)
    -b^2 g_2(\alpha)g_2\!\left(\frac{x}{\alpha}\right)\\
    &\quad -ab\Bigl[g_1(\alpha)g_2\!\left(\frac{x}{\alpha}\right)+g_2(\alpha)g_1\!\left(\frac{x}{\alpha}\right)\Bigr]\\
    &\le a\bigl[x^2 g_1''(x)-g_1(\alpha)g_1\!\left(\frac{x}{\alpha}\right)\bigr]
      + b\bigl[x^2 g_2''(x)-g_2(\alpha)g_2\!\left(\frac{x}{\alpha}\right)\bigr]\le 0.
  \end{align*}
  Thus $g\in\mathcal{T}^{+}$.

  \item Assume $\gamma>\theta$ without loss of generality. The case $\gamma=\theta$ is trivial.  Since $a x^\gamma + b x^\theta\ge 0$, we have $a,b\ge 0$. Substituting into the defining inequality for $\mathcal{T}^{+}$ gives (where \eqref{eqn:2} and \eqref{eqn:3} are obtained by dividing \eqref{eqn:1} by $x^\theta$ and $x^\gamma$, respectively):
  \begin{align}
    &(a-a^2)(\gamma(\gamma-1))x^\gamma
    + (b-b^2)(\theta(\theta-1))x^\theta \nonumber\\
    &\quad -ab(\gamma(\gamma-1))(\theta(\theta-1))
      \Bigl[\alpha^\gamma\!\left(\frac{x}{\alpha}\right)^\theta
            + \alpha^\theta\!\left(\frac{x}{\alpha}\right)^\gamma\Bigr] \le 0, \label{eqn:1}\\
    &(a-a^2)(\gamma(\gamma-1))x^{\gamma-\theta}
    + (b-b^2)(\theta(\theta-1)) \nonumber\\
    &\quad -ab(\gamma(\gamma-1))(\theta(\theta-1))
      \Bigl[\alpha^{\gamma-\theta} + \alpha^{\theta-\gamma}x^{\gamma-\theta}\Bigr] \le 0, \label{eqn:2}\\
    &(a-a^2)(\gamma(\gamma-1))
    + (b-b^2)(\theta(\theta-1))x^{\theta-\gamma} \nonumber\\
    &\quad -ab(\gamma(\gamma-1))(\theta(\theta-1))
      \Bigl[\alpha^{\gamma-\theta}x^{\theta-\gamma} + \alpha^{\theta-\gamma}\Bigr] \le 0. \label{eqn:3}
  \end{align}
  Taking $x\to 0^+$ then $\alpha\to 0^+$ in \eqref{eqn:2} gives $b-b^2\le 0$. Taking $x\to+\infty$ then $\alpha\to+\infty$ in \eqref{eqn:3} gives $a-a^2\le 0$. Thus $a,b\geq 1$.

  \item Suppose $g\in\mathcal{T}^{+}$ vanishes at some $x_0>0$. Setting $\alpha=x_0$ gives $x^2 g''(x)\le 0$ for all $x\ge 0$. If $g$ were convex at any positive point, this would be a contradiction. Thus, $g$ must be concave. We know that every non-negative concave function on $[0, \infty)$ must be increasing. Thus, having a positive root is impossible.
\item Substituting $h + c x^{2}$ into the defining inequality of $\mathcal{T}^{+}$, we obtain
\begin{align*}
    x^{2}\bigl(h + c x^{2}\bigr)'' 
    - \bigl(h(\alpha) + c\alpha^{2}\bigr)
      \Bigl(h\!\left(\tfrac{x}{\alpha}\right)
            + c\left(\tfrac{x}{\alpha}\right)^{2}\Bigr)
    &\le (2c + M)x^{2} - c^{2}x^{2} \\
    &= (2c + M - c^{2})x^{2} \le 0.
\end{align*}
Hence we must have $2c + M - c^{2} \le 0$, which implies $c \;\ge\; 1 + \sqrt{M+1}$.
\end{enumerate}
\end{proof}
\begin{example}
According to Lemma~\ref{T:+:obser}, the following functions are contained in $\mathcal{T}^{+}$:
\begin{itemize}
  \item $g(x) = \log(1 + x)$,\quad $g(x) = \frac{x+1}{x+2}$ \hspace{1em} (\text{by observation~1})
  \item $g(x) = 6x^3 + \frac{x+1}{x+2}$ \hspace{1em} (\text{by observation~2})
  \item $g(x) = 1 + \sin(x) + 3x^2$ \hspace{1em} (\text{by observation~5})
\end{itemize}
\end{example}
\begin{remark}
We observed in Lemma~\ref{T:+:obser} that every non-negative concave function belongs to $\mathcal{T}^{+}$. Moreover, we saw that certain convex functions of the form $\gamma(\gamma-1)x^{\gamma}$ also lie in $\mathcal{T}^{+}$. Using Lemma~\ref{T:+:obser}, one can construct elements of $\mathcal{T}^{+}$ that are neither convex nor concave. In the next lemma, we present further examples of convex functions in $\mathcal{T}^{+}$.
\end{remark}

\begin{lemma}\label{other:convex:log+}
Functions of the form $\dfrac{x+b}{x+c}$ belong to $\mathcal{T}^{+}$ for all $b\ge \max\{c,2\}$ and $c>0$.
\end{lemma}
\begin{proof}
Substituting $g(x)=\dfrac{x+b}{x+c}$ into the inequality
\[
x^{2}g''(x)-g(\alpha)\,g\!\left(\frac{x}{\alpha}\right)\le 0
\]
yields
\begin{align}
  A_4(\alpha,b,c)x^{4}
  + A_3(\alpha,b,c)x^{3}
  + A_2(\alpha,b,c)x^{2}
  + A_1(\alpha,b,c)x
  + A_0(\alpha,b,c)
  \le 0,
\end{align}
where
\begin{align*}
  A_0(\alpha,b,c)
    &= -\alpha b^{2}c^{3} - \alpha^{2} b c^{3},\\
  A_1(\alpha,b,c)
    &= -\alpha c^{3} - b c^{3}
       - 3\alpha b^{2}c^{2} - 3\alpha^{2} b c^{2},\\
  A_2(\alpha,b,c)
    &= 2\alpha b c^{2} - 2\alpha c^{3}
       - 3 b c^{2} - 2\alpha^{2} c^{2}
       - 3\alpha c^{2} - 3\alpha b^{2}c - \alpha^{2} b c,\\
  A_3(\alpha,b,c)
    &= 2\alpha b - 5\alpha c - b c - \alpha b^{2} - \alpha^{2} b - 2c^{2},\\
  A_4(\alpha,b,c)
    &= -\alpha - b.
\end{align*}
To prove the claim, it suffices to show that all coefficients $A_i(\alpha,b,c)$ are non-positive. Clearly $A_0$, $A_1$, and $A_4$ are non-positive for all $\alpha\ge 0$, $b\ge 0$, and $c>0$. Thus it remains to verify that $A_2$ and $A_3$ are non-positive.

For $A_2$, it is enough to check that the coefficient of $\alpha$,
\[
  2b c^{2} - 2c^{3} - 2c^{2} - 3 b^{2} c,
\]
is non-positive. When $0<c\leq b$, this follows from
\[
  2b c^{2} - 3 b^{2} c \le 0.
\]
For $A_3$, it suffices to show that the coefficient of $\alpha$,
\[
  2b - 5c - b^{2},
\]
is non-positive. For every $b\ge 2$, we have $2b - b^{2}\le 0$, and hence $2b - 5c - b^{2}\le 0$ for all $c>0$. Therefore $A_2(\alpha,b,c)\le 0$ and $A_3(\alpha,b,c)\le 0$ under the stated conditions, so the polynomial is non-positive and $g\in\mathcal{T}^{+}$.
\end{proof}
We note that, by an argument entirely analogous to the proof of Lemma~\ref{other:convex:log+}, the function $g(x)=\dfrac{x+b}{x+c}$ also belongs to $\mathcal{T}^{+}$ whenever $b\ge c\ge \tfrac{1}{5}$.

\subsection{Perturbation around $\lambda(\lambda-1)x^{\lambda}$}

We seek non-negative functions $h(x)$ on $[0,\infty)$ such that, for every $\epsilon\ge 0$, the perturbed function
\[
  \lambda(\lambda-1)x^{\lambda} + \epsilon h(x)
\]
belongs to $\mathcal{T}^{+}$. Substituting $\lambda(\lambda-1)x^{\lambda}+\epsilon h(x)$ into the defining inequality of $\mathcal{T}^{+}$ shows that it is sufficient for $h$ to satisfy
\begin{align}
  x^{2}h''(x)
  - \lambda(\lambda-1)\alpha^{\lambda}h\!\left(\frac{x}{\alpha}\right)
  - \lambda(\lambda-1)h(\alpha)\left(\frac{x}{\alpha}\right)^{\lambda}
  \le 0\label{eqnAAeq1}
\end{align}
for all $x,\alpha\ge 0$. Accordingly, define
\[
  \mathcal{Q}_{\lambda}
  = \Bigl\{h:[0,\infty)\to[0,\infty)\,\Big|\,
     x^{2}h''(x)
     - \lambda(\lambda-1)\alpha^{\lambda}h\!\left(\tfrac{x}{\alpha}\right)
     - \lambda(\lambda-1)h(\alpha)\left(\tfrac{x}{\alpha}\right)^{\lambda}
     \le 0,\ \forall x,\alpha\ge 0
    \Bigr\}.
\]
Clearly, $\lambda(\lambda-1)x^{\lambda}+\epsilon \mathcal{Q}_{\lambda}\subset \mathcal{T}^{+}$. Moreover, every non-negative concave function $h$ belongs to $\mathcal{Q}_{\lambda}$. 

Let us first construct a function that satisfies \eqref{eqnAAeq1} with equality. Note that this function  $h(x)$ is not strictly non-negative on $[0,\infty)$. Thus $h(x)\notin \mathcal{Q}_{\lambda} $. However, it will be used to construct a function that belongs to $\mathcal{Q}_{\lambda}$ later. 

\begin{lemma}
Functions of the form 
\[
  h(x)
  = h(1)\,x^{\lambda}\!\left(1 + \frac{\lambda^{2}-\lambda}{2\lambda-1}\log x\right)
\]
satisfy
\[
  x^{2}h''(x)
  - \lambda(\lambda-1)\alpha^{\lambda}h\!\left(\frac{x}{\alpha}\right)
  - \lambda(\lambda-1)h(\alpha)\left(\frac{x}{\alpha}\right)^{\lambda}
  = 0
\]
for all $x\ge 0$ and $\alpha\ge 0$.
\end{lemma}

\begin{proof}
Since
\[
  x^{2}h''(x)
  - \lambda(\lambda-1)\alpha^{\lambda}h\!\left(\frac{x}{\alpha}\right)
  - \lambda(\lambda-1)h(\alpha)\left(\frac{x}{\alpha}\right)^{\lambda}
  = 0
\]
holds for all $x,\alpha\ge 0$, we can rewrite it as
\begin{align}
  \frac{x^{2}h''(x)}{\lambda(\lambda-1)}
    &= \alpha^{\lambda}h\!\left(\frac{x}{\alpha}\right)
       + h(\alpha)\left(\frac{x}{\alpha}\right)^{\lambda},\\
  \frac{x^{2}h''(x)}{\lambda(\lambda-1)}
    &= \beta^{\lambda}h\!\left(\frac{x}{\beta}\right)
       + h(\beta)\left(\frac{x}{\beta}\right)^{\lambda}
\end{align}
for any $\beta\ge 0$. Hence
\[
  \alpha^{\lambda}h\!\left(\tfrac{x}{\alpha}\right)
  + h(\alpha)\left(\tfrac{x}{\alpha}\right)^{\lambda}
  = \beta^{\lambda}h\!\left(\tfrac{x}{\beta}\right)
  + h(\beta)\left(\tfrac{x}{\beta}\right)^{\lambda}.
\]
Setting $x=\alpha$ gives
\[
  \alpha^{\lambda}h(1) + h(\alpha)
  = \beta^{\lambda}h\!\left(\tfrac{\alpha}{\beta}\right)
    + h(\beta)\left(\tfrac{\alpha}{\beta}\right)^{\lambda}.
\]
Dividing by $\alpha^{\lambda}$ yields
\[
  h(1) + \frac{h(\alpha)}{\alpha^{\lambda}}
  = \frac{h\!\left(\tfrac{\alpha}{\beta}\right)}
         {\left(\tfrac{\alpha}{\beta}\right)^{\lambda}}
    + \frac{h(\beta)}{\beta^{\lambda}}.
\]
Define $h(\alpha) = \alpha^{\lambda}U(\log\alpha)$ and set $a=\log\alpha$, $b=\log\beta$. Then
\[
  U(0) + U(a) = U(a-b) + U(b),
\]
which implies that $U$ is affine, i.e. $U(x)=cx+d$. Thus
\[
  h(x) = x^{\lambda}(c\log x + d).
\]
Substituting this expression into
\[
  x^{2}h''(x)
  - \lambda(\lambda-1)\alpha^{\lambda}h\!\left(\tfrac{x}{\alpha}\right)
  - \lambda(\lambda-1)h(\alpha)\left(\tfrac{x}{\alpha}\right)^{\lambda}
  = 0
\]
and simplifying (all $x$-dependent factors cancel) shows that
\[
  c = \frac{\lambda^{2}-\lambda}{2\lambda-1}\,d.
\]
Writing $d=h(1)$ yields
\[
  h(x)
  = h(1)\,x^{\lambda}
    \left(1 + \frac{\lambda^{2}-\lambda}{2\lambda-1}\log x\right),
\]
which completes the proof.
\end{proof}

\begin{remark}
    $\mathcal{Q}_{\lambda}$ forms a linear cone. In other words, for any 
    $h_1,h_2\in \mathcal{Q}_{\lambda}$ and for any non-negative constants
    $a,b$, the function $ah_1+bh_2$ would be a member of $\mathcal{Q}_{\lambda}$. As a result, assume that $h_{\beta}(x)$ is a parametric family which lies in $\mathcal{Q}_{\lambda}$ for every $\beta$ in the interval $I$. Let $\mu_{\beta}$ be any non-negative measure on $I$. Then the function $h(x)=\int_{I}h_{\beta}(x)d\mu_{\beta}$ will lie in $\mathcal{Q}_{\lambda}$. 
    \end{remark}
    \begin{lemma}\label{per:pos}
Let $a$ and $\beta$ be positive constants such that $a\beta\ge \tfrac{1}{e}$. Then the function $g(x)=\dfrac{a}{x^{\beta}}+\log x$ is non-negative on $(0,\infty)$.
\end{lemma}

\begin{proof}
We have $\lim_{x\to 0^{+}} g(x) = +\infty$ and $\lim_{x\to +\infty} g(x) = +\infty$, so any minimum of $g$ must occur at a critical point in $(0,\infty)$. Differentiating gives
\[
  g'(x) = -\frac{a\beta}{x^{\beta+1}} + \frac{1}{x},
\]
and solving $g'(x)=0$ yields a unique critical point at $x_{*} = (a\beta)^{1/\beta}$. Evaluating $g$ there,
\[
  g(x_{*}) = \frac{1}{\beta}\bigl(1 + \log(a\beta)\bigr).
\]
Thus $g(x_{*})\ge 0$ precisely when $1+\log(a\beta)\ge 0$, i.e. $a\beta\ge \tfrac{1}{e}$. Since $x_{*}$ is the global minimum and $g(x)\to+\infty$ at both endpoints, it follows that $g(x)\ge 0$ for all $x>0$.
\end{proof}
\begin{lemma}
Let $a>0$ and $\beta>0$ satisfy $a\beta\ge \tfrac{1}{e}$ and $\lambda-1\le \beta\le \lambda$, and assume $\lambda\ge 1$. Then the functions of the form
\[
  h(x)
  = x^{\lambda}\left(1 + \frac{\lambda^{2}-\lambda}{2\lambda-1}
      \left(\frac{a}{x^{\beta}}+\log x\right)\right),
\]
are elements of the class $\mathcal{Q}_{\lambda}$.
\end{lemma}

\begin{proof}
By Lemma~\ref{per:pos}, the function $h(x)$ is non-negative. Since the inequality
\[
  x^{2}h''(x)
  - \lambda(\lambda-1)\alpha^{\lambda}h\!\left(\frac{x}{\alpha}\right)
  - \lambda(\lambda-1)h(\alpha)\left(\frac{x}{\alpha}\right)^{\lambda}
  \le 0
\]
is linear in $h$, and the function
\[
  h(x)
  = h(1)\,x^{\lambda}\left(1+\frac{\lambda^{2}-\lambda}{2\lambda-1}\log x\right)
\]
satisfies the corresponding equality, it suffices to verify the inequality for $h(x) = x^{\lambda-\beta}$. In this case we obtain
\begin{align*}
  x^{2}h''(x)
  &- \lambda(\lambda-1)\alpha^{\lambda}h\!\left(\tfrac{x}{\alpha}\right)
  - \lambda(\lambda-1)h(\alpha)\left(\tfrac{x}{\alpha}\right)^{\lambda} \\
  &= (\lambda-\beta)(\lambda-\beta-1)x^{\lambda-\beta}
     - \lambda(\lambda-1)\alpha^{\beta}x^{\lambda-\beta}
     - \lambda(\lambda-1)\alpha^{-\beta}x^{\lambda}.
\end{align*}
Factoring out $x^{\lambda-\beta}>0$ shows that this expression is non-positive for all $x,\alpha>0$ if and only if
\[
  (\lambda-\beta)(\lambda-\beta-1)\le 0,
\]
which holds precisely when $\lambda-1\le \beta\le \lambda$. This completes the proof.
\end{proof}
\begin{example}
Let $\lambda>1$ and $a\ge \tfrac{1}{e(\lambda-1)}$, and let $\mu_{\beta}$ be the uniform distribution on the interval $[\lambda-1,\lambda]$. Define
\[
  h_{\beta}(x)
  = x^{\lambda}\left(1+\frac{\lambda^{2}-\lambda}{2\lambda-1}
      \left(\frac{a}{x^{\beta}}+\log x\right)\right).
\]
For every $\beta\in[\lambda-1,\lambda]$ we have $h_{\beta}\in\mathcal{Q}_{\lambda}$. Hence the averaged function
\[
  \hat{h}(x) = \int_{[\lambda-1,\lambda]} h_{\beta}(x)\,d\mu_{\beta}
\]
also belongs to $\mathcal{Q}_{\lambda}$. A direct computation shows that
\[
  \hat{h}(x)
  = x^{\lambda}\left(1+\frac{\lambda^{2}-\lambda}{2\lambda-1}\log x\right)
    + \frac{(\lambda^{2}-\lambda)a}{2\lambda-1}\,\frac{x-1}{\log x}.
\]
\end{example}

\begin{lemma}
Let $\lambda=2$. Then the function $h(x)=\dfrac{x+b}{x+c}$ belongs to $\mathcal{Q}_{\lambda}$ whenever $1\le c<b\le 2c$.
\end{lemma}
Note that when $0\le b\le c$, the function $h(x)=\dfrac{x+b}{x+c}$ is non-negative and concave, and hence $h\in \mathcal{Q}_{\lambda}$.
\begin{proof}
Substituting $h(x)=\dfrac{x+b}{x+c}$ into the defining inequality of $\mathcal{Q}_{\lambda}$ with $\lambda=2$ and simplifying yields
\[
  A_6(\alpha,b,c)x^6 + A_5(\alpha,b,c)x^5 + A_4(\alpha,b,c)x^4
  + A_3(\alpha,b,c)x^3 + A_2(\alpha,b,c)x^2 + A_1(\alpha,b,c)x + A_0(\alpha,b,c) \le 0,
\]
where
\begin{align*}
A_0(\alpha,b,c) &= -\alpha^5 b c^4 - \alpha^6 b c^3,\\
A_1(\alpha,b,c) &= -\alpha^4 c^4 - \alpha^5 c^3 - 3 \alpha^5 b c^3 - 3 \alpha^6 b c^2,\\
A_2(\alpha,b,c) &= \alpha^3 b c^2 - \alpha^3 c^3 - \alpha^4 c^2 - 3 \alpha^4 c^3 - 3 \alpha^5 c^2 - \alpha^2 c^4 \\
&\quad - 3 \alpha^5 b c^2 - \alpha b c^4 + \alpha^4 b c - 3 \alpha^6 b c,\\
A_3(\alpha,b,c) &= \alpha^3 b - \alpha^6 b - \alpha c^3 - \alpha^3 c - 3 \alpha^5 c - b c^3 \\
&\quad - \alpha^2 c^2 - 3 \alpha^2 c^3 - 3 \alpha^4 c^2 + \alpha^2 b c - 3 \alpha b c^3 - \alpha^5 b c,\\
A_4(\alpha,b,c) &= -3 \alpha c^2 - \alpha^4 c - 3 b c^2 - 3 \alpha^2 c^2 - 3 \alpha b c^2-\alpha^5,\\
A_5(\alpha,b,c) &= -3 \alpha c - 3 b c - \alpha^2 c - \alpha b c,\\
A_6(\alpha,b,c) &= -\alpha - b.
\end{align*}
The coefficients $A_0,A_1,A_4,A_5,A_6$ are clearly non-positive for all $\alpha\ge 0$ and $1\le c\le b\le 2c$, so it remains to show that $A_2$ and $A_3$ are non-positive.

For $A_3$, consider the coefficients of $\alpha^2,\alpha^3$, and $\alpha^4$:
\[
  o_2 = c(b - c - 3c^2),\qquad
  o_3 = b - c,\qquad
  o_4 = -3c^2.
\]
All other terms in $A_3$ are non-positive. For $1\le c\le b\le 2c$, we have $o_2<0$.  
If $\alpha \le -o_2/o_3$, then $o_3\alpha^3 + o_2\alpha^2 \le 0$, and since the remaining terms are non-positive, $A_3(\alpha,b,c)\le 0$.  
On the other hand, if $\alpha \ge -o_3/o_4$, then $o_3\alpha^3 + o_4\alpha^4 \le 0$, and again the remaining terms are non-positive, so $A_3(\alpha,b,c)\le 0$.  
For $1\le c\le b\le 2c$ one checks that
\[
  \frac{-o_2}{o_3} \;\ge\; \frac{-o_3}{o_4},
\]
so these two ranges cover all $\alpha\ge 0$, and hence $A_3\le 0$.

For $A_2$, focus on the $\alpha^2,\alpha^3,\alpha^4$ terms; the remaining terms are non-positive. Writing
\[
  \hat{o}_{4}\alpha^4 + \hat{o}_{3}\alpha^3 + \hat{o}_{2}\alpha^2
  = \alpha^2\Bigl(-(3c^3 + c^2 - b c)\alpha^2 + (b c^2 - c^3)\alpha - c^4\Bigr),
\]
where $\hat{o}_{2}$, $\hat{o}_{3}$, and $\hat{o}_{4}$ denote the coefficients of $\alpha^2$, $\alpha^3$, and $\alpha^4$ in $A_2$, respectively. We see that this quadratic in $\alpha$ is non-positive provided
\[
  3c^3 + c^2 - b c \;\ge\; 0
  \quad\text{and}\quad
  \Delta \le 0,
\]
where the discriminant is
\[
  \Delta = c^2(b-c)^2 - 4c^3(3c^2 + c - b).
\]
Under $1\le c\le b\le 2c$ we have $3c^3 + c^2 - b c \ge 0$, and
\[
  (b-c)^2 \le c^2,\quad
  3c^2 + c - b \ge 3c^2 - c,\quad
  c \le 4(3c^2 - c),
\]
which together imply $\Delta\le 0$. Hence $\hat{o}_{4}\alpha^4 + \hat{o}_{3}\alpha^3 + \hat{o}_{2}\alpha^2\le 0$, and adding the remaining non-positive terms yields $A_2(\alpha,b,c)\le 0$.

Thus all coefficients $A_i(\alpha,b,c)$ are non-positive, so the polynomial is non-positive for all $x\ge 0$, which shows $h\in\mathcal{Q}_{2}$. 
\end{proof}
\begin{example}
Let $1\le c\le 2$, set $b=2$, and let $\mu_{c}$ be a non-negative measure on $[1,2]$ with density $d\mu_{c} = \tfrac{1}{\sqrt{c}}\,dc$. Then
\[
  \hat{h}(x)
  = \int_{1}^{2} \frac{x+2}{x+c}\,d\mu_{c}
  = \frac{2x+4}{\sqrt{x}}\left(\tan^{-1}\!\sqrt{\frac{2}{x}}
                               - \tan^{-1}\!\sqrt{\frac{1}{x}}\right)
\]
belongs to $\mathcal{Q}_{\lambda}$ for $\lambda=2$.
\end{example}


\section{Proof of Theorem \ref{Th:main:G:log-} and Properties of $\mathcal{T}^{-}$}\label{app:main:G:log-}
In the first part of this appendix, we prove Theorem~\ref{Th:main:G:log-}. In the remainder of this section, we investigate the properties of $\mathcal{T}^{-}$ and identify several of its important subsets.
\subsection{Proof of Theorem \ref{Th:main:G:log-} }\label{app:main:G:log-:p1}
\begin{proof}
     Since binary distributions are adequate to find sufficient conditions for \eqref{logm-new-form}, we can rewrite \eqref{logm-new-form} 
     for binary distributions as follows
    \begin{align}
    &M_{x,r,s}(y)=xrf\left(\frac{ys}{xr}\right)+x(1-r)f\left(\frac{y(1-s)}{x(1-r)}\right)+(1-x)rf\left(\frac{(1-y)s}{(1-x)r}\right)+(1-x)(1-r)f\left(\frac{(1-y)(1-s)}{(1-x)(1-r)}\right)
    \nonumber\\&-\left(xf\left(\frac{y}{x}\right)+(1-x)f\left(\frac{1-y}{1-x}\right)\right)\times
   \left(rf\left(\frac{s}{r}\right)+(1-r)f\left(\frac{1-s}{1-r}\right)\right).\label{gen:2:su}
\end{align}    
where $x,y,r,s\in[0,1]$ are defined as in \eqref{r_Y}–\eqref{q_Z}. We aim to show that $M_{x,r,s}(y)\geq 0$ for all $x,y,r,s\in[0,1]$. Since f is continuous it suffices to verify \eqref{gen:2:su} for $x,y,r,s\in(0,1)$. Observe that $M_{x,r,s}(y)\big|_{y=x}=0$ and $\frac{\partial M_{x,r,s}(y)}{\partial y}\big|_{y=x}=0$. Thus if 
$\frac{\partial^2 M_{x,r,s}(y)}{\partial y^2}\geq 0, ~\forall x,y,r\in(0,1)$, and $\forall s \in[0,1]$, then inequality $M_{x,r,s}(y)\geq 0$ holds. 
The $\frac{\partial^2 M_{x,r,s}(y)}{\partial y^2}$ is as follows:

    \begin{align*}
     &\frac{s^2}{xr}f''\left(\frac{ys}{xr}\right)+\frac{(1-s)^2}{x(1-r)}f''\left(\frac{y(1-s)}{x(1-r)}\right)-\frac{1}{x}f''\left(\frac{y}{x}\right)\times A_{r,s}
     \nonumber \\&+
     \frac{s^2}{(1-x)r}f''\left(\frac{(1-y)s}{(1-x)r}\right)+\frac{(1-s)^2}{(1-x)(1-r)}f''\left(\frac{(1-y)(1-s)}{(1-x)(1-r)}\right)-\frac{1}{1-x}f''\left(\frac{1-y}{1-x}\right)\times A_{r,s},
\end{align*}
where $A_{r,s}=rf\left(\frac{s}{r}\right)+(1-r)f\left(\frac{1-s}{1-r}\right)$. Let $\ell(x)=x^2f''(x)$. We have:
\begin{align}
    &\frac{\partial^2 M_{x,r,s}(y)}{\partial y^2}=\nonumber\\
    &\quad\frac{x}{y^2}\left(r\ell\left(\frac{ys}{xr}\right)+(1-r)\ell\left(\frac{y(1-s)}{x(1-r)}\right)-\ell\left(\frac{y}{x}\right)A_{r,s}\right)\nonumber \\
    &\quad+\frac{1-x}{(1-y)^2}\left(r\ell\left(\frac{(1-y)s}{(1-x)r}\right)+(1-r)\ell\left(\frac{(1-y)(1-s)}{(1-x)(1-r)}\right)-\ell\left(\frac{1-y}{1-x}\right)A_{r,s}\right).
\end{align}
For $r\in(0,1), \; s \in [0,1]$,  define $Q_{r,s}(v)=r\ell\left(\frac{vs}{r}\right)+(1-r)\ell\left(\frac{v(1-s)}{(1-r)}\right)-\ell(v)A_{r,s}$. We obtain:
\begin{align}
    \frac{\partial^2 M_{x,r,s}(y)}{\partial y^2}=
    \frac{x}{y^2}\left(Q_{r,s}\left(\frac{y}{x}\right)\right)
    +\frac{1-x}{(1-y)^2}\left(Q_{r,s}\left(\frac{1-y}{1-x}\right)\right).
\end{align}
If $Q_{r,s}(v)\geq 0,\;~\forall v\geq0,\;0\leq s\leq 1$, and $0<r<1$ then inequality \eqref{gen:2:su} holds. 
Expanding $Q_{r,s}(v)$ 
\begin{align}
  Q_{r,s}(v)&=r\ell\left(\frac{vs}{r}\right)+(1-r)\ell\left(\frac{v(1-s)}{(1-r)}\right)-\ell(v)A_{r,s} \nonumber\\
  &=r\left(\ell\left(\frac{vs}{r}\right)-\ell(v)f\left(\frac{s}{r}\right)\right)+(1-r)\left(\ell\left(\frac{v(1-s)}{(1-r)}\right)-\ell(v)f\left(\frac{1-s}{1-r}\right)\right).\label{gen:3:su}
\end{align}
Let $H(\beta,\alpha)=\ell(\alpha\beta)-\ell(\beta)f(\alpha)$, for $\alpha,\beta \ge 0$. Then, $Q_{r,s}(v)=rH(v,\frac{s}{r})+(1-r)H(v,\frac{1-s}{1-r})$.\\
If $H(\beta,\alpha) $ is convex in $\alpha$
$$Q_{r,s}(v) \ge H\left(v,r \cdot \frac sr +(1-r)\frac{1-s}{1-r}\right)=H(v,1)=0$$
Thus,
\begin{align}
   \frac{\partial^2 H}{\partial \alpha^2}=\beta^2\ell''(\alpha\beta)-\ell(\beta)f''(\alpha)\geq 0.
\end{align}
 Multiplying both sides by $\alpha^2$, and applying $\ell(\alpha)=\alpha^2 f''(\alpha)$, we obtain:
\begin{align}
   \frac{\partial^2 H}{\partial \alpha^2}=\alpha^2\beta^2\ell''(\alpha\beta)-\ell(\beta)\ell(\alpha)\geq0.
\end{align}
The above is equivalent to 
\begin{align}
   x^2\ell''(x)-\ell(\alpha)\ell\left(\frac{x}{\alpha}\right)\geq0,~~\forall x,\alpha\geq0.
\end{align}
Set $g(x)=-\ell(x)=-x^2f''$, we get:
\begin{align}
   -x^2g''(x)-(-g(\alpha))\left(-g\left(\frac{x}{\alpha}\right)\right)&\geq0,~~\forall x,\alpha\geq0.\\
   x^2g''(x)+g(\alpha)g\left(\frac{x}{\alpha}\right)&\leq0,~~\forall x,\alpha\geq0.
\end{align}
This completes the proof.
\end{proof}

\subsection{On the differential inequality $x^{2}g''(x) + g(\alpha)g\!\left(\tfrac{x}{\alpha}\right)\le 0$}\label{app:main:G:log-:p2}
Recall that the defining differential inequality for $\mathcal{T}^{-}$ is
\[
  x^{2}g''(x) + g(\alpha)g\!\left(\tfrac{x}{\alpha}\right)\le 0
\]
for all $x,\alpha\ge 0$, where $g$ is required to be non-negative.

\begin{lemma}\label{triv}
The following properties hold.
\begin{enumerate}
  \item If $g\in\mathcal{T}^{-}$, then $g$ is a non-negative concave function on $[0,\infty)$.
  \item If $g\in\mathcal{T}^{-}$ and $0\le a\le 1$, then $a g\in\mathcal{T}^{-}$.
  \item Equality in the defining inequality of $\mathcal{T}^{-}$ holds for all $x,\alpha\ge 0$ if and only if
  \[
    g(x) = \lambda(1-\lambda)x^{\lambda}
  \]
  for some $\lambda\in[0,1]$.
\end{enumerate}
\end{lemma}
\begin{proof}
The claims in items~1 and~2 follow directly from the differential inequality
$x^{2}g''(x) + g(\alpha)g\!\left(\tfrac{x}{\alpha}\right)\le 0$.
Moreover, the proof of item~3 is completely analogous to the argument used in Lemma~\eqref{eq:def:log+}, and is therefore omitted.
\end{proof}
In the following lemma, we show that the class $\mathcal{T}^{-}$ is closed under taking geometric means.
\begin{lemma}
If $g_{1},g_{2}\in\mathcal{T}^{-}$, then their geometric mean
\[
  g(x):=\sqrt{g_{1}(x)\,g_{2}(x)}
\]
also belongs to $\mathcal{T}^{-}$. More generally, for any fixed $\theta\in[0,1]$, the weighted geometric mean
\[
  g(x)=g_{1}(x)^{\theta}g_{2}(x)^{1-\theta}
\]
is in $\mathcal{T}^{-}$.
\end{lemma}

\begin{proof}
Set $a(x)=g_{1}(x)$, $b(x)=g_{2}(x)$ and $g(x)=\sqrt{a(x)b(x)}$. A straightforward computation gives
\begin{equation}\label{eq:fpp}
  g''(x)
  = \frac{a''(x)b(x)+a(x)b''(x)}{2\sqrt{a(x)b(x)}}
    - \frac{\bigl(a'(x)b(x)-a(x)b'(x)\bigr)^{2}}{4\bigl(a(x)b(x)\bigr)^{3/2}}.
\end{equation}
The second term on the right-hand side of \eqref{eq:fpp} is non-positive, so discarding it can only increase $g''(x)$. Hence
\[
  g''(x)
  \le \frac{a''(x)b(x)+a(x)b''(x)}{2\sqrt{a(x)b(x)}}.
\]
Multiplying by $x^{2}$ and using the defining inequality for $\mathcal{T}^{-}$,
\[
  x^{2}a''(x)\le -a(y)a\!\left(\tfrac{x}{y}\right),\qquad
  x^{2}b''(x)\le -b(y)b\!\left(\tfrac{x}{y}\right),
\]
we obtain
\begin{align*}
  x^{2}g''(x)
  &\le \frac{1}{2\sqrt{a(x)b(x)}}
       \Bigl\{-a(y)a\!\left(\tfrac{x}{y}\right)b(x)
              -b(y)b\!\left(\tfrac{x}{y}\right)a(x)\Bigr\}\\
  &=-\frac{1}{2}\Bigl[
      \sqrt{\tfrac{b(x)}{a(x)}}\,a(y)a\!\left(\tfrac{x}{y}\right)
      + \sqrt{\tfrac{a(x)}{b(x)}}\,b(y)b\!\left(\tfrac{x}{y}\right)
    \Bigr].
\end{align*}
Define
\[
  P=\sqrt{a(y)a\!\left(\tfrac{x}{y}\right)},\qquad
  Q=\sqrt{b(y)b\!\left(\tfrac{x}{y}\right)}.
\]
Then, by the geometric–mean inequality, the above bound simplifies to
$x^{2}g''(x)\le -P Q$.
 On the other hand,
\[
  g(y)g\!\left(\tfrac{x}{y}\right)
  = \sqrt{a(y)b(y)a\!\left(\tfrac{x}{y}\right)b\!\left(\tfrac{x}{y}\right)}
  = P Q.
\]
Therefore
\[
  x^{2}g''(x) + g(y)g\!\left(\tfrac{x}{y}\right)\le 0
\]
for all $x,y>0$, so $g\in\mathcal{T}^{-}$.

The same argument applies, with minor notational changes, to the weighted geometric mean $g(x)=g_{1}(x)^{\theta}g_{2}(x)^{1-\theta}$, $\theta\in[0,1]$, and is omitted.
\end{proof}
In the following lemma, we investigate additional properties of functions belonging to $\mathcal{T}^{-}$.
\begin{lemma}
Assume that $g\in \mathcal{T}^{-}$. Then:
\begin{itemize}
  \item $g(0)=0$.
  \item If $g(1)=\lambda(1-\lambda)$ for some $0\le \lambda\le \tfrac{1}{2}$, then $\frac{g(x)}{x^{\lambda}}$ is increasing and $\frac{g(x)}{x^{1-\lambda}}$ is decreasing on $(0,\infty)$.
\item $g(1)=\tfrac{1}{4}$ if and only if $g(x)=\tfrac{1}{4}\sqrt{x}$.
  \item $g(1)\le \tfrac{1}{4}$.
 
  \item $g(x)\to+\infty$ as $x\to+\infty$, the ratio $\frac{g(x)}{x}$ tends to $0$ as $x\to+\infty$, and $\frac{g(x)}{x}$ is decreasing on $(0,\infty)$.
\end{itemize}
\end{lemma}
\begin{proof}
\begin{enumerate}
\item Since $g$ is non-negative and concave, it is increasing on $[0,\infty)$. For $x\ge\alpha^{2}$ we have $g(\alpha)\le g\!\left(\tfrac{x}{\alpha}\right)$, hence
\[
x^{2}g''(x)+g(\alpha)^{2}
\le x^{2}g''(x)+g(\alpha)g\!\left(\tfrac{x}{\alpha}\right)\le 0,
\]
so $g''(x)\le -g(\alpha)^{2}/x^{2}$ for all $x\ge\alpha^{2}$. Integrating twice gives, for $x\ge\alpha^{2}$,
\begin{align*}
g'(x)-g'(\alpha^{2}) &\le -g(\alpha)^{2}\Bigl(-\tfrac{1}{x}+\tfrac{1}{\alpha^{2}}\Bigr),\\
g(x)-g(\alpha^{2})-g'(\alpha^{2})(x-\alpha^{2})
&\le -g(\alpha)^{2}\Bigl(-\log\tfrac{x}{\alpha^{2}}+\tfrac{1}{\alpha^{2}}(x-\alpha^{2})\Bigr),
\end{align*}
so
\[
g(x)\le 2g(\alpha^{2})-g'(\alpha^{2})\alpha^{2}
 + x\Bigl(g'(\alpha^{2})-\tfrac{g(\alpha)^{2}}{\alpha^{2}}\Bigr)
 + g(\alpha)^{2}\log\!\tfrac{x}{\alpha^{2}},\qquad x\ge\alpha^{2}.
\]
Since $g(x)\ge 0$ and the right-hand side must be bounded below as $x\to\infty$, the coefficient of $x$ must be non-negative:
\[
g'(\alpha^{2})-\frac{g(\alpha)^{2}}{\alpha^{2}}\ge 0,\qquad\forall \alpha>0.
\]
Using monotonicity of $g$ we have $g(\sqrt{\alpha})\ge g(\alpha)$ for $0\le\alpha\le 1$, and hence
\[
g'(\alpha)\ge \frac{g(\sqrt{\alpha})^{2}}{\alpha}\ge \frac{g(\alpha)^{2}}{\alpha},\qquad 0<\alpha\le 1.
\]
Thus
\[
\frac{g'(\alpha)}{g(\alpha)^{2}}\ge \frac{1}{\alpha},\qquad 0<\alpha\le 1,
\]
and integrating from $\eta$ to $\alpha$ gives
\[
-\frac{1}{g(\alpha)}+\frac{1}{g(\eta)}\ge \log\frac{\alpha}{\eta},\qquad 0<\eta\le\alpha\le 1.
\]
Fix $\alpha>0$ and let $\eta\to 0^{+}$. The right-hand side tends to $+\infty$, so necessarily $1/g(\eta)\to +\infty$, i.e. $g(\eta)\to 0$ as $\eta\to 0^{+}$. Hence $g(0)=0$.

\item Suppose $g(1)=\lambda(1-\lambda)$ with $0\le\lambda\le \tfrac{1}{2}$ and define
\[
\mathcal{T}_{1}^{-}
=\Bigl\{h:[0,\infty)\to[0,\infty):x^{2}h''(x)+h(1)h(x)\le 0\Bigr\}.
\]
Clearly $\mathcal{T}^{-}\subseteq\mathcal{T}_{1}^{-}$. Let $g\in\mathcal{T}_{1}^{-}$ and write $g(x)=x^{\lambda}Q(x)$. Then
\[
x^{2}g''(x)+\lambda(1-\lambda)g(x)
= x^{\lambda}\bigl[xQ''(x)+2\lambda Q'(x)\bigr]\le 0,
\]
so
\[
xQ''(x)+2\lambda Q'(x)\le 0.
\]
For $0\le\lambda<\tfrac{1}{2}$, Lemma~\ref{good:lemma} implies that $Q$ is increasing.  
Similarly, writing $g(x)=x^{1-\lambda}\hat Q(x)$ yields
\[
x\hat Q''(x)+2(1-\lambda)\hat Q'(x)\le 0,
\]
and since $1-\lambda>\tfrac{1}{2}$, Lemma~\ref{good:lemma} implies that $\hat Q$ is decreasing. Translating these back to $g$ shows that
\[
\frac{g(x)}{x^{\lambda}} \text{ is increasing and }
\frac{g(x)}{x^{1-\lambda}} \text{ is decreasing on }(0,\infty).
\]
In the borderline case $\lambda=\tfrac{1}{2}$, Lemma~\ref{good:lemma} forces $Q$ to be constant, so $g(x)/x^{1/2}$ is constant and the same monotonicity statements follow (with “increasing’’ replaced by “constant’’ where appropriate).

\item From the previous item, taking $\lambda=\tfrac12$ gives $g(1)=\tfrac14$ and shows that $g(x)/\sqrt{x}$ must be constant. Hence $g(x)=\tfrac14\sqrt{x}$, and the converse implication is immediate.

\item Now assume $g(1)>\tfrac14$. Since $g\in\mathcal{T}^{-}$, item~2 of Lemma~\ref{triv} implies that
\[
\hat g(x):=\frac{g(x)}{4g(1)}\in\mathcal{T}^{-}.
\]
Because $\hat g(1)=\tfrac14$, the previous item yields $\hat g(x)=\tfrac14\sqrt{x}$, so
\[
g(x)=4g(1)\,\hat g(x)=g(1)\sqrt{x}.
\]
Substituting this form into the differential inequality shows that necessarily $g(1)\le\tfrac14$, a contradiction. Therefore $g(1)\le\tfrac14$.
\item From item~2, for $x\ge 1$ we have
\[
g(1)x^{\lambda}\le g(x)\le g(1)x^{1-\lambda},
\]
with some $\lambda\in[0,\tfrac{1}{2}]$. Thus $g(x)\to+\infty$ as $x\to+\infty$, while
\[
0\le \frac{g(x)}{x}\le g(1)x^{-\lambda}\to 0\quad(x\to\infty),
\]
so $\lim_{x\to\infty}g(x)/x=0$. Moreover, since $g(x)/x^{1-\lambda}$ is decreasing and positive, the product
\[
\frac{g(x)}{x}
= \frac{g(x)}{x^{1-\lambda}}\cdot \frac{1}{x^{\lambda}}
\]
is also decreasing on $(0,\infty)$.
\end{enumerate}
\end{proof}

\begin{lemma}\label{good:lemma}
Suppose \(Q(x)\) is a non-negative function on \([0,\infty)\) and satisfies the inequality \(x Q''(x) + \beta Q'(x) \leq 0\), where \(\beta\) is a non-negative constant. Then
    \begin{align*}
       &Q(x)~~ \text{is increasing when $0\leq \beta <1$}, \\ 
       &Q(x)~~ \text{is decreasing when $\beta >1$},\\
       &Q(x)~~ \text{is constant when $\beta =1$}.
    \end{align*}
    \end{lemma}

\begin{proof}
Multiplying the inequality \(xQ''(x)+\beta Q'(x)\le 0\) by \(x^{\beta-1}\) gives
\[
(x^{\beta}Q'(x))' \le 0,
\]
so \(x^{\beta}Q'(x)\) is decreasing. Hence, for any fixed \(a>0\),
\[
x^{\beta}Q'(x)\le a^{\beta}Q'(a)\quad (x\ge a),
\qquad
x^{\beta}Q'(x)\ge a^{\beta}Q'(a)\quad (x\le a).
\]
Integrating, for \(\beta\ne1\),
\begin{align}
Q(x)&\le Q(a)+a^{\beta}Q'(a)\frac{x^{1-\beta}-a^{1-\beta}}{1-\beta},\quad x\ge a,\label{b:1}\\
Q(x)&\le Q(a)+a^{\beta}Q'(a)\frac{x^{1-\beta}-a^{1-\beta}}{1-\beta},\quad x\le a.\label{b:2}
\end{align}
For \(\beta=1\) we obtain
\begin{align}
Q(x)&\le Q(a)+aQ'(a)\log\!\frac{x}{a},\quad x\ge a,\label{b=1:1}\\
Q(x)&\le Q(a)+aQ'(a)\log\!\frac{x}{a},\quad x\le a.\label{b=1:2}
\end{align}

Since \(Q(x)\ge0\), the right-hand sides in \eqref{b:1}–\eqref{b=1:2} must be bounded below for all admissible \(x\).

\emph{Case \(\mathbf{0\le\beta<1}\).}  
Fix \(a>0\) and let \(x\to\infty\) in \eqref{b:1}. Then
\(
\frac{x^{1-\beta}-a^{1-\beta}}{1-\beta}\to+\infty
\),
so non-negativity of \(Q(x)\) forces \(Q'(a)\ge0\). Thus \(Q\) is increasing.

\emph{Case \(\mathbf{\beta>1}\).}  
Fix \(a>0\) and let \(x\to0^{+}\) in \eqref{b:2}. Now
\(
\frac{x^{1-\beta}-a^{1-\beta}}{1-\beta}\to-\infty
\),
so non-negativity of \(Q(x)\) implies \(Q'(a)\le0\). Hence \(Q\) is decreasing.

\emph{Case \(\mathbf{\beta=1}\).}  
From \eqref{b=1:1}, fixing \(a\) and letting \(x\to\infty\) gives \(Q'(a)\ge0\); from \eqref{b=1:2}, fixing \(a\) and letting \(x\to0^{+}\) gives \(Q'(a)\le0\). Thus \(Q'(a)=0\) for all \(a>0\), so \(Q\) is constant.
This completes the proof.
\end{proof}
Assume that $g(x)\in \mathcal{T}^{-}$; that is, for all $\alpha\geq 0$, we have
$
x^2g''(x) + g(\alpha)g\left(\frac{x}{\alpha}\right) \le 0.
$
Because this holds for every $\alpha$, we deduce
\[
x^2g''(x) + g(\alpha)g\left(\frac{x}{\alpha}\right) \le x^2g''(x) + \max_{\alpha \geq 0} g(\alpha)g\left(\frac{x}{\alpha}\right) \le 0.
\]
Therefore, the inequality
\[
x^2g''(x) + \max_{\alpha \geq 0} g(\alpha)g\left(\frac{x}{\alpha}\right) \le 0
\]
implies
\[
x^2g''(x) + g(\alpha)g\left(\frac{x}{\alpha}\right) \le 0.
\]
Define
\[
\mathtt{G}_g(x) = \max_{\alpha \geq 0} g(\alpha)g\left(\frac{x}{\alpha}\right).
\]
If one can compute $\mathtt{G}_g(x)$ explicitly, then working with the inequality
\[
x^2g''(x) + \mathtt{G}_g(x) \le 0
\]
is more tractable than the functional inequality for all $\alpha$.

In the following lemma, we show that under certain conditions, $\mathtt{G}_g(x) = g^2\left(\sqrt{x}\right)$.
\begin{lemma}\label{equ:cond:log-}
Assume that $g$ is a non-negative, concave function on $(0,\infty)$ and fix $x>0$.
Define
\[
E(\alpha)=g(\alpha)\,g\!\left(\frac{x}{\alpha}\right),\qquad \alpha>0.
\]
Suppose further that
\[
\lim_{\alpha\to 0^{+}}E(\alpha)
=\lim_{\alpha\to+\infty}E(\alpha)=0,
\]
and that the function
\[
T_{0}(t):=\frac{t g'(t)}{g(t)}
\]
is one-to-one on $(0,\infty)$. Then $E$ attains its (unique) global maximum at
\[
\alpha=\sqrt{x}.
\]
\end{lemma}

\begin{proof}
If $g\equiv 0$, the lemma is trivial. Assume $g$ is not identically zero.
Then $E(\alpha)$ is continuous and non-negative, and
\[
\lim_{\alpha\to 0^{+}}E(\alpha)
=\lim_{\alpha\to +\infty}E(\alpha)=0,
\]
so $E$ attains at least one maximum on $(0,\infty)$; any maximizer
satisfies $E'(\alpha)=0$ with $0<\alpha<\infty$.

Differentiating,
\begin{align}
  E'(\alpha)
= g'(\alpha)g\!\left(\tfrac{x}{\alpha}\right)
  - \frac{x}{\alpha^{2}}g(\alpha)g'\!\left(\tfrac{x}{\alpha}\right)
  =0. \label{ab:1}  
\end{align}
At a maximizer we must have $g(\alpha)>0$ and
$g\!\left(\tfrac{x}{\alpha}\right)>0$ (otherwise $E(\alpha)=0$ and cannot
beat nearby positive values), so we may divide by
$g(\alpha)g\!\left(\tfrac{x}{\alpha}\right)$ in \eqref{ab:1}, obtaining
\begin{align}
 \alpha\,\frac{g'(\alpha)}{g(\alpha)}
=\frac{x}{\alpha}\,\frac{g'\!\left(\tfrac{x}{\alpha}\right)}
                        {g\!\left(\tfrac{x}{\alpha}\right)}.
\label{ab:2}   
\end{align}
Then \eqref{ab:2} is
$
T_{0}(\alpha)=T_{0}\!\left(\tfrac{x}{\alpha}\right).
$
By assumption, $T_{0}$ is one-to-one on $(0,\infty)$, so
$
\alpha=\frac{x}{\alpha}$
and hence the unique critical point is $\alpha=\sqrt{x}$.
\end{proof}
\begin{example}\label{ex:equ:cond:log-}
Consider functions of the form $g_1(x)=x^{\lambda}\psi(\log x)$ for $x\ge 1$, where $\psi:[0,\infty)\to[0,\infty)$ is concave. Then the map
\[
x\mapsto x\frac{g_1'(x)}{g_1(x)}
\]
is one-to-one (in fact, decreasing), provided that $t\mapsto \psi'(t)/\psi(t)$ is decreasing. Indeed,
\[
x\frac{g_1'(x)}{g_1(x)}
= \lambda + \frac{\psi'(\log x)}{\psi(\log x)},
\]
and the second term is decreasing in $x$ by assumption on $\psi$. Thus the whole expression is decreasing.
\end{example}
\subsection{Examples of functions in $\mathcal{T}^{-}$}\label{log-:app:ex}
\begin{lemma}
    Let $T:\mathbb{R} \to [0,1]$ be a bounded function such that $\sup_{x\in\mathbb{R}} |T'(x)| = \widetilde{n}$ and $\sup_{x\in\mathbb{R}} |T''(x)| = \widetilde{m}$. Suppose $a$ and $b$ are positive numbers satisfying
    \begin{align}
        &0 < b < \frac{s(1-s)}{\widetilde{m} + \widetilde{n} |2s-1|},\label{in:1:b}\\
        &0 < a \le \frac{s(1-s) - b\bigl(\widetilde{m} + \widetilde{n}|2s-1|\bigr)}{s(1-s)(1+b)^2},\label{in:2:a}
    \end{align}
    for some $s \in [0,1]$. Then the function
    \begin{align}
        g_{T,s}(x) = a\,s(1-s)\,x^{s}\bigl(1 + b\,T(\log x)\bigr)
    \end{align}
    belongs to $\mathcal{T}^{-}$.
\end{lemma}
\begin{remark}
If we define
\begin{align}
    b_s &= \frac{s(1-s)}{2\bigl(\widetilde{m} + \widetilde{n}|2s-1|\bigr)}, \\
    0 &< a \le \frac{32\widetilde{m}^2}{(8\widetilde{m}+1)^2} \le \frac{1}{2(1+b)^2},
\end{align}
then the main conditions \eqref{in:1:b}–\eqref{in:2:a} are satisfied, and the parameter $a$ can be chosen independently of $s$.
\end{remark}
\begin{proof}
It is immediate that $g_{T,s}(x)$ is a non-negative function. By substituting $g$ into the inequality
\[
    x^2 g''(x) + g(\alpha)\,g\!\left(\tfrac{x}{\alpha}\right) \le 0,
\]
we obtain
\begin{align}
    x^2 g_{T,s}''(x) + g_{T,s}(\alpha)\,g_{T,s}\!\left(\tfrac{x}{\alpha}\right)
    &= a s(1-s) x^{s}\Bigl( s(s-1)\bigl(1 + b T(\log x)\bigr)
        + b(2s-1) T'(\log x) + b T''(\log x) \Bigr) \nonumber\\
    &\quad + a s(1-s)\bigl(1 + b T(\alpha)\bigr)\bigl(1 + b T(\tfrac{x}{\alpha})\bigr) \nonumber\\
    &\le a s(1-s)x^{s}\Bigl( s(s-1) + b~ \widetilde{n} |2s-1| + b~ \widetilde{m} 
        + a s(1-s)(1+b)^{2} \Bigr) \le 0. \label{in:3}
\end{align}
The inequality \eqref{in:3} holds provided that
\begin{align}
    a \le \frac{s(1-s) - b\bigl(\widetilde{m} + \widetilde{n}|2s-1|\bigr)}{s(1-s)(1+b)^2}.
\end{align}
Since $a$ is assumed to be positive, this implies that the conditions \eqref{in:1:b} and \eqref{in:2:a} are necessarily satisfied.
\end{proof}
\begin{example}
    Some possible choices of $T$ are:
    \begin{itemize}
        \item $T(x) = \dfrac{1 + \sin x}{2}$, for which $\widetilde{m} = \widetilde{n} = 0.5$.
        \item $T(x) = \dfrac{1 + \cos x}{2}$, for which $\widetilde{m} = \widetilde{n}= 0.5$.
        \item $T(x) = 0.5 + \dfrac{1}{\pi}\tan^{-1}(x)$, for which $\widetilde{n} = \dfrac{1}{\pi}$ and $\widetilde{m} = \dfrac{\sqrt{27}}{8\pi}$.
    \end{itemize}
\end{example}
In the following example, the goal is to construct a corresponding non-negative function $f_s$ with $f_s(1)=1$ satisfying $-x^2 f_{s}''(x) = g_{T,s}(x)$ for the choice $T(x)=\frac{1+\sin(x)}{2}$.
\begin{example}
Let $T(x)=\dfrac{1+\sin x}{2}$ and $s\in[0,1]$. Then $\widetilde{m}=\widetilde{n}=\dfrac{1}{2}$. Define
\begin{align}
    b_s &= \frac{s(1-s)}{2\bigl(\widetilde{m} + \widetilde{n}|2s-1|\bigr)}, \\
    \Delta_s &= \frac{b_s}{2} + \frac{b_s\,s(1-s)(2s-1)}{2\bigl((s^2 - s - 1)^2 + (2s-1)^2\bigr)}, \\
    0 < a &< \min\left\{\frac{32\widetilde{m}^2}{(8\widetilde{m}+1)^2},\,1,\,\frac{1}{\sup_{0\le s\le 1}\bigl(1+\Delta_s\bigr)}\right\} = 0.32.
\end{align}
Define $\psi_s(x)$ by
\begin{align}
    \psi_s(x) = a + \frac{a b_s}{2} + T_s\bigl((1 + s - s^2)\sin x + (2s-1)\cos x\bigr),
\end{align}
where
\begin{align}
    T_s = \frac{a b_s s(1-s)}{2\bigl((s^2 - s - 1)^2 + (2s-1)^2\bigr)}.
\end{align}
Let $\widetilde{a} \in [0,\,1 - \psi_s(0)]$. Then
\begin{align}
    f_s(x) = x^s \psi_s(\log x) + \widetilde{a}x + 1 - \widetilde{a} - \psi_s(0).
\end{align}
\end{example}
\begin{remark}
In the above example, it is clear that $\psi_s(x)$ is a positive and bounded function. In other words, there exist constants $0 < \alpha_1 < \alpha_2 < +\infty$ such that
\[
    \alpha_1 \leq \psi_s(x) \leq \alpha_2
\]
for all $s \in [0,1]$ and $x \in \mathbb{R}$.
\end{remark}

\begin{lemma}\label{ga:fun}
Let $a$ be a positive number and let $0 < \lambda < \gamma < 1$. Define three real numbers $b,c,d$ by
\begin{align}
    b &= \frac{(\lambda - \gamma)(1 - \lambda)}{\gamma},\\
    c &= \frac{\lambda(1 - \lambda)}{\gamma(1 - \gamma)},\\
    d &= \frac{\lambda(\lambda - \gamma)}{1 - \gamma}.
\end{align}
Define the function $g_a:[0,\infty)\to\mathbb{R}$ by
\[
    g_a(x) = a
    \begin{cases}
        x^{\lambda} + b, & x \ge 1,\\[2mm]
        c x^{\gamma} + d x, & 0 \le x \le 1,
    \end{cases}
    =
    \begin{cases}
        g_{1,a}(x), & x \ge 1,\\[1mm]
        g_{2,a}(x), & 0 \le x \le 1.
    \end{cases}
\]
If
\[
    0 < a < \min\left\{ \lambda(1 - \lambda), \,\frac{\gamma^{2}(1 - \gamma)^{2}}{\lambda(1 - \lambda)} \right\},
\]
then $g_a \in \mathcal{T}^{-}$.
\end{lemma}

\begin{proof}
The parameters $b,c$ and $d$ are chosen so that $g_a$ is twice continuously differentiable. We first show that $g_a(x) \ge 0$ for all $x \ge 0$. Note that $a$ and $c$ are positive, whereas $b$ and $d$ are negative. At $x=1$ we have
\[
    g_a(1) = g_{1,a}(1) = g_{2,a}(1)
    = a(c+d) = a(1+b)
    = \frac{a\lambda(1-\gamma)(1+\gamma-\lambda)}{\gamma(1-\gamma)} > 0.
\]
The first branch $g_{1,a}$ is  concave, and since $g_{1,a}(1) > 0$ and $g_{1,a}(0)=0$, it follows that $g_{1,a}(x) \geq  0$ for all $x \ge 1$.
The second branch $g_{2,a}$ is also concave, and satisfies $g_{2,a}(0) = 0$, $g_{2,a}'(0^{+}) = +\infty$ and $g_{2,a}(1) > 0$, hence $g_{2,a}(x) \ge 0$ for all $x \geq 1$. Thus $g_a(x) \ge 0$ on $[0,\infty)$.

Next, we verify that $g_a$ satisfies the conditions of Lemma~\ref{equ:cond:log-}. Since $\gamma > \lambda$, it is straightforward to check that
\[
    \lim_{\alpha \to 0} g_a(\alpha)\,g_a\!\left(\frac{x}{\alpha}\right)
    = \lim_{\alpha \to \infty} g_a(\alpha)\,g_a\!\left(\frac{x}{\alpha}\right) = 0.
\]
Define
\[
    T_0(x) = x \frac{g_a'(x)}{g_a(x)}.
\]
Then $T_0$ is one-to-one (indeed, it is decreasing), since
\[
    T_0'(x) =
    \begin{cases}
        \dfrac{b\lambda^2 x^{\lambda-1}}{(x^\lambda + b)^2}, & x \ge 1, \\[2mm]
        \dfrac{(1-\gamma)^2 d c\, x^{\gamma}}{(c x^{\gamma} + d x)^2}, & 0 < x \le 1.
    \end{cases}
\]
Therefore, to prove that $g_a \in \mathcal{T}^{-}$, it suffices to show that
\[
    x^2 g_a''(x) + g_a^2\!\bigl(\sqrt{x}\bigr) \le 0
    \quad \text{for all } x \ge 0.
\]

For $x \ge 1$ set
\[
    T_1(x) = a x^2 g_{1,a}''(x) + a^2 g_{1,a}^2\!\bigl(\sqrt{x}\bigr).
\]
A direct computation gives
\[
    T_1(x) = (a\lambda(\lambda-1)+a^2)x^{\lambda} + 2a^2 b\, x^{\lambda/2} + a^2 b^2.
\]
Since $0 < a < \lambda(1-\lambda)$ and $b<0$ by definition, the function $T_1$ is convex on $[1,\infty)$. Hence its maximum on $[1,\infty)$ is attained at $x=1$ or $x=\infty$. Clearly $T_1(\infty) \le 0$, and if we additionally impose
\[
    a \le \frac{\lambda(1-\lambda)}{(1+b)^2},
\]
then $T_1(1) \le 0$. Consequently, $T_1(x) \le 0$ for all $x \ge 1$.

For $0 \le x \le 1$ define
\[
    T_2(x) = a x^2 g_{2,a}''(x) + a^2 g_{2,a}^2\!\bigl(\sqrt{x}\bigr).
\]
This can be written as
\[
    T_2(x)
    = \bigl(a c \gamma(\gamma-1) + a^2 c^2\bigr) x^{\gamma}
      + 2 a^2 c d\, x^{(\gamma+1)/2} + a^2 d^2 x.
\]
If we require
\[
    a \le \frac{\gamma^2(1-\gamma)^2}{\lambda(1-\lambda)},
\]
 the coefficient of $x^{\gamma}$ is non-positive. Since $d<0$, this implies that $T_2$ is convex on $[0,1]$. Hence the maximum of $T_2$ on $[0,1]$ is attained at $x=0$ or $x=1$. We have $T_2(0)=0$ and
\[
    T_2(1) = T_1(1) \le 0,
\]
so $T_2(x) \le 0$ for all $0 \le x \le 1$.

Combining the two cases, if
\[
    0 < a < \min\left\{ \lambda(1-\lambda), \,\frac{\gamma^2(1-\gamma)^2}{\lambda(1-\lambda)}, \,\frac{\lambda(1-\lambda)}{(1+b)^2} \right\}
    = \min\left\{ \lambda(1-\lambda), \,\frac{\gamma^2(1-\gamma)^2}{\lambda(1-\lambda)} \right\},
\]
then
\[
    x^2 g_a''(x) + g_a^2\!\bigl(\sqrt{x}\bigr) \le 0
    \quad \text{for all } x \ge 0.
\]
This shows that $g_a \in \mathcal{T}^{-}$, which completes the proof.
\end{proof}
In the following example, we construct a non-negative function $f_a(x)$ with $f_a(1)=1$ satisfying 
\[
    -x^2 f_a''(x) = g_a(x),
\]
where $g_a$ is defined in Lemma~\ref{ga:fun}.
\begin{example}
The function $f_a:[0,\infty)\to\mathbb{R}$ is given by
\begin{align*}
f_a(x) = 1 +
\begin{cases}
    a\left(\dfrac{x^\lambda}{\lambda(1-\lambda)} + b \log x + \widetilde{a}(x-1)
      - \dfrac{1}{\lambda(1-\lambda)}\right), & x \ge 1, \\[2mm]
    a\left(\dfrac{c x^\gamma}{\gamma(1-\gamma)} - d x \log x
      + \Bigl(b + \dfrac{1}{1-\lambda} + \widetilde{a} - \dfrac{c}{1-\gamma} + d\Bigr)(x-1)
      - \dfrac{c}{\gamma(1-\gamma)}\right), & 0 < x < 1.
\end{cases}
\end{align*}
Each branch of $f_a$ is concave and satisfies $f_a(1)=1$, and the minimum of $f_a$ is attained at $x=0$ or $x=+\infty$. To ensure $f_a(x)\ge 0$ on $[0,\infty)$, it is necessary and sufficient that
\[
    f_a(0) \ge 0
    \quad\text{and}\quad
    f_a(+\infty) \ge 0.
\]
The condition $f_a(+\infty)\ge 0$ requires $\widetilde{a} \ge 0$, while $f_a(0)\ge 0$ is equivalent to
\[
    \widetilde{a} \le \frac{1}{a} - d - b - \frac{c}{1-\gamma} - \frac{1}{1-\lambda}.
\]
Therefore, choosing the parameter $A$ so that
\[
    0 \le \widetilde{a} \le \frac{1}{a} - d - b - \frac{c}{1-\gamma} - \frac{1}{1-\lambda}
\]
guarantees that $f_a(x)$ is non-negative on $[0,\infty)$.
 Choose $a^{*}>0$ such that, for every $0<a\le a^{*}$,
\begin{align*}
   \frac{1}{a} - d - b - \frac{c}{1-\gamma} - \frac{1}{1-\lambda} \;\ge\; 0.
\end{align*}
Thus, if we choose
\[
0<a< \min\left\{ a^{*},\,\lambda(1-\lambda),\,\frac{\gamma^{2}(1-\gamma)^{2}}{\lambda(1-\lambda)} \right\},
\]
we guarantee that $f_a(x)\in \mathcal{T}^{-}$. 
\end{example}
\begin{lemma}
    Let $0<\lambda<\gamma<1$ and $c>0$. Set
    \begin{align}
        \delta &= \dfrac{3 - 2\lambda + \gamma(2\lambda - \gamma - 2)}{(\gamma-\lambda)(1-\lambda)(1-\gamma)} + 1,\\
        a &= \dfrac{\lambda(1-\lambda)(\delta-1) + 3 - 2\lambda}{\gamma(1-\gamma)}\,c \;\triangleq\; \hat{a}\,c,\\
        \theta &= \dfrac{\lambda(\delta-1)(\lambda-\gamma) + 2\lambda - \gamma - 2}{1-\gamma}\,c,\\
        b &= \delta c.
    \end{align}
    Define the function $g:[0,\infty)\to\mathbb{R}$ by
    \[
        g(x)=
        \begin{cases}
            x^{\lambda}\!\left(b - \dfrac{c}{1+\log x}\right), & x \ge 1,\\[2mm]
            a x^{\gamma} + \theta x, & 0 \le x \le 1,
        \end{cases}
        =
        \begin{cases}
            g_1(x), & x \ge 1,\\
            g_2(x), & 0 \le x \le 1.
        \end{cases}
    \]
    Assume that $\lambda,\gamma$ and $c$ satisfy
    \begin{align}\label{ex:3:ineq}
        \lambda &< \tfrac{1}{2},\qquad
        c < \tfrac{1-2\lambda}{4},\qquad
        c < \frac{\lambda(1-\lambda)}{\delta},\\
        c &\ge \frac{\lambda(1-\lambda)}{2\delta},\qquad
        c \le \frac{\gamma(1-\gamma)}{\hat{a}},\qquad
        c \le \frac{\hat{a}\gamma(1-\gamma)}{(\delta-1)^2}.
    \end{align}
    Then $g \in \mathcal{T}^{-}$.
\end{lemma}

\begin{proof}
Observe first that $\delta \ge 1$. Under the above definitions, $g$ is non-negative, twice continuously differentiable, and concave. Furthermore, since $\lambda < \gamma$, one has
\[
    \lim_{\alpha \to 0} g(\alpha)\,g\!\left(\tfrac{x}{\alpha}\right)
    =
    \lim_{\alpha \to \infty} g(\alpha)\,g\!\left(\tfrac{x}{\alpha}\right)
    = 0.
\]
By Example~\ref{ex:equ:cond:log-} and Lemma~\ref{ga:fun}, the functions
\[
    x \mapsto x \frac{g_1'(x)}{g_1(x)}
    \quad\text{and}\quad
    x \mapsto x \frac{g_2'(x)}{g_2(x)}
\]
are decreasing on their respective domains. Hence
\[
    x \mapsto x \frac{g'(x)}{g(x)}
\]
is injective. Lemma~\ref{equ:cond:log-} therefore reduces the verification $g \in \mathcal{T}^{-}$ to proving
\[
    x^2 g''(x) + g^{2}\!\bigl(\sqrt{x}\bigr) \le 0.
\]

From the proof of Lemma~\ref{ga:fun} we already know that $g_2$ satisfies
\[
    x^2 g_2''(x) + g_2^{2}\!\bigl(\sqrt{x}\bigr) \le 0, \qquad 0 \le x \le 1,
\]
provided
\[
    c \le \frac{\gamma(1-\gamma)}{\hat{a}}
    \quad\text{and}\quad
    c \le \frac{\hat{a}\gamma(1-\gamma)}{(\delta-1)^2}.
\]
Thus it remains to check the inequality for $g_1$. For $x \ge 1$ we compute
\[
    x^2 g_1''(x) + g_1^{2}\!\bigl(\sqrt{x}\bigr)
    =c~A_2(\log x) + c~A_1(\log x) + A_0,
\]
where
\begin{align}
    A_0 &= b\,\lambda(\lambda-1) + b^2,\\
    A_1(x) &= \frac{\bigl(\tfrac{\lambda(1-\lambda)}{2} - 2b\bigr) x + \lambda(1-\lambda) - 2b}{(1+x)\bigl(1 + x/2\bigr)},\\
    A_2(x) &= \frac{\bigl(c + \tfrac{2\lambda-1}{4}\bigr)x^3
                 + \bigl(\tfrac{12c + 10\lambda - 7}{4}\bigr)x^2
                 + (3c + 4\lambda - 4)x
                 + 2\lambda + c - 3}
                {(1+x)^3\bigl(1 + x/2\bigr)^2}.
\end{align}
Since $\log x \ge 0$ for $x \ge 1$, it suffices to ensure that all coefficients of $x^i$, $i=0,1,2,3$, appearing in $A_0$, $A_1$ and $A_2$ are non-positive. The bounds in \eqref{ex:3:ineq} on $c$ and $\lambda$ are precisely obtained by enforcing these coefficients to be non-positive. Under these conditions, one obtains
\[
    x^2 g_1''(x) + g_1^{2}\!\bigl(\sqrt{x}\bigr) \le 0
    \quad\text{for all } x \ge 1.
\]

Combining the two ranges $0 \le x \le 1$ and $x \ge 1$, we conclude that
\[
    x^2 g''(x) + g^{2}\!\bigl(\sqrt{x}\bigr) \le 0
    \quad\text{for all } x \ge 0,
\]
and hence $g \in \mathcal{T}^{-}$, as claimed.
\end{proof}
  
\begin{example}
    Let $\lambda = 0.2$, $\gamma = 0.25$, $c = 0.0017$, $b = 0.12$, $a = 0.124$ and $\theta = -0.0057$. These parameter values satisfy the required conditions and therefore yield a function $g$ belonging to $\mathcal{T}^{-}$.
    \end{example}
  \begin{lemma}
    Let  $\frac{1}{2}<s< 1$. If $a$ satisfies the following inequality 
\begin{align}
  0<  a&\leq (1-s)\min\left\{s, 4s-2\right\}.\label{coef:ch:1}\\
\end{align}
Then function of the form
\begin{align}
 g_s(x) = a\, \frac{x^s}{1 + x^{1-s}}\in \mathcal{T}^{-}.\label{g_sa:func}    
\end{align}
\end{lemma}
\begin{proof}
    it is easy to check that function $g_s(x)$ satisfies the conditions in Lemma \ref{equ:cond:log-}. Thus we check only the inequality  $x^2g_{s}^{''}(x)+g_{s}^{2}(\sqrt{x})\leq 0$.
    Let $u^2=x^{1-s}$. The inequality $x^2g_{s}^{''}(x)+g_{s}^{2}(\sqrt{x})\leq 0$ simplifies into following inequality:
    \begin{align}
        c_6u^6+c_5u^5+\cdots+c_1u+c_0\leq 0 ~~~\forall u\geq 0;
    \end{align}
    where 
    \begin{align}
        c_6&=(s-1)(4s-2)+a\\
        c_5&=2(s-1)(4s-2)\\
        c_4&=(s-1)(7s-2)+3a\\
        c_3&=6(s-1)s\\
        c_2&=4s(s-1)+3a\\
        c_1&=2s(s-1)\\
        c_0&=s(s-1)+a.
    \end{align}
    The condition \eqref{coef:ch:1} yields non-positivity of $c_i$. This completes the proof. 
\end{proof}
In the following example, we construct a non-negative function $f_s(x)$ with $f_s(1)=1$ satisfying 
\[
    -x^2 f_s''(x) = g_s(x),
\]
where $g_s$ is defined in Lemma~\ref{g_sa:func}.
\begin{example}
    The function $f_s(x)$ is defined by
    \begin{align}
        f_s(x) = a\,\psi_s(x) + \widetilde{a}(x-1) + 1 - a\,\psi_s(1),
    \end{align}
    where the non-negative function $\psi_s(x)$ is given by
    \begin{align}
        \psi_s(x)
        = \frac{x^s}{s(1-s)} - \frac{1}{1-s}\int_{0}^{x}\log\bigl(1+t^{\,s-1}\bigr)\,dt.
    \end{align}
    To guarantee that $f_s(x)$ belongs to $\mathcal{T}^{-}$, it is sufficient to choose the parameters so that
    \begin{align}
        &0\leq \widetilde{a} \le 1 - a\psi_s(1),\\
        &0 < a \leq \min\left\{(1-s)s,~ (1-s)(4s-2),~ \frac{1}{\psi_s(1)}\right\} .
    \end{align}
\end{example}

\section{Proof of Theorem~\ref{main:313:theoremn}}\label{er:expo:proof}
\begin{proof}
Let $G(x) = -\log(1-x)$ and  $f_s$ be a non-negative concave function satisfying the conditions of the theorem.
Let $\hat{f}_s(x)=x\widetilde{f}_s(1/x)=1-f_s(x)$. Then, 
\begin{align}
 \mathtt{E}_{f-\mathrm{sp}}(R)
  &=
 \sup_{p_X}\sup_{s\in(0,1)} \inf_{q_Y}
  \Biggl[
\frac{\sum_{x} p_X(x) \mD_{G,\hat f_s}(q_Y\|W_{Y|X=x})}{1-s}
    - \frac{s}{1-s}\,R
  \Biggr]
  \\&=
\sup_{p_X}\sup_{s\in(0,1)} \inf_{q_Y}
  \Biggl[
-\frac{\sum_{x} p_X(x) \log\left(D_{f_s}(q_Y\|W_{Y|X=x})\right)}{1-s}
    - \frac{s}{1-s}\,R
  \Biggr]\label{second_form:ef_sp}
\end{align} 
Here is the proof for the second representation of $\mathtt{E}_{f-\mathrm{sp}}$, as defined in \eqref{second_form:ef_sp}. Thus, in the remainder of the proof we work with $f_s(x)$ and $\hat{f}_s(x)=1-f_s(x)$ instead of $\widetilde{f}_s(x)$.

We generalize the proof of Shannon–Gallager–Berlekamp. As in the original proof of the sphere–packing by Shannon–Gallager–Berlekamp, we assume (without loss of generality) that all codewords share the same composition (type) $p_X(x)$.

\subsection{Some definitions}
    We next introduce several quantities that will be used throughout this sub-section.

\begin{definition}\label{def:nesse}
    For two distinct full support probability distributions $P_1$ and $P_2$  on $\mathcal{Y}^n$, define
    \begin{align}
        Q_{Y^n}(y^n)
        &\triangleq
        \frac{
            P_1(y^n)\,f_s\!\bigl(\tfrac{P_2(y^n)}{P_1(y^n)}\bigr)
        }{
            \displaystyle \sum_{z^n} P_1(z^n)\,
            f_s\!\bigl(\tfrac{P_2(z^n)}{P_1(z^n)}\bigr)
        }.
    \end{align}
    Moreover, define the scalar $\mu$ and the functions
    $\Theta: \mathcal{Y}^n \to \mathbb{R}$,
    $a : \mathbb{R} \to \mathbb{R}$,
    $b : \mathbb{R} \to \mathbb{R}$ by
    \begin{align}
        \mu
        &\triangleq
        \log \sum_{y^n}P_1(y^n)\,
            f_s\!\bigl(\tfrac{P_2(y^n)}{P_1(y^n)}\bigr)=-\mD_{G,\hat f_s}(P_2\|P_1),\\
        \Theta(y^n)
        &\triangleq
        \log\!\bigl(\tfrac{P_2(y^n)}{P_1(y^n)}\bigr),\\
        a(x)
        &\triangleq
        -\log\!\bigl(f_s(e^{x})\bigr).
    \end{align}
\end{definition}
Define the typical set
\begin{equation}
\tilde{\mathcal{Y}}^n\!=
 \Bigl\{y^n\in \mathcal{Y}^{n}\!:\!
  \bigl|\Theta(y^n)-\mathbb{E}_{Q_{Y^n}}[\Theta]\bigr|
        \le\sqrt{2\,\operatorname{Var}_{Q_{Y^n}}[\Theta]}\Bigr\}.
\label{eq:Ys}
\end{equation}
Let
\begin{equation}
\mathcal{W}\!=
 \Bigl\{x\!:\!
  \bigl|x-\mathbb{E}_{Q}[\Theta]\bigr|
        \le\sqrt{2\,\operatorname{Var}_{Q}[\Theta]}\Bigr\}.
\label{eq:W}
\end{equation}
Set
\begin{equation}
\ua\;=\;\min_{x\in\mathcal{W}}a(x),
\quad
\ub\;=-\ua+\min_{x\in\mathcal{W}}\bigl[a(x)+x\bigr].
\label{eq:def‑ua‑ub}
\end{equation}

\bigskip
Note that for all $y^n\in\mathcal{Y}^n$,
\begin{equation}
P_{1}(y^n)=Q_{Y^n}(y^n)\,e^{\mu+a(\Theta(y^n))},
\qquad
P_{2}(y^n)=Q_{Y^n}(y^n)\,e^{\mu+a(\Theta(y^n))+\Theta(y^n)}.
\label{eq:P1‑P2‑Qs}
\end{equation}
Moreover, for every $y^n\in\tilde{\mathcal{Y}}^n$,
\begin{equation}
P_{1}(y^n)\geq Q_{Y^n}(y^n)\,e^{\mu+\ua},
\qquad
P_{2}(y^n)\geq Q_{Y^n}(y^n)\,e^{\mu+\ua+\ub}.
\label{eq:P‑lower‑typical}
\end{equation}
The function
\[
a(x) = -\log\bigl(f_s(e^x)\bigr)
\]
is non-increasing since
\[
a'(x) = -\frac{e^xf'_s(e^x)}{f(e^x)}
\]
is non-positive (any non-negative concave function on $[0,\infty)$ must be non-decreasing).

Let $\ell(x)\triangleq a(x)+x$. Next, we show that  $\ell(x)$ is non-decreasing. The first-order inequality for concave $f_s$ implies
\[
f_s(x_2) + (x_1-x_2) f_{s}^{'}(x_2)\geq f_s(x_1), \qquad \forall x_1\neq x_2.
\]
Applying the above inequality at $(x_1,x_2) = (0,e^x)$, we obtain
\begin{align*}
\ell'(x)
&= a'(x) + 1
 = \frac{f_s(e^x) - e^x f_{s}^{'}(e^x)}{f_s(e^x)}
\geq \frac{f_s(0)}{f_s(e^x)}
 = 0,
\end{align*}
where we used the fact that, since $s\in(0,1)$ and the function $\psi_{s}(\log x)$ is bounded, we have
\[
f_s(0)\triangleq \lim_{x\to 0^{+}} f_s(x) = 0.
\]
Hence $\ell(x)$ is non-decreasing.
Consequently, in the interval
\[
\mathcal{W}
=
\Bigl[\mathbb{E}_{Q}[\Theta] - \sqrt{2\,\operatorname{Var}_{Q}[\Theta]},\;
      \mathbb{E}_{Q}[\Theta] + \sqrt{2\,\operatorname{Var}_{Q}[\Theta]}\Bigr],
\]
the minimum of $a(x)$ is attained at
$\mathbb{E}_{Q}[\Theta] + \sqrt{2\,\operatorname{Var}_{Q}[\Theta]}$, whereas the minimum of
$\ell(x)$ is attained at
$\mathbb{E}_{Q}[\Theta] - \sqrt{2\,\operatorname{Var}_{Q}[\Theta]}$. Thus
\begin{align}
   \ua
   &= a\!\left(
        \mathbb{E}_{Q_{Y^n}}[\Theta]
        + \sqrt{2\,\operatorname{Var}_{Q_{Y^n}}[\Theta]}
      \right), \\
   \ua + \ub
   &= \ell\!\left(
        \mathbb{E}_{Q_{Y^n}}[\Theta]
        - \sqrt{2\,\operatorname{Var}_{Q_{Y^n}}[\Theta]}
      \right).
   \label{ua:ub}
\end{align}
\subsection{Two-point lemma}
In the next lemma, we derive a lower bound on the error probability in the special case of only two codewords, often referred to as the \emph{two-point lemma}.
\begin{lemma}[Two-point lemma]\label{lem:binary-hyp}
Let $P_{1},P_{2},Q_{Y^n},\mu,\Theta(\cdot),$ and $a(\cdot)$ be as defined in Definition~\ref{def:nesse}. For any set $\mathbb{S}\subseteq \mathcal{Y}^n$, define the following probabilities
\begin{equation}\label{eq:Pe1Pe2-def}
P_{e1} \triangleq
   \sum_{y^n\in \mathbb{S}^{c} } P_{1}(y^n),
\qquad
P_{e2} \triangleq
   \sum_{y^n\in \mathbb{S}} P_{2}(y^n).
\end{equation}
Then either
\begin{align}\label{eq:Pe1Pe2-split:1}
& P_{e1}  >  \frac{1}{4} e^{\mu+\ua},\\
&\text{or}\nonumber\\
& P_{e2}  >  \frac{1}{4} e^{\mu+\ua+\ub}
\end{align}
holds. 

\end{lemma}

\begin{proof}
Chebyshev’s inequality implies that \(Q_{Y^n}(\tilde{\mathcal{Y}}^n) \geq \tfrac{1}{2}\) where $\tilde{\mathcal{Y}}^n$ is defined in \eqref{eq:Ys}. Thus $\tilde{\mathcal{Y}}^n$ is  nonempty. We have
\begin{align}
\frac{P_{e1}}{e^{\mu+\ua}}
  &= \sum_{y^n\in \mathbb{S}^c} \frac{P_{1}(y^n)}{e^{\mu+\ua}}
     \overset{(a)}{\geq} 
    \sum_{y^n \in \tilde{\mathcal{Y}}^n\cap \mathbb{S}^{c}} Q_{Y^n}(y^n)\label{ineq:careful:1}
 \\[4pt]
\frac{P_{e2}}{e^{\mu+\ua+\ub}}
  &= \sum_{y^n\in \mathbb{S}} \frac{P_{2}(y^n)}{e^{\mu+\ua+\ub}}
  \overset{(b)}{\geq } 
   \sum_{y^n \in \tilde{\mathcal{Y}}^n\cap  \mathbb{S}} Q_{Y^n}(y^n),\label{ineq:careful:2}
\end{align}
where $(a)$ and $(b)$ follow from \eqref{eq:P‑lower‑typical}. We consider two cases:
Case (i): at least one of the inequalities $(a)$ or $(b)$ is strict (in \eqref{ineq:careful:1} and \eqref{ineq:careful:2}). Summing these inequalities gives
\[
\frac{P_{e1}}{e^{\mu+\ua}}
  + \frac{P_{e2}}{e^{\mu+\ua+\ub}}
  > Q_{Y^n}(\tilde{\mathcal{Y}}^n)
  \geq \frac{1}{2},
\]
where we used that at least one of the inequalities \eqref{ineq:careful:1} and \eqref{ineq:careful:2} is strict. Since both terms are non-negative and their sum is greater than $\tfrac{1}{2}$, at least one of them must be no smaller than $\tfrac{1}{4}$, which proves the claim.

Case (ii) both the inequalities $(a)$ or $(b)$ hold with equality (in \eqref{ineq:careful:1} and \eqref{ineq:careful:2}). Since $P_1$ and $P_2$ are strictly positive, we must have $\tilde{\mathcal{Y}}^n\cap \mathbb{S}^{c}=\mathbb{S}^c$ and $\tilde{\mathcal{Y}}^n\cap  \mathbb{S}=\mathbb{S}$. This can only happen when $\mathbb{S}^{c}\subset\tilde{\mathcal{Y}}^n$ and $\mathbb{S}\subset\tilde{\mathcal{Y}}^n$. In other words, this can only happen when $\tilde{\mathcal{Y}}^n=\mathcal{Y}^n$. In this case, $Q_{Y^n}(\tilde{\mathcal{Y}}^n)=1$. Thus, 
summing \eqref{ineq:careful:1} and \eqref{ineq:careful:2}  gives
\[
\frac{P_{e1}}{e^{\mu+\ua}}
  + \frac{P_{e2}}{e^{\mu+\ua+\ub}}
 = Q_{Y^n}(\tilde{\mathcal{Y}}^n)=1.
\]
Since both terms are non-negative and their sum is equal to $1$, at least one of them must be no smaller than $\tfrac{1}{4}$, which again proves the claim.
\end{proof}

\subsection{Application of the two-point lemma to a general codebook}

Let $\mathcal{C} = \{x_1^{n},x_2^{n},\ldots,x_{|\mathcal{M}|}^{n}\}$ be a code where all codewords share the same composition (type) $p_X(x)$. For each codeword $x_m^{n} \in \mathcal{C}$, let $\mathbb{S}_m$ denote the decision region in which the decoder outputs $\hat{M}=m$. For any message $m \in \mathcal{M}$ and given distributions $P_{1}$ and $P_{2}$ on $\mathcal{Y}^n$, define
\begin{align}
   P_{e1}^{m}
   &= \sum_{y^n \in \mathbb{S}_m^{c}} P_{1}(y^n), \\
   P_{e2}^{m}
   &= \sum_{y^n \in \mathbb{S}_m} P_{2}(y^n).
\end{align}
Since
\[
   \sum_{m=1}^{|\mathcal{M}|} P_{e2}^{m}
   = \sum_{y^n \in \mathcal{Y}^n} P_{2}(y^n) = 1,
\]
there exists at least one message $m^{*}$ such that
\[
   P_{e2}^{m^{*}} \le \frac{1}{|\mathcal{M}|} = e^{-nR}.
\]

For this message $m^{*}$, let the corresponding codeword be
\[
   x_{m^*}^{n} = (x_1^{*},x_2^{*},\ldots,x_n^{*}),
\]
and define
\begin{align}
     P_{1}(y^n) = \prod_{i=1}^{n} W_{Y|X}(y_i | x_i^{*}),
   \qquad
   P_{2}(y^n) = \prod_{i=1}^{n} q_Y(y_i),\label{new:p1:p2}  
\end{align}
where $q_Y$ is an arbitrary strictly positive probability distribution on $\mathcal{Y}$, i.e.,
$q_Y(y) > 0$ for all $y \in \mathcal{Y}$. Define
\[
q_{Y,\min} \triangleq \min_{y \in \mathcal{Y}} q_Y(y),
\]
and note that, under this assumption, $q_{Y,\min} > 0$.

If the regions $\{\mathbb{S}_m\}$ are chosen according to maximum-likelihood decoding, then
\[
   P_{e1}^{m^{*}} = \Pr(M \ne \hat{M} | M = m^{*}) \le P_{e}^{\max}.
\]

Using $P_{e1}^{m^{*}} \le P_{e}^{\max}$ together with $P_{e2}^{m^{*}} \le e^{-nR}$ and applying the two-point lemma (Lemma~\ref{lem:binary-hyp}) to the pair $(P_1,P_2)$ defined above, we obtain 
\begin{align}
   \text{either}\quad
   P_{e}^{\max} \ge P_{e1}^{m^{*}} > \frac{1}{4} e^{\mu_{*}+\ua_{*}},
   \quad\text{or}\quad
   e^{-nR} \ge P_{e2}^{m^{*}} > \frac{1}{4} e^{\mu_{*}+\ua_{*}+\ub_{*}},\label{expo:1}
\end{align}
where $\mu_{*},\ua_{*},\ub_{*}$ are the values of $\mu,\ua,\ub$ corresponding to the specific choice
\[
   P_{1}(y^n) = \prod_{i=1}^{n} W_{Y|X}(y_i | x_i^{*}),
   \qquad
   P_{2}(y^n) = \prod_{i=1}^{n} q_Y(y_i).
\]
We adopt the following convention: any quantity that carries a superscript or subscript “$*$” (for example, $\mu_{*}$, $P_{e1}^{m^{*}}$) is allowed to depend on the specific message $m^{*}$, whereas quantities written without “$*$” are understood to be independent of $m^{*}$.

Note that, under the choices in \eqref{new:p1:p2}, the quantity $\Theta_{*}(y^n)$ does not depend on $m^{*}$, since
\begin{align*}
\Theta_{*}(y^n)
  &= \log\!\left(\prod_{i=1}^{n} \frac{q_Y(y_i)}{W_{Y|X}(y_i | x_i^{*})}\right)
   = \sum_{i=1}^{n} \log\!\left(\frac{q_Y(y_i)}{W_{Y|X}(y_i | x_i^{*})}\right) \\
  &\overset{(a)}{=}
   n \sum_{x} p_X(x) \log\!\left(\frac{q_Y(y_i)}{W_{Y|X}(y_i | x)}\right),
\end{align*}
where $(a)$ uses the constant–composition assumption on the code $\mathcal{C}$ and  $p_X(k)$ denotes the empirical frequency of the symbol $k$ in the code.

Moreover, the subadditivity of $\mD_{G,f_s}$ for $G(x) = -\log(1-x)$ implies
\begin{align}
    \mu_{*}=-\mD_{G,\hat f_s}\left(\prod_{i=1}^{n} q_Y(y_i) \big\|\prod_{i=1}^{n} W_{Y|X}(y_i | x_i^{*}) \right)   
    \;\overset{(a)}{\ge}\;
    n \sum_{x} p_X(x) \mu_x,
    \label{new:sub:part3}
\end{align}
where
\begin{align}
    \mu_x
    \triangleq -\mD_{G,\hat{f}_s}\bigl(q_Y \,\|\, W_{Y|X=x}\bigr)
   \triangleq \log\Bigl(D_{f_s}\bigl(q_Y \,\|\, W_{Y|X=x}\bigr)\Bigr)
\end{align}
Here, the constant–composition property is again used in \eqref{new:sub:part3}. In particular, $\mu_{*}$ admits a lower bound that does not depend on $m^{*}$. Thus by plugging lower bound \eqref{new:sub:part3} in \eqref{expo:1}, we obtain
\begin{align}
   \text{either}\quad
   P_{e}^{\max}  > \frac{1}{4} \exp(\ua_{*}+n \sum_{x} p_X(x)\mu_x),
   \quad\text{or}\quad
   e^{-nR}  > \frac{1}{4} \exp(\ua_{*}+\ub_{*}+n \sum_{x} p_X(x) \mu_x).\label{expo:2}
\end{align}

Using \eqref{ua:ub}, the bound in \eqref{expo:2} becomes
\begin{align}
&\text{either}\quad
   P_{e}^{\max}
   > \frac{1}{4}
     \exp\Bigl(
         n \sum_{x} p_X(x) \mu_x
         + a\!\left(
               \mathbb{E}_{Q_{Y^n}^{*}}[\Theta]
               + \sqrt{2\,\operatorname{Var}_{Q_{Y^n}^{*}}[\Theta]}
             \right)
       \Bigr),
  \label{expo:3:1}\\
&\text{or}\quad
   e^{-nR}
   > \frac{1}{4}
     \exp\Bigl(
         n \sum_{x} p_X(x) \mu_x
         + \ell\!\left(
               \mathbb{E}_{Q_{Y^n}^{*}}[\Theta]
               - \sqrt{2\,\operatorname{Var}_{Q_{Y^n}^{*}}[\Theta]}
             \right)
       \Bigr).
   \label{expo:3:2}
\end{align}
Using the assumption that $\psi_s(\log x)\in[a,b]$ for all $s\in[0,1]$ with $0<a<b<\infty$, we can lower bound $a(x)$ as
\begin{align}
    a(x)
    &= -\log\bigl(f_s(e^x)\bigr) \\
    &= -\log\bigl(e^{s x}\psi_s(x)\bigr) \\
    &= -s x - \log\bigl(\psi_s(x)\bigr) \nonumber\\
    &\ge  -s x - \log b,
    \label{lowe:a}
\end{align}
 Similarly, for $\ell(x)=a(x)+x$ we obtain
\begin{align}
    \ell(x)
    \ge (1-s)x - \log b.
    \label{lowe:el}
\end{align}
Substituting \eqref{lowe:a} and \eqref{lowe:el} into \eqref{expo:3:1} and \eqref{expo:3:2} yields
\begin{align}
&\text{either}\quad
   P_{e}^{\max}
   >
     \exp\Bigl(
         n \sum_{x} p_X(x) \mu_x
         - s\Bigl(\mathbb{E}_{Q_{Y^n}^{*}}[\Theta]
               + \sqrt{2\,\operatorname{Var}_{Q_{Y^n}^{*}}[\Theta]}\Bigr)
         - \log b - \log 4
     \Bigr),
  \label{expo:4:1}\\
&\text{or}\quad
   e^{-nR}
   >
     \exp\Bigl(
         n \sum_{x} p_X(x) \mu_x
         + (1-s)\Bigl(\mathbb{E}_{Q_{Y^n}^{*}}[\Theta]
               - \sqrt{2\,\operatorname{Var}_{Q_{Y^n}^{*}}[\Theta]}\Bigr)
         - \log b - \log 4
     \Bigr).
   \label{expo:4:2}
\end{align}
To further study \eqref{expo:4:1} and \eqref{expo:4:2}, we need to analyze the asymptotic behavior of $\mathbb{E}_{Q_{Y^n}^{*}}[\Theta]$ and $\operatorname{Var}_{Q_{Y^n}^{*}}[\Theta]$. This is done in the following theorem.
\begin{theorem}\label{mean:var:th1}
 For functions of the form $f_s(x)=x^{s}\,\psi_s(\log x)$ with $s\in[0,1]$  satisfying assumptions of Theorem \ref{main:313:theoremn}, we have
\begin{enumerate}
    \item $\mathbb{E}_{Q_{Y^n}^{*}}\!\left(\Theta\right)
           = n\sum_{x} p_X(x)\theta_x + \mathcal{O}(\sqrt{n})$,\\[2pt]
    \item $\operatorname{Var}_{Q_{Y^n}^{*}}(\Theta)
           \le \dfrac{n}{s^2}\,\mathcal{O}(1) \triangleq \frac{v_n}{s^2}$,
\end{enumerate}
where
\[
\theta_{x}
  = \frac{\displaystyle\sum_{y\in \mathcal{Y}} W_{Y|X}^{1-s}(y|x)\,q_{Y}^{s}(y)
          \log\!\Bigl(\frac{q_Y(y)}{W_{Y|X}(y|x)}\Bigr)}
         {\displaystyle\sum_{y\in \mathcal{Y}} W_{Y|X}^{1-s}(y|x)\,q_{Y}^{s}(y)},
\]
and $v_n$'s do not depend on $s$.
\end{theorem}

\begin{proof}
The proof can be found in Appendix~\ref{th:var:mean:and:e=infty}.
\end{proof}

\begin{remark}
Theorem~\ref{mean:var:th1} implies
\[
\lim_{n\to\infty}
\frac{\sqrt{\operatorname{Var}_{Q_{Y^n}^{*}}(\Theta)}}
     {\mathbb{E}_{Q_{Y^n}^{*}}(\Theta)} = 0.
\]
Hence, for large $n$,
\[
\mathbb{E}_{Q_{Y^n}^{*}}(\Theta)
 \pm \sqrt{\operatorname{Var}_{Q_{Y^n}^{*}}(\Theta)}
 \simeq \mathbb{E}_{Q_{Y^n}^{*}}(\Theta),
\]
i.e., the relative fluctuations of $\Theta$ under $Q_{Y^n}^{*}$ are negligible compared to its mean.
\end{remark}
Now, by substituting the variance upper bound $\operatorname{Var}_{Q_{Y^n}^{*}}(\Theta)\le \tfrac{v_n}{s^2}$ provided by Theorem~\ref{mean:var:th1} into \eqref{expo:4:1} and \eqref{expo:4:2}, we obtain
\begin{align}
&\text{either}\quad
   P_{e}^{\max}
   >
     \exp\Bigl(
         n \sum_{x} p_X(x) \mu_x
         - s\Bigl(\mathbb{E}_{Q_{Y^n}^{*}}[\Theta]
               + \frac{1}{s}\sqrt{2\,v_n}\Bigr)
         - \log b - \log 4
     \Bigr),
  \label{expo:6:1}\\
&\text{or}\quad
   e^{-nR}
   >
     \exp\Bigl(
         n \sum_{x} p_X(x) \mu_x
         + (1-s)\Bigl(\mathbb{E}_{Q_{Y^n}^{*}}[\Theta]
               - \frac{1}{s}\sqrt{2\,v_n}\Bigr)
         - \log b - \log 4
     \Bigr).
   \label{expo:6:2}
\end{align}
Now define
\begin{align}
\mathtt{R}_{n}^{*}(s,p_X,q_Y)
 &\triangleq
 -\frac{
         n \sum_{x}p_X(x) \mu_x
         + (1-s)\Bigl(\mathbb{E}_{Q_{Y^n}^{*}}[\Theta]
               - \frac{1}{s}\sqrt{2\,v_n}\Bigr)
         - \log b - \log 4
       }{n}.
 \label{RR}
\end{align}
Solving \eqref{RR} for $\mathbb{E}_{Q_{Y^n}^{*}}[\Theta]$ and substituting into \eqref{expo:6:1}, we obtain after simplification
\begin{align}
&\text{either}\quad
   P_{e}^{\max}
   >
     \exp\left\{{n\Bigl(
      \frac{\sum_{x} p_X(x) \mu_x}{1-s}
          + \frac{s}{1-s}
            \Bigl(\mathtt{R}_{n}^{*}(s,p_X,q_Y)
                  - \frac{\log 4 + \log b}{n}\Bigr)
       - \frac{\sqrt{8\,v_n}+\log b + \log 4}{n}
     \Bigr)}\right\},
  \label{expo:5:1}\\
&\text{or}\quad
   R < \mathtt{R}_{n}^{*}(s,p_X,q_Y).
   \label{expo:5:2}
\end{align}
Clearly
\begin{align}
  \frac{\log 4  +\log b}{n}=\mathcal{O}\left(\frac{1}{n}\right)  ,~~\frac{\sqrt{8\,v_n}+\log b + \log 4}{n}=\mathcal{O}\left(\frac{1}{\sqrt{n}}\right).  
\end{align}

Note that, by \eqref{RR}, we have
\[
\lim_{s\to 0^{+}} \mathtt{R}_{n}^{*}(s,p_X,q_Y) = +\infty.
\]
Furthermore $\mathtt{R}_{n}^{*}(s,p_X,q_Y)$ is a continuous function of $s$. 
Hence, the equation $\mathtt{R}_{n}^{*}(s,p_X,q_Y)=R$ admits a solution in $s\in(0,1)$ for some range of rates $R$.

We are given a fixed code rate $R$. Two cases are possible: 

\textbf{Case 1:} for the given rate $R$, there is $\hat{s}\in(0,1)$ such that
\[
\mathtt{R}_{n}^{*}(\hat{s},p_X,q_Y) = R.
\]
In this case, the inequality \eqref{expo:5:2} cannot hold, so \eqref{expo:5:1} must hold. Substituting $\mathtt{R}_{n}^{*}(\hat{s},p_X,q_Y)=R$ into \eqref{expo:5:1} gives
\begin{align}
   P_{e}^{\max}
   >
   \exp\Biggl\{
     n\Biggl(
      \frac{\sum_{x} p_X(x)  \mu_x}{1-\hat{s}}
      + \frac{\hat{s}}{1-\hat{s}}
        \Bigl(R - \frac{\log 4 + \log b}{n}\Bigr)
      -\frac{\sqrt{8\,v_n} +\log b+ \log 4}{n}
     \Biggr)
   \Biggr\}.
  \label{expo:8:1}
\end{align}
Recall that $\mu_k$ depends on $s$ and $q_Y$, i.e., $\mu_k = \mu_k(s,q_Y)$. Since the actual value of $\hat{s}\in(0,1)$ is unknown, we take the infimum over $s\in(0,1)$ in \eqref{expo:8:1} and obtain
\begin{align}
   P_{e}^{\max}
   >
   \exp\Biggl\{
     n\Biggl(
      \inf_{s\in(0,1)}
      \frac{\sum_{x} p_X(x) \mu_x(s,q_Y)}{1-s}
      + \frac{s}{1-s}
        \Bigl(R - \frac{\log 4 + \log b}{n}\Bigr)
      - \frac{\sqrt{8\,v_n} +\log b+ \log 4}{n}
     \Biggr)
   \Biggr\}.
  \label{expo:9:1}
\end{align}
Since \eqref{expo:9:1} holds for every code composition $p_X$, and $p_X$ may be unknown, we further take the infimum over all $p_X$:
\begin{align}
   P_{e}^{\max}
   >
   \exp\Biggl\{
     n\Biggl(
      \inf_{p_X}\inf_{s\in(0,1)}
      \frac{\sum_{x}p_X(x) \mu_x(s,q_Y)}{1-s}
      + \frac{s}{1-s}
        \Bigl(R - \frac{\log 4 + \log b}{n}\Bigr)
      - \frac{\sqrt{8\,v_n} +\log b+ \log 4}{n}
     \Biggr)
   \Biggr\}.
  \label{expo:10:1}
\end{align}

The auxiliary distribution $q_Y$ on $\mathcal{Y}$ is still arbitrary, and can be chosen to depend on $s$ and $p_X$. We take supremum over $q_Y$ as follows
\begin{align}
   P_{e}^{\max}
   >
   \exp\Biggl\{
     n\Biggl(
     \inf_{p_X}\inf_{s\in(0,1)} \sup_{q_Y}
      \frac{\sum_{x}p_X(x) \mu_x(s,q_Y)}{1-s}
      + \frac{s}{1-s}
        \Bigl(R - \frac{\log 4 + \log b}{n}\Bigr)
      - \frac{\sqrt{8\,v_n} +\log b+ \log 4}{n}
     \Biggr)
   \Biggr\}.
  \label{expo:12:1}
\end{align}

We now define the sphere-packing exponent associated with $f$ as
\begin{align}
  \mathtt{E}_{f-\mathrm{sp}}(R)
  \triangleq
 \sup_{p_X}\sup_{s\in(0,1)} \inf_{q_Y}
  \Biggl[
    -\frac{\sum_{x} p_X(x) \mu_x(s,q_Y)}{1-s}
    - \frac{s}{1-s}\,R
  \Biggr].
  \label{expo:13:1}
\end{align}
With this definition, \eqref{expo:12:1} can be written as
\begin{align}
   P_{e}^{\max}
   >
   \exp\Bigl\{
      -n\Bigl(
         \mathtt{E}_{f-\mathrm{sp}}\bigl(R-\frac{\log 4 + \log b}{n}\bigr)
         + \frac{\sqrt{8\,v_n} + \log b+\log 4}{n}
      \Bigr)
   \Bigr\}.
  \label{expo:14:1}
\end{align}

In summary, we have shown that for any rate $R$ such that the equation $\mathtt{R}_{n}^{*}(s,p_X,q_Y)=R$ admits a solution $s\in(0,1)$, the lower bound \eqref{expo:14:1} on $P_{e}^{\max}$ holds.

\textbf{Case 2:} Suppose instead that, for a given $R$, the equation $\mathtt{R}_{n}^{*}(s,p_X,q_Y)=R$ has no solution in $s\in(0,1)$. Since $\lim_{s\to 0^{+}} \mathtt{R}_{n}^{*}(s,p_X,q_Y)=+\infty$, this implies
\[
R < \mathtt{R}_{n}^{*}(s,p_X,q_Y) \quad \text{for all } s\in(0,1).
\]
The next lemma (proved in Appendix~\ref{th:var:mean:and:e=infty}) shows that in this case the sphere-packing exponent diverges:

\begin{lemma}\label{range:2}
For a given rate $R$, if the equation $\mathtt{R}_{n}^{*}(s,p_X,q_Y)=R$ has no solution in $s\in(0,1)$, then
\[
\mathtt{E}_{f-\mathrm{sp}}\bigl(R-\frac{\log 4 + \log b}{n}\bigr) = +\infty.
\]
\end{lemma}

As a consequence, in this second regime, the bound \eqref{expo:14:1} becomes trivial (right-hand side tends to zero). Combining both cases, we conclude that \eqref{expo:14:1} is valid for all positive rates.
\end{proof}


\subsection{Proof of Theorem~\ref{mean:var:th1} and Lemma~\ref{range:2}}
We first prove Theorem~\ref{mean:var:th1}.\label{th:var:mean:and:e=infty}

\begin{proof}[Proof of Theorem~\ref{mean:var:th1}]
Recall the definition of $Q_{Y^n}^{*}$:
\begin{align}
   Q_{Y^n}^{*}(y^n)
   = \frac{P_1(y^n)\,
           f_s\!\Bigl(\frac{P_2(y^n)}{P_1(y^n)}\Bigr)}
          {\displaystyle\sum_{\tilde y^n\in \mathcal{Y}^n} P_1(\tilde y^n)\,
           f_s\!\Bigl(\frac{P_2(\tilde y^n)}{P_1(\tilde y^n)}\Bigr)},
   \label{q:1}
\end{align}
where $P_{1}(y^n)=\prod_{i=1}^{n}W_{Y|X}(y_i|x_{i}^{*})$ and
$P_{2}(y^n)=\prod_{i=1}^{n}q_Y(y_i)$.

Define
\begin{align}
   J_{Y^n}^{*}(y^n)
   &\triangleq
   \frac{P_{1}^{\,1-s}(y^n) P_{2}^{\,s}(y^n)}
        {\displaystyle\sum_{\tilde y^n} P_{1}^{\,1-s}(\tilde y^n) P_{2}^{\,s}(\tilde y^n)}
    = \prod_{i=1}^{n}
      \frac{W_{Y|X}^{\,1-s}(y_i|x_{i}^{*}) q_{Y}^{s}(y_i)}
           {\displaystyle\sum_{\tilde y\in\mathcal{Y}}
                 W_{Y|X}^{\,1-s}(\tilde y|x_{i}^{*}) q_{Y}^{s}(\tilde y)}.
\end{align}
For notational simplicity, we write $Q^{*}$ and $J^{*}$ instead of
$Q_{Y^n}^{*}$ and $J_{Y^n}^{*}$, respectively. Observe that both $Q^{*}$ and $J^{*}$ are probability distributions on $\mathcal{Y}^n$.

Let
\[
   Z_{i}^{*}(y_i)
   \triangleq
   \log\!\Bigl(\frac{q_Y(y_i)}{W_{Y|X}(y_i|x_{i}^{*})}\Bigr),
   \qquad
   \Theta(y^n) = \sum_{i=1}^{n} Z_{i}^{*}(y_i),
   \qquad
   \mathtt{T}_n(y^n) \triangleq \frac{\Theta(y^n)}{n}.
\]
With $f_s(x)=x^{s}\psi_{s}(\log x)$, we can rewrite \eqref{q:1} as
\begin{align}
 Q^{*}(y^n)
 &= \frac{P_1(y^n)\bigl(\tfrac{P_2(y^n)}{P_1(y^n)}\bigr)^{s}
                   \psi_{s}\!\Bigl(\log \tfrac{P_2(y^n)}{P_1(y^n)}\Bigr)}
          {\displaystyle\sum_{z^n} P_1(z^n)
                 \bigl(\tfrac{P_2(z^n)}{P_1(z^n)}\bigr)^{s}
                 \psi_{s}\!\Bigl(\log \tfrac{P_2(z^n)}{P_1(z^n)}\Bigr)} \\
 &= \frac{J^{*}(y^n)\,\psi_{s}\bigl(n\mathtt{T}_n(y^n)\bigr)}
          {\mathbb{E}_{Y^n\sim J^{*}}\!\bigl[\psi_{s}(n\mathtt{T}_n(Y^n))\bigr]}.
\end{align}
Therefore, for any function $\Delta(y^n)$,
\begin{align}
\mathbb{E}_{Y^n\sim Q^{*}}[\Delta(Y^n)]
  = \frac{\mathbb{E}_{Y^n\sim J^{*}}\!\bigl[\Delta(Y^n)\,\psi_{s}(n\mathtt{T}_n(Y^n))\bigr]}
         {\mathbb{E}_{Y^n\sim J^{*}}\!\bigl[\psi_{s}(n\mathtt{T}_n(Y^n))\bigr]}.
\end{align}
We write $\mathtt{T}_n$ as a random variable to denote $\mathtt{T}_n(Y^n)$.
Since $ \mathtt{T}_n(y^n)=(1/n)\sum_{i=1}^{n} Z_{i}^{*}(y_i)$, $J^{*}$ is a product measure, and $(x^*_1, x^*_2, \cdots, x^*_n)$ has the type $p_X(x)$, it is immediate that
\begin{align}
  \mathbb{E}_{J^{*}}[\mathtt{T}_n]
  = \omega\triangleq \sum_{x}p_X(x) \theta_x,
\end{align}
where
\[
\theta_{x}
  = \frac{\displaystyle\sum_{y\in \mathcal{Y}} W^{1-s}(y|x)\,q_{Y}^{s}(y)
              \log\!\Bigl(\frac{q_{Y}(y)}{W_{Y|X}(y|x)}\Bigr)}
         {\displaystyle\sum_{\tilde y\in\mathcal{Y}} W_{Y|X}^{1-s}(\tilde y|x)\,q_{Y}^{s}(\tilde y)}.
\]

Next, note that $Z_{i}^{*}\in[\log q_{Y,\min}, -\log W_{\min}]$ for all $i$,
hence
\[
\operatorname{Var}(Z_i^{*})
  \le \frac{\bigl(\log(W_{\min}q_{Y,\min})\bigr)^{2}}{4}
  \triangleq \mathtt{B},
\]
under any distribution. Thus
\begin{align}
  \mathbb{E}_{J^{*}}\bigl|\mathtt{T}_n - \omega\bigr|
  &\overset{(a)}{\le} \sqrt{\operatorname{Var}_{J^{*}}(\mathtt{T}_n)}
   = \sqrt{\frac{\sum_{x}p_X(x) \operatorname{Var}_{J^{*}}(Z_x^{*})}{n}}
   \le \frac{\sqrt{\mathtt{B}}}{\sqrt{n}},
  \label{BB}
\end{align}
where $(a)$ is the $\mathcal{L}^1$–$\mathcal{L}^2$ inequality
$\mathbb{E}|\Delta|\le \sqrt{\mathbb{E}[\Delta^{2}]}$.

We are now ready to prove the first part of Theorem~\ref{mean:var:th1}. We write $\Theta$ to denote the random variable $\Theta(Y^n)$.
We have
\begin{align}
   \mathbb{E}_{Q^{*}}\!\Bigl(\frac{\Theta}{n}\Bigr)
   &= \mathbb{E}_{Q^{*}}[\mathtt{T}_n]
    = \omega + \mathbb{E}_{Q^{*}}[\mathtt{T}_n - \omega] \\
   &= \omega
      + \frac{\mathbb{E}_{J^{*}}\!\bigl[(\mathtt{T}_n - \omega)
                      \psi_{s}(n\mathtt{T}_n)\bigr]}
             {\mathbb{E}_{J^{*}}\!\bigl[\psi_{s}(n\mathtt{T}_n)\bigr]}
    = \omega + \mathtt{I}_n,
\end{align}
where
\[
\mathtt{I}_{n}
  \triangleq
  \frac{\mathbb{E}_{J^{*}}\!\bigl[(\mathtt{T}_n - \omega)
                      \psi_{s}(n\mathtt{T}_n)\bigr]}
       {\mathbb{E}_{J^{*}}\!\bigl[\psi_{s}(n\mathtt{T}_n)\bigr]}.
\]
Using the bounds $0<a \le \psi_s(\cdot)\le b<+\infty$ and \eqref{BB}, we obtain
\begin{align}
    |\mathtt{I}_n|
    &\le \frac{b}{a}\,
          \mathbb{E}_{J^{*}}\bigl|\mathtt{T}_n - \omega\bigr|
     \le \frac{b}{a}\,\sqrt{\frac{\mathtt{B}}{n}}
\end{align}
 Thus
$\mathtt{I}_n = \mathcal{O}(n^{-1/2})$, uniformly in $s$, and hence
\begin{align}
   \mathbb{E}_{Q^{*}}[\Theta]
   = n\omega + \mathcal{O}(\sqrt{n}).
   \label{oo1}
\end{align}
This proves the first statement of Theorem~\ref{mean:var:th1}.

Before proving the second part, we recall the following result
from~\cite{Shannon1967}.

\begin{theorem}[{\cite[p.~35]{Shannon1967}}]
For each $1\le i\le n$, the variance of $Z_i^{*}$ under $J^{*}$ satisfies
\begin{align}
      \operatorname{Var}_{J^{*}}(Z_i^{*})
      \le \frac{\log^{2}\!\bigl(\tfrac{e}{W_{\min}}\bigr)}{s^{2}}.
      \label{var:gallager}
\end{align}
\end{theorem}
Let $\mathtt{C}=\log^{2}\!\bigl(\tfrac{e}{W_{\min}}\bigr)$, we have
\begin{align*}
    \frac{\Var_{Q^{*}}(\Theta)}{n}
    &= \Var_{Q^{*}}\!\left(\frac{\Theta}{\sqrt{n}}\right)
      \overset{(a)}{\leq}
      \E_{Q^{*}}\!\left(\frac{\Theta}{\sqrt{n}}-\sqrt{n}\omega\right)^{2}
      = n\,\E_{Q^{*}}\bigl(\mathtt{T}_n-\omega\bigr)^{2} \\[2mm]
    &=
      n\,
      \frac{\E_{J^{*}}\!\bigl[(\mathtt{T}_n-\omega)^{2}\psi_{s}(n\mathtt{T}_n)\bigr]}
           {\E_{J^{*}}\!\bigl[\psi_{s}(n\mathtt{T}_n)\bigr]}
      \;\overset{(b)}{\leq}\;
      n\,\frac{b}{a}\,\E_{J^{*}}\bigl[(\mathtt{T}_n-\omega)^{2}\bigr] \\[2mm]
    &= n\,\frac{b}{a}\,\Var_{J^{*}}(\mathtt{T}_n)
       = n\,\frac{b}{a}\,\frac{\sum_{x}p_X(x)\Var_{J^{*}}(Z_x)}{n} \\[2mm]
    &\le \frac{b}{a}\,\frac{\mathtt{C}}{s^{2}}.
\end{align*}
Here, (a) uses the fact that for any random variable $\Delta$ and any real
$\delta$,
\[
\Var(\Delta)\le \E(\Delta-\delta)^{2},
\]
and (b) follows from ,
$0<a\le\psi_s\le b<+\infty$. Thus
\[
  \Var_{Q^{*}}(\Theta)\;\le\;n\,\frac{\mathtt{C}\,b}{a\,s^{2}}.
\]
Let $v_n \triangleq n\,\frac{\mathtt{C}\,b}{a} = \mathcal{O}(n)$. Then
\begin{align}
  \Var_{Q^{*}}(\Theta)\;\le\;\frac{v_n}{s^{2}}.
\end{align}
Note that $v_n$ does not depend on $s$. This completes the proof.
\end{proof}
Now we give the proof of Lemma~\ref{range:2}.

\begin{proof}[Proof of Lemma~\ref{range:2}]
For convenience, define
\begin{align}
    \mathtt{V}(s,p_X,q_Y)
    \triangleq
      - \frac{\sum_{x} p_X(x) \mu_x(s,q_Y)}{1-s}
      - \frac{s}{1-s}
        \Bigl(R - \frac{\log 4 + \log b}{n}\Bigr)
      + \frac{\sqrt{8\,v_n} +\log b+ \log 4}{n}.
    \label{vv}
\end{align}
Recall that
\[
\mathtt{E}_{f-\mathrm{sp}}\left(R-\frac{\log 4 + \log b}{n}\right)
 = \sup_{p_X}\sup_{s\in(0,1)} \inf_{q_Y}\mathtt{V}(s,p_X,q_Y).
\]

Assume that, for the given $R$, the inequality
\[
 R < \mathtt{R}_{n}^{*}(s,p_X,q_Y)
\]
holds for all $s\in(0,1)$, where $\mathtt{R}_{n}^{*}(s,p_X,q_Y)$ is defined in \eqref{RR}. Let $R<\mathtt{R}_{n}^{*}(s,p_X,q_Y)-\epsilon$ for some $\epsilon>0$ and  for all $s\in(0,1)$.  Substituting the expression of $\mathtt{R}_{n}^{*}(s,p_X,q_Y)$ from \eqref{RR} into \eqref{vv} and yields
\begin{align}
   \mathtt{V}(s,p_X,q_Y)
   &\geq
   -\sum_{x}p_X(x) \mu_x(s,q_Y)
    + \varsigma_{n}(s)
    +\frac{s\epsilon}{1-s}
      \triangleq \mathtt{U}(s,p_X,q_Y),
\end{align}
where
\begin{align}
   \varsigma_{n}(s)
   =
   \frac{\sqrt{2\,v_n} +\log b+ \log 4}{n}
   + \frac{s}{n}\,\E_{Q_{Y^{n}}^{*}}[\Theta].
\end{align}
Thus
\[
\mathtt{E}_{f-\mathrm{sp}}\left(R-\frac{\log 4 + \log b}{n}\right)
\geq \sup_{p_X}\sup_{s\in(0,1)} \inf_{q_Y}\mathtt{U}(s,p_X,q_Y).
\]

 The term
\(
 -
\sum_{x}p_X(x) \mu_x(s,q_Y) 
+\varsigma_{n}(s)
\)
is uniformly bounded in $(s,p_X,q_Y)$. We get
\[
\frac{s\epsilon}{1-s}\xrightarrow[s\to 1^{-}]{} +\infty.
\]
Hence, for any fixed positive distribution $q_Y$ and any composition $p_X$,
\[
\lim_{s\to 1^{-}}\mathtt{U}(s,p_X,q_Y) = +\infty.
\]
It follows that
\[
\sup_{p_X}\sup_{s\in(0,1)} \inf_{q_Y}\mathtt{U}(s,p_X,q_Y) = +\infty,
\]
and consequently
\[
\mathtt{E}_{f-\mathrm{sp}}\left(R-\frac{\log 4 + \log b}{n}\right) = +\infty.
\]
This proves Lemma~\ref{range:2}.
\end{proof}

\end{document}